\documentclass[aps,prl,twocolumn,superscriptaddress,10pt]{revtex4-2}

\usepackage[colorlinks=true,allcolors=blue]{hyperref}
\usepackage{url}
\usepackage[utf8]{inputenc}
\usepackage{amsmath,amssymb,amsthm}
\usepackage[normalem]{ulem}
\usepackage{graphicx}
\usepackage{diagbox}
\usepackage{booktabs}
\usepackage{longtable}
\usepackage{float}
\usepackage{microtype}
\usepackage{cleveref}
\usepackage{xcolor}
\usepackage{outlines}
\usepackage{tikz}
\usetikzlibrary{decorations.pathreplacing,calc,intersections,quantikz}

\usepackage{silence}
\usepackage[T1]{fontenc}

\newcommand{\dv}[1]{{\fontfamily{DejaVuSans-TLF}\selectfont #1}}

\newtheorem{theorem}{Theorem}
\newtheorem{lemma}[theorem]{Lemma}

\newtheorem{proposition}[theorem]{Proposition}
\theoremstyle{definition}
\newtheorem{definition}[theorem]{Definition}

\newcommand{\qparentfour}{https://algassert.com/quirk\#circuit=\%7B\%22cols\%22\%3A\%5B\%5B\%22H\%22\%2C\%22H\%22\%2C\%22H\%22\%2C\%22H\%22\%5D\%2C\%5B\%22zpar\%22\%2C\%22zpar\%22\%2C\%22zpar\%22\%2C\%22zpar\%22\%2C\%22\%E2\%88\%9Ai\%22\%5D\%2C\%5B\%22zpar\%22\%2C1\%2C\%22zpar\%22\%2C\%22zpar\%22\%2C\%22\%E2\%88\%9Ai\%22\%5D\%2C\%5B\%22zpar\%22\%2C\%22zpar\%22\%2C1\%2C\%22zpar\%22\%2C\%22\%E2\%88\%9Ai\%22\%5D\%2C\%5B\%22zpar\%22\%2C\%22zpar\%22\%2C\%22zpar\%22\%2C1\%2C\%22\%E2\%88\%9Ai\%22\%5D\%2C\%5B\%22zpar\%22\%2C\%22zpar\%22\%2C1\%2C1\%2C\%22\%E2\%88\%9Ai\%22\%5D\%2C\%5B\%22zpar\%22\%2C1\%2C\%22zpar\%22\%2C1\%2C\%22\%E2\%88\%9Ai\%22\%5D\%2C\%5B\%22zpar\%22\%2C1\%2C1\%2C\%22zpar\%22\%2C\%22\%E2\%88\%9Ai\%22\%5D\%2C\%5B\%22~fu1l\%22\%5D\%2C\%5B\%22Bloch\%22\%5D\%2C\%5B1\%2C\%22zpar\%22\%2C\%22zpar\%22\%2C\%22zpar\%22\%2C\%22\%E2\%88\%9Ai\%22\%5D\%2C\%5B1\%2C\%22zpar\%22\%2C\%22zpar\%22\%2C1\%2C\%22\%E2\%88\%9Ai\%22\%5D\%2C\%5B1\%2C\%22zpar\%22\%2C1\%2C\%22zpar\%22\%2C\%22\%E2\%88\%9Ai\%22\%5D\%2C\%5B1\%2C\%22~fu1l\%22\%5D\%2C\%5B1\%2C\%22Bloch\%22\%5D\%2C\%5B1\%2C1\%2C\%22zpar\%22\%2C\%22zpar\%22\%2C\%22\%E2\%88\%9Ai\%22\%5D\%2C\%5B1\%2C1\%2C\%22~fu1l\%22\%5D\%2C\%5B1\%2C1\%2C\%22Bloch\%22\%5D\%2C\%5B1\%2C1\%2C1\%2C\%22~fu1l\%22\%5D\%2C\%5B1\%2C1\%2C1\%2C\%22Bloch\%22\%5D\%2C\%5B\%22H\%22\%2C\%22H\%22\%2C\%22H\%22\%2C\%22H\%22\%5D\%2C\%5B\%22\%7C0\%E2\%9F\%A9\%E2\%9F\%A80\%7C\%22\%2C\%22\%7C0\%E2\%9F\%A9\%E2\%9F\%A80\%7C\%22\%2C\%22\%7C0\%E2\%9F\%A9\%E2\%9F\%A80\%7C\%22\%2C\%22\%7C0\%E2\%9F\%A9\%E2\%9F\%A80\%7C\%22\%5D\%5D\%2C\%22gates\%22\%3A\%5B\%7B\%22id\%22\%3A\%22~dvdi\%22\%2C\%22matrix\%22\%3A\%22\%7B\%7B\%E2\%88\%9A\%C2\%BD-\%E2\%88\%9A\%C2\%BDi\%2C0\%7D\%2C\%7B0\%2C\%E2\%88\%9A\%C2\%BD\%2B\%E2\%88\%9A\%C2\%BDi\%7D\%7D\%22\%7D\%2C\%7B\%22id\%22\%3A\%22~2gmh\%22\%2C\%22name\%22\%3A\%22Z(pi\%2F4)\%22\%2C\%22matrix\%22\%3A\%22\%7B\%7B\%E2\%88\%9A\%C2\%BD-\%E2\%88\%9A\%C2\%BDi\%2C0\%7D\%2C\%7B0\%2C\%E2\%88\%9A\%C2\%BD\%2B\%E2\%88\%9A\%C2\%BDi\%7D\%7D\%22\%7D\%2C\%7B\%22id\%22\%3A\%22~fu1l\%22\%2C\%22name\%22\%3A\%22Z(pi\%2F8)\%22\%2C\%22matrix\%22\%3A\%22\%7B\%7B0.9238795-0.3826834i\%2C0\%7D\%2C\%7B0\%2C0.9238795\%2B0.3826834i\%7D\%7D\%22\%7D\%5D\%7D}
\newcommand{\qaaazthreetwo}{https://algassert.com/quirk\#circuit=\%7B\%22cols\%22\%3A\%5B\%5B1\%2C1\%2C\%22H\%22\%2C\%22H\%22\%2C\%22H\%22\%2C\%22H\%22\%5D\%2C\%5B1\%2C1\%2C\%22zpar\%22\%2C1\%2C\%22zpar\%22\%2C\%22zpar\%22\%2C\%22\%E2\%88\%9Ai\%22\%5D\%2C\%5B1\%2C1\%2C\%22zpar\%22\%2C\%22zpar\%22\%2C1\%2C\%22zpar\%22\%2C\%22\%E2\%88\%9Ai\%22\%5D\%2C\%5B1\%2C1\%2C\%22zpar\%22\%2C1\%2C1\%2C\%22zpar\%22\%2C\%22\%E2\%88\%9Ai\%22\%5D\%2C\%5B1\%2C1\%2C\%22zpar\%22\%2C\%22zpar\%22\%2C\%22zpar\%22\%2C\%22zpar\%22\%2C\%22\%E2\%88\%9Ai\%22\%5D\%2C\%5B1\%2C1\%2C\%22zpar\%22\%2C\%22zpar\%22\%2C\%22zpar\%22\%2C1\%2C\%22\%E2\%88\%9Ai\%22\%5D\%2C\%5B1\%2C1\%2C\%22zpar\%22\%2C\%22zpar\%22\%2C1\%2C1\%2C\%22\%E2\%88\%9Ai\%22\%5D\%2C\%5B1\%2C1\%2C\%22zpar\%22\%2C1\%2C\%22zpar\%22\%2C1\%2C\%22\%E2\%88\%9Ai\%22\%5D\%2C\%5B1\%2C1\%2C\%22~d9b7\%22\%5D\%2C\%5B1\%2C1\%2C\%22H\%22\%5D\%2C\%5B1\%2C1\%2C\%22Chance\%22\%5D\%2C\%5B1\%2C1\%2C\%22\%7C0\%E2\%9F\%A9\%E2\%9F\%A80\%7C\%22\%5D\%2C\%5B1\%2C1\%2C1\%2C\%22X\%5E\%C2\%BD\%22\%5D\%2C\%5B1\%2C1\%2C1\%2C\%22\%E2\%80\%A2\%22\%2C\%22X\%22\%2C\%22X\%22\%5D\%2C\%5B1\%2C1\%2C1\%2C\%22\%E2\%8A\%96\%22\%2C\%22Z\%22\%2C\%22Z\%22\%5D\%2C\%5B1\%2C1\%2C1\%2C\%22Z\%5E-\%C2\%BC\%22\%5D\%2C\%5B1\%2C1\%2C1\%2C\%22\%E2\%8A\%96\%22\%2C\%22Z\%22\%2C\%22Z\%22\%5D\%5D\%2C\%22gates\%22\%3A\%5B\%7B\%22id\%22\%3A\%22~dvdi\%22\%2C\%22matrix\%22\%3A\%22\%7B\%7B\%E2\%88\%9A\%C2\%BD-\%E2\%88\%9A\%C2\%BDi\%2C0\%7D\%2C\%7B0\%2C\%E2\%88\%9A\%C2\%BD\%2B\%E2\%88\%9A\%C2\%BDi\%7D\%7D\%22\%7D\%2C\%7B\%22id\%22\%3A\%22~2gmh\%22\%2C\%22name\%22\%3A\%22Z(pi\%2F4)\%22\%2C\%22matrix\%22\%3A\%22\%7B\%7B\%E2\%88\%9A\%C2\%BD-\%E2\%88\%9A\%C2\%BDi\%2C0\%7D\%2C\%7B0\%2C\%E2\%88\%9A\%C2\%BD\%2B\%E2\%88\%9A\%C2\%BDi\%7D\%7D\%22\%7D\%2C\%7B\%22id\%22\%3A\%22~ntmg\%22\%2C\%22name\%22\%3A\%22Z(pi\%2F4)\%22\%2C\%22matrix\%22\%3A\%22\%7B\%7B0.9238795-0.3826834i\%2C0\%7D\%2C\%7B0\%2C0.9238795\%2B0.3826834i\%7D\%7D\%22\%7D\%2C\%7B\%22id\%22\%3A\%22~d9b7\%22\%2C\%22name\%22\%3A\%22Z(pi\%2F8)\%22\%2C\%22matrix\%22\%3A\%22\%7B\%7B0.9238795-0.3826834i\%2C0\%7D\%2C\%7B0\%2C0.9238795\%2B0.3826834i\%7D\%7D\%22\%7D\%2C\%7B\%22id\%22\%3A\%22~703g\%22\%2C\%22name\%22\%3A\%22Z(-\%5C\%5Cpi\%2F8)\%22\%2C\%22matrix\%22\%3A\%22\%7B\%7B0.9238795\%2B0.3826834i\%2C0\%7D\%2C\%7B0\%2C0.9238795-0.3826834i\%7D\%7D\%22\%7D\%2C\%7B\%22id\%22\%3A\%22~i2pe\%22\%2C\%22name\%22\%3A\%22Z(-pi\%2F4)\%22\%2C\%22matrix\%22\%3A\%22\%7B\%7B\%E2\%88\%9A\%C2\%BD\%2B\%E2\%88\%9A\%C2\%BDi\%2C0\%7D\%2C\%7B0\%2C\%E2\%88\%9A\%C2\%BD-\%E2\%88\%9A\%C2\%BDi\%7D\%7D\%22\%7D\%5D\%7D}

\newcommand{\qtwelvetwotwo}{https://algassert.com/quirk\#circuit=\%7B\%22cols\%22\%3A\%5B\%5B1\%2C1\%2C\%22H\%22\%2C\%22H\%22\%2C\%22H\%22\%2C\%22H\%22\%5D\%2C\%5B1\%2C1\%2C\%22zpar\%22\%2C1\%2C\%22zpar\%22\%2C\%22zpar\%22\%2C\%22\%E2\%88\%9Ai\%22\%5D\%2C\%5B1\%2C1\%2C\%22zpar\%22\%2C\%22zpar\%22\%2C1\%2C\%22zpar\%22\%2C\%22\%E2\%88\%9Ai\%22\%5D\%2C\%5B1\%2C1\%2C\%22zpar\%22\%2C1\%2C1\%2C\%22zpar\%22\%2C\%22\%E2\%88\%9Ai\%22\%5D\%2C\%5B1\%2C1\%2C\%22zpar\%22\%2C\%22zpar\%22\%2C\%22zpar\%22\%2C\%22zpar\%22\%2C\%22\%E2\%88\%9Ai\%22\%5D\%2C\%5B1\%2C1\%2C\%22zpar\%22\%2C\%22zpar\%22\%2C\%22zpar\%22\%2C1\%2C\%22\%E2\%88\%9Ai\%22\%5D\%2C\%5B1\%2C1\%2C\%22zpar\%22\%2C\%22zpar\%22\%2C1\%2C1\%2C\%22\%E2\%88\%9Ai\%22\%5D\%2C\%5B1\%2C1\%2C\%22zpar\%22\%2C1\%2C\%22zpar\%22\%2C1\%2C\%22\%E2\%88\%9Ai\%22\%5D\%2C\%5B1\%2C1\%2C1\%2C\%22zpar\%22\%2C\%22zpar\%22\%2C1\%2C\%22\%E2\%88\%9Ai\%22\%5D\%2C\%5B1\%2C1\%2C1\%2C\%22zpar\%22\%2C\%22zpar\%22\%2C\%22zpar\%22\%2C\%22\%E2\%88\%9Ai\%22\%5D\%2C\%5B1\%2C1\%2C1\%2C\%22zpar\%22\%2C1\%2C\%22zpar\%22\%2C\%22\%E2\%88\%9Ai\%22\%5D\%2C\%5B1\%2C1\%2C\%22~d9b7\%22\%2C\%22~d9b7\%22\%5D\%2C\%5B1\%2C1\%2C\%22H\%22\%2C\%22H\%22\%5D\%2C\%5B1\%2C1\%2C\%22\%7C0\%E2\%9F\%A9\%E2\%9F\%A80\%7C\%22\%2C\%22\%7C0\%E2\%9F\%A9\%E2\%9F\%A80\%7C\%22\%5D\%2C\%5B1\%2C1\%2C1\%2C1\%2C\%22zpar\%22\%2C\%22zpar\%22\%2C\%22\%E2\%88\%9Ai\%22\%5D\%5D\%2C\%22gates\%22\%3A\%5B\%7B\%22id\%22\%3A\%22~dvdi\%22\%2C\%22matrix\%22\%3A\%22\%7B\%7B\%E2\%88\%9A\%C2\%BD-\%E2\%88\%9A\%C2\%BDi\%2C0\%7D\%2C\%7B0\%2C\%E2\%88\%9A\%C2\%BD\%2B\%E2\%88\%9A\%C2\%BDi\%7D\%7D\%22\%7D\%2C\%7B\%22id\%22\%3A\%22~2gmh\%22\%2C\%22name\%22\%3A\%22Z(pi\%2F4)\%22\%2C\%22matrix\%22\%3A\%22\%7B\%7B\%E2\%88\%9A\%C2\%BD-\%E2\%88\%9A\%C2\%BDi\%2C0\%7D\%2C\%7B0\%2C\%E2\%88\%9A\%C2\%BD\%2B\%E2\%88\%9A\%C2\%BDi\%7D\%7D\%22\%7D\%2C\%7B\%22id\%22\%3A\%22~ntmg\%22\%2C\%22name\%22\%3A\%22Z(pi\%2F4)\%22\%2C\%22matrix\%22\%3A\%22\%7B\%7B0.9238795-0.3826834i\%2C0\%7D\%2C\%7B0\%2C0.9238795\%2B0.3826834i\%7D\%7D\%22\%7D\%2C\%7B\%22id\%22\%3A\%22~d9b7\%22\%2C\%22name\%22\%3A\%22Z(pi\%2F8)\%22\%2C\%22matrix\%22\%3A\%22\%7B\%7B0.9238795-0.3826834i\%2C0\%7D\%2C\%7B0\%2C0.9238795\%2B0.3826834i\%7D\%7D\%22\%7D\%2C\%7B\%22id\%22\%3A\%22~703g\%22\%2C\%22name\%22\%3A\%22Z(-\%5C\%5Cpi\%2F8)\%22\%2C\%22matrix\%22\%3A\%22\%7B\%7B0.9238795\%2B0.3826834i\%2C0\%7D\%2C\%7B0\%2C0.9238795-0.3826834i\%7D\%7D\%22\%7D\%2C\%7B\%22id\%22\%3A\%22~i2pe\%22\%2C\%22name\%22\%3A\%22Z(-pi\%2F4)\%22\%2C\%22matrix\%22\%3A\%22\%7B\%7B\%E2\%88\%9A\%C2\%BD\%2B\%E2\%88\%9A\%C2\%BDi\%2C0\%7D\%2C\%7B0\%2C\%E2\%88\%9A\%C2\%BD-\%E2\%88\%9A\%C2\%BDi\%7D\%7D\%22\%7D\%5D\%7D}

\newcommand{\qelevenoneone}{https://algassert.com/quirk\#circuit=\%7B\%22cols\%22\%3A\%5B\%5B1\%2C1\%2C\%22H\%22\%2C\%22H\%22\%2C\%22H\%22\%2C\%22H\%22\%5D\%2C\%5B1\%2C1\%2C\%22zpar\%22\%2C1\%2C\%22zpar\%22\%2C\%22zpar\%22\%2C\%22\%E2\%88\%9Ai\%22\%5D\%2C\%5B1\%2C1\%2C\%22zpar\%22\%2C\%22zpar\%22\%2C1\%2C\%22zpar\%22\%2C\%22\%E2\%88\%9Ai\%22\%5D\%2C\%5B1\%2C1\%2C\%22zpar\%22\%2C1\%2C1\%2C\%22zpar\%22\%2C\%22\%E2\%88\%9Ai\%22\%5D\%2C\%5B1\%2C1\%2C\%22zpar\%22\%2C\%22zpar\%22\%2C\%22zpar\%22\%2C\%22zpar\%22\%2C\%22\%E2\%88\%9Ai\%22\%5D\%2C\%5B1\%2C1\%2C\%22zpar\%22\%2C\%22zpar\%22\%2C\%22zpar\%22\%2C1\%2C\%22\%E2\%88\%9Ai\%22\%5D\%2C\%5B1\%2C1\%2C\%22zpar\%22\%2C\%22zpar\%22\%2C1\%2C1\%2C\%22\%E2\%88\%9Ai\%22\%5D\%2C\%5B1\%2C1\%2C\%22zpar\%22\%2C1\%2C\%22zpar\%22\%2C1\%2C\%22\%E2\%88\%9Ai\%22\%5D\%2C\%5B1\%2C1\%2C1\%2C\%22X\%5E\%C2\%BD\%22\%5D\%2C\%5B1\%2C1\%2C1\%2C\%22\%E2\%80\%A2\%22\%2C\%22X\%22\%2C\%22X\%22\%5D\%2C\%5B1\%2C1\%2C1\%2C\%22\%E2\%8A\%96\%22\%2C\%22Z\%22\%2C\%22Z\%22\%5D\%2C\%5B1\%2C1\%2C1\%2C\%22zpar\%22\%2C1\%2C1\%2C\%22\%E2\%88\%9A-i\%22\%5D\%2C\%5B1\%2C1\%2C1\%2C\%22\%E2\%8A\%96\%22\%2C\%22Z\%22\%2C\%22Z\%22\%5D\%2C\%5B1\%2C1\%2C1\%2C\%22zpar\%22\%2C1\%2C1\%2C\%22\%E2\%88\%9Ai\%22\%5D\%2C\%5B1\%2C1\%2C1\%2C1\%2C\%22zpar\%22\%2C1\%2C\%22\%E2\%88\%9Ai\%22\%5D\%2C\%5B1\%2C1\%2C1\%2C1\%2C1\%2C\%22zpar\%22\%2C\%22\%E2\%88\%9Ai\%22\%5D\%2C\%5B1\%2C1\%2C1\%2C\%22Z\%5E\%C2\%BD\%22\%2C\%22Z\%5E\%C2\%BD\%22\%2C\%22Z\%5E\%C2\%BD\%22\%5D\%2C\%5B1\%2C1\%2C1\%2C\%22H\%22\%2C\%22H\%22\%2C\%22H\%22\%5D\%2C\%5B1\%2C1\%2C1\%2C\%22X\%22\%2C\%22X\%22\%2C\%22X\%22\%5D\%2C\%5B1\%2C1\%2C1\%2C\%22\%7C0\%E2\%9F\%A9\%E2\%9F\%A80\%7C\%22\%2C\%22\%7C0\%E2\%9F\%A9\%E2\%9F\%A80\%7C\%22\%2C\%22\%7C0\%E2\%9F\%A9\%E2\%9F\%A80\%7C\%22\%5D\%5D\%2C\%22gates\%22\%3A\%5B\%7B\%22id\%22\%3A\%22~dvdi\%22\%2C\%22matrix\%22\%3A\%22\%7B\%7B\%E2\%88\%9A\%C2\%BD-\%E2\%88\%9A\%C2\%BDi\%2C0\%7D\%2C\%7B0\%2C\%E2\%88\%9A\%C2\%BD\%2B\%E2\%88\%9A\%C2\%BDi\%7D\%7D\%22\%7D\%2C\%7B\%22id\%22\%3A\%22~2gmh\%22\%2C\%22name\%22\%3A\%22Z(pi\%2F4)\%22\%2C\%22matrix\%22\%3A\%22\%7B\%7B\%E2\%88\%9A\%C2\%BD-\%E2\%88\%9A\%C2\%BDi\%2C0\%7D\%2C\%7B0\%2C\%E2\%88\%9A\%C2\%BD\%2B\%E2\%88\%9A\%C2\%BDi\%7D\%7D\%22\%7D\%2C\%7B\%22id\%22\%3A\%22~ntmg\%22\%2C\%22name\%22\%3A\%22Z(pi\%2F4)\%22\%2C\%22matrix\%22\%3A\%22\%7B\%7B0.9238795-0.3826834i\%2C0\%7D\%2C\%7B0\%2C0.9238795\%2B0.3826834i\%7D\%7D\%22\%7D\%2C\%7B\%22id\%22\%3A\%22~d9b7\%22\%2C\%22name\%22\%3A\%22Z(pi\%2F8)\%22\%2C\%22matrix\%22\%3A\%22\%7B\%7B0.9238795-0.3826834i\%2C0\%7D\%2C\%7B0\%2C0.9238795\%2B0.3826834i\%7D\%7D\%22\%7D\%2C\%7B\%22id\%22\%3A\%22~703g\%22\%2C\%22name\%22\%3A\%22Z(-\%5C\%5Cpi\%2F8)\%22\%2C\%22matrix\%22\%3A\%22\%7B\%7B0.9238795\%2B0.3826834i\%2C0\%7D\%2C\%7B0\%2C0.9238795-0.3826834i\%7D\%7D\%22\%7D\%2C\%7B\%22id\%22\%3A\%22~i2pe\%22\%2C\%22name\%22\%3A\%22Z(-pi\%2F4)\%22\%2C\%22matrix\%22\%3A\%22\%7B\%7B\%E2\%88\%9A\%C2\%BD\%2B\%E2\%88\%9A\%C2\%BDi\%2C0\%7D\%2C\%7B0\%2C\%E2\%88\%9A\%C2\%BD-\%E2\%88\%9A\%C2\%BDi\%7D\%7D\%22\%7D\%5D\%7D}
\newcommand{\qtentwotwo}{https://algassert.com/quirk\#circuit=\%7B\%22cols\%22\%3A\%5B\%5B1\%2C1\%2C\%22H\%22\%2C\%22H\%22\%2C\%22H\%22\%2C\%22H\%22\%5D\%2C\%5B1\%2C1\%2C\%22zpar\%22\%2C1\%2C\%22zpar\%22\%2C\%22zpar\%22\%2C\%22\%E2\%88\%9Ai\%22\%5D\%2C\%5B1\%2C1\%2C\%22zpar\%22\%2C\%22zpar\%22\%2C1\%2C\%22zpar\%22\%2C\%22\%E2\%88\%9Ai\%22\%5D\%2C\%5B1\%2C1\%2C\%22zpar\%22\%2C1\%2C1\%2C\%22zpar\%22\%2C\%22\%E2\%88\%9Ai\%22\%5D\%2C\%5B1\%2C1\%2C\%22zpar\%22\%2C\%22zpar\%22\%2C\%22zpar\%22\%2C\%22zpar\%22\%2C\%22\%E2\%88\%9Ai\%22\%5D\%2C\%5B1\%2C1\%2C\%22zpar\%22\%2C\%22zpar\%22\%2C\%22zpar\%22\%2C1\%2C\%22\%E2\%88\%9Ai\%22\%5D\%2C\%5B1\%2C1\%2C\%22zpar\%22\%2C\%22zpar\%22\%2C1\%2C1\%2C\%22\%E2\%88\%9Ai\%22\%5D\%2C\%5B1\%2C1\%2C\%22zpar\%22\%2C1\%2C\%22zpar\%22\%2C1\%2C\%22\%E2\%88\%9Ai\%22\%5D\%2C\%5B1\%2C1\%2C\%22~d9b7\%22\%2C\%22X\%5E\%C2\%BD\%22\%5D\%2C\%5B1\%2C1\%2C\%22H\%22\%5D\%2C\%5B1\%2C1\%2C\%22\%7C0\%E2\%9F\%A9\%E2\%9F\%A80\%7C\%22\%5D\%2C\%5B1\%2C1\%2C1\%2C\%22\%E2\%80\%A2\%22\%2C\%22X\%22\%2C\%22X\%22\%5D\%2C\%5B1\%2C1\%2C1\%2C\%22\%E2\%8A\%96\%22\%2C\%22Z\%22\%2C\%22Z\%22\%5D\%2C\%5B1\%2C1\%2C1\%2C\%22Z\%5E-\%C2\%BC\%22\%5D\%2C\%5B1\%2C1\%2C1\%2C\%22\%E2\%8A\%96\%22\%2C\%22Z\%22\%2C\%22Z\%22\%5D\%2C\%5B1\%2C1\%2C1\%2C\%22~d9b7\%22\%5D\%2C\%5B1\%2C1\%2C1\%2C\%22Z\%5E\%C2\%BD\%22\%5D\%2C\%5B1\%2C1\%2C1\%2C\%22H\%22\%5D\%2C\%5B1\%2C1\%2C1\%2C\%22X\%22\%5D\%2C\%5B1\%2C1\%2C1\%2C\%22\%7C0\%E2\%9F\%A9\%E2\%9F\%A80\%7C\%22\%5D\%5D\%2C\%22gates\%22\%3A\%5B\%7B\%22id\%22\%3A\%22~dvdi\%22\%2C\%22matrix\%22\%3A\%22\%7B\%7B\%E2\%88\%9A\%C2\%BD-\%E2\%88\%9A\%C2\%BDi\%2C0\%7D\%2C\%7B0\%2C\%E2\%88\%9A\%C2\%BD\%2B\%E2\%88\%9A\%C2\%BDi\%7D\%7D\%22\%7D\%2C\%7B\%22id\%22\%3A\%22~2gmh\%22\%2C\%22name\%22\%3A\%22Z(pi\%2F4)\%22\%2C\%22matrix\%22\%3A\%22\%7B\%7B\%E2\%88\%9A\%C2\%BD-\%E2\%88\%9A\%C2\%BDi\%2C0\%7D\%2C\%7B0\%2C\%E2\%88\%9A\%C2\%BD\%2B\%E2\%88\%9A\%C2\%BDi\%7D\%7D\%22\%7D\%2C\%7B\%22id\%22\%3A\%22~ntmg\%22\%2C\%22name\%22\%3A\%22Z(pi\%2F4)\%22\%2C\%22matrix\%22\%3A\%22\%7B\%7B0.9238795-0.3826834i\%2C0\%7D\%2C\%7B0\%2C0.9238795\%2B0.3826834i\%7D\%7D\%22\%7D\%2C\%7B\%22id\%22\%3A\%22~d9b7\%22\%2C\%22name\%22\%3A\%22Z(pi\%2F8)\%22\%2C\%22matrix\%22\%3A\%22\%7B\%7B0.9238795-0.3826834i\%2C0\%7D\%2C\%7B0\%2C0.9238795\%2B0.3826834i\%7D\%7D\%22\%7D\%2C\%7B\%22id\%22\%3A\%22~703g\%22\%2C\%22name\%22\%3A\%22Z(-\%5C\%5Cpi\%2F8)\%22\%2C\%22matrix\%22\%3A\%22\%7B\%7B0.9238795\%2B0.3826834i\%2C0\%7D\%2C\%7B0\%2C0.9238795-0.3826834i\%7D\%7D\%22\%7D\%2C\%7B\%22id\%22\%3A\%22~i2pe\%22\%2C\%22name\%22\%3A\%22Z(-pi\%2F4)\%22\%2C\%22matrix\%22\%3A\%22\%7B\%7B\%E2\%88\%9A\%C2\%BD\%2B\%E2\%88\%9A\%C2\%BDi\%2C0\%7D\%2C\%7B0\%2C\%E2\%88\%9A\%C2\%BD-\%E2\%88\%9A\%C2\%BDi\%7D\%7D\%22\%7D\%5D\%7D}

\newcommand{\qparentfive}{https://algassert.com/quirk\#circuit=\%7B\%22cols\%22\%3A\%5B\%5B\%22H\%22\%2C\%22H\%22\%2C\%22H\%22\%2C\%22H\%22\%2C\%22H\%22\%5D\%2C\%5B\%22zpar\%22\%2C\%22zpar\%22\%2C\%22zpar\%22\%2C\%22zpar\%22\%2C\%22zpar\%22\%2C\%22\%E2\%88\%9Ai\%22\%5D\%2C\%5B\%22zpar\%22\%2C\%22zpar\%22\%2C\%22zpar\%22\%2C1\%2C1\%2C\%22\%E2\%88\%9Ai\%22\%5D\%2C\%5B\%22zpar\%22\%2C\%22zpar\%22\%2C1\%2C\%22zpar\%22\%2C1\%2C\%22\%E2\%88\%9Ai\%22\%5D\%2C\%5B\%22zpar\%22\%2C1\%2C\%22zpar\%22\%2C\%22zpar\%22\%2C1\%2C\%22\%E2\%88\%9Ai\%22\%5D\%2C\%5B\%22zpar\%22\%2C\%22zpar\%22\%2C1\%2C1\%2C\%22zpar\%22\%2C\%22\%E2\%88\%9Ai\%22\%5D\%2C\%5B\%22zpar\%22\%2C1\%2C\%22zpar\%22\%2C1\%2C\%22zpar\%22\%2C\%22\%E2\%88\%9Ai\%22\%5D\%2C\%5B\%22zpar\%22\%2C1\%2C1\%2C\%22zpar\%22\%2C\%22zpar\%22\%2C\%22\%E2\%88\%9Ai\%22\%5D\%2C\%5B1\%2C1\%2C\%22zpar\%22\%2C\%22zpar\%22\%2C\%22zpar\%22\%2C\%22\%E2\%88\%9Ai\%22\%5D\%2C\%5B1\%2C\%22zpar\%22\%2C\%22zpar\%22\%2C1\%2C\%22zpar\%22\%2C\%22\%E2\%88\%9Ai\%22\%5D\%2C\%5B1\%2C\%22zpar\%22\%2C1\%2C\%22zpar\%22\%2C\%22zpar\%22\%2C\%22\%E2\%88\%9Ai\%22\%5D\%2C\%5B1\%2C\%22zpar\%22\%2C\%22zpar\%22\%2C\%22zpar\%22\%2C1\%2C\%22\%E2\%88\%9Ai\%22\%5D\%2C\%5B\%22~d9b7\%22\%2C\%22~d9b7\%22\%2C\%22~d9b7\%22\%2C\%22~d9b7\%22\%2C\%22~d9b7\%22\%5D\%2C\%5B\%22H\%22\%2C\%22H\%22\%2C\%22H\%22\%2C\%22H\%22\%2C\%22H\%22\%5D\%2C\%5B\%22\%7C0\%E2\%9F\%A9\%E2\%9F\%A80\%7C\%22\%2C\%22\%7C0\%E2\%9F\%A9\%E2\%9F\%A80\%7C\%22\%2C\%22\%7C0\%E2\%9F\%A9\%E2\%9F\%A80\%7C\%22\%2C\%22\%7C0\%E2\%9F\%A9\%E2\%9F\%A80\%7C\%22\%2C\%22\%7C0\%E2\%9F\%A9\%E2\%9F\%A80\%7C\%22\%5D\%5D\%2C\%22gates\%22\%3A\%5B\%7B\%22id\%22\%3A\%22~dvdi\%22\%2C\%22matrix\%22\%3A\%22\%7B\%7B\%E2\%88\%9A\%C2\%BD-\%E2\%88\%9A\%C2\%BDi\%2C0\%7D\%2C\%7B0\%2C\%E2\%88\%9A\%C2\%BD\%2B\%E2\%88\%9A\%C2\%BDi\%7D\%7D\%22\%7D\%2C\%7B\%22id\%22\%3A\%22~2gmh\%22\%2C\%22name\%22\%3A\%22Z(pi\%2F4)\%22\%2C\%22matrix\%22\%3A\%22\%7B\%7B\%E2\%88\%9A\%C2\%BD-\%E2\%88\%9A\%C2\%BDi\%2C0\%7D\%2C\%7B0\%2C\%E2\%88\%9A\%C2\%BD\%2B\%E2\%88\%9A\%C2\%BDi\%7D\%7D\%22\%7D\%2C\%7B\%22id\%22\%3A\%22~ntmg\%22\%2C\%22name\%22\%3A\%22Z(pi\%2F4)\%22\%2C\%22matrix\%22\%3A\%22\%7B\%7B0.9238795-0.3826834i\%2C0\%7D\%2C\%7B0\%2C0.9238795\%2B0.3826834i\%7D\%7D\%22\%7D\%2C\%7B\%22id\%22\%3A\%22~d9b7\%22\%2C\%22name\%22\%3A\%22Z(pi\%2F8)\%22\%2C\%22matrix\%22\%3A\%22\%7B\%7B0.9238795-0.3826834i\%2C0\%7D\%2C\%7B0\%2C0.9238795\%2B0.3826834i\%7D\%7D\%22\%7D\%2C\%7B\%22id\%22\%3A\%22~703g\%22\%2C\%22name\%22\%3A\%22Z(-\%5C\%5Cpi\%2F8)\%22\%2C\%22matrix\%22\%3A\%22\%7B\%7B0.9238795\%2B0.3826834i\%2C0\%7D\%2C\%7B0\%2C0.9238795-0.3826834i\%7D\%7D\%22\%7D\%2C\%7B\%22id\%22\%3A\%22~i2pe\%22\%2C\%22name\%22\%3A\%22Z(-pi\%2F4)\%22\%2C\%22matrix\%22\%3A\%22\%7B\%7B\%E2\%88\%9A\%C2\%BD\%2B\%E2\%88\%9A\%C2\%BDi\%2C0\%7D\%2C\%7B0\%2C\%E2\%88\%9A\%C2\%BD-\%E2\%88\%9A\%C2\%BDi\%7D\%7D\%22\%7D\%5D\%7D}
\newcommand{\qaaazfourtwo}{https://algassert.com/quirk\#circuit=\%7B\%22cols\%22\%3A\%5B\%5B\%22H\%22\%2C\%22H\%22\%2C\%22H\%22\%2C\%22H\%22\%2C\%22H\%22\%5D\%2C\%5B\%22zpar\%22\%2C\%22zpar\%22\%2C\%22zpar\%22\%2C\%22zpar\%22\%2C\%22zpar\%22\%2C\%22\%E2\%88\%9Ai\%22\%5D\%2C\%5B\%22zpar\%22\%2C\%22zpar\%22\%2C\%22zpar\%22\%2C1\%2C1\%2C\%22\%E2\%88\%9Ai\%22\%5D\%2C\%5B\%22zpar\%22\%2C\%22zpar\%22\%2C1\%2C\%22zpar\%22\%2C1\%2C\%22\%E2\%88\%9Ai\%22\%5D\%2C\%5B\%22zpar\%22\%2C1\%2C\%22zpar\%22\%2C\%22zpar\%22\%2C1\%2C\%22\%E2\%88\%9Ai\%22\%5D\%2C\%5B\%22zpar\%22\%2C\%22zpar\%22\%2C1\%2C1\%2C\%22zpar\%22\%2C\%22\%E2\%88\%9Ai\%22\%5D\%2C\%5B\%22zpar\%22\%2C1\%2C\%22zpar\%22\%2C1\%2C\%22zpar\%22\%2C\%22\%E2\%88\%9Ai\%22\%5D\%2C\%5B\%22zpar\%22\%2C1\%2C1\%2C\%22zpar\%22\%2C\%22zpar\%22\%2C\%22\%E2\%88\%9Ai\%22\%5D\%2C\%5B\%22~d9b7\%22\%5D\%2C\%5B\%22H\%22\%5D\%2C\%5B\%22\%7C0\%E2\%9F\%A9\%E2\%9F\%A80\%7C\%22\%5D\%2C\%5B1\%2C\%22Bloch\%22\%2C\%22Bloch\%22\%2C\%22Bloch\%22\%2C\%22Bloch\%22\%5D\%2C\%5B1\%2C\%22X\%5E\%C2\%BD\%22\%5D\%2C\%5B1\%2C\%22\%E2\%80\%A2\%22\%2C\%22X\%22\%2C\%22X\%22\%2C\%22X\%22\%5D\%2C\%5B1\%2C\%22H\%22\%5D\%2C\%5B1\%2C\%22\%7C0\%E2\%9F\%A9\%E2\%9F\%A80\%7C\%22\%5D\%2C\%5B1\%2C1\%2C\%22\%E2\%8A\%96\%22\%2C\%22Z\%22\%2C\%22Z\%22\%5D\%2C\%5B1\%2C1\%2C\%22Z\%5E-\%C2\%BC\%22\%5D\%2C\%5B1\%2C1\%2C\%22\%E2\%8A\%96\%22\%2C\%22Z\%22\%2C\%22Z\%22\%5D\%5D\%2C\%22gates\%22\%3A\%5B\%7B\%22id\%22\%3A\%22~dvdi\%22\%2C\%22matrix\%22\%3A\%22\%7B\%7B\%E2\%88\%9A\%C2\%BD-\%E2\%88\%9A\%C2\%BDi\%2C0\%7D\%2C\%7B0\%2C\%E2\%88\%9A\%C2\%BD\%2B\%E2\%88\%9A\%C2\%BDi\%7D\%7D\%22\%7D\%2C\%7B\%22id\%22\%3A\%22~2gmh\%22\%2C\%22name\%22\%3A\%22Z(pi\%2F4)\%22\%2C\%22matrix\%22\%3A\%22\%7B\%7B\%E2\%88\%9A\%C2\%BD-\%E2\%88\%9A\%C2\%BDi\%2C0\%7D\%2C\%7B0\%2C\%E2\%88\%9A\%C2\%BD\%2B\%E2\%88\%9A\%C2\%BDi\%7D\%7D\%22\%7D\%2C\%7B\%22id\%22\%3A\%22~ntmg\%22\%2C\%22name\%22\%3A\%22Z(pi\%2F4)\%22\%2C\%22matrix\%22\%3A\%22\%7B\%7B0.9238795-0.3826834i\%2C0\%7D\%2C\%7B0\%2C0.9238795\%2B0.3826834i\%7D\%7D\%22\%7D\%2C\%7B\%22id\%22\%3A\%22~d9b7\%22\%2C\%22name\%22\%3A\%22Z(pi\%2F8)\%22\%2C\%22matrix\%22\%3A\%22\%7B\%7B0.9238795-0.3826834i\%2C0\%7D\%2C\%7B0\%2C0.9238795\%2B0.3826834i\%7D\%7D\%22\%7D\%2C\%7B\%22id\%22\%3A\%22~703g\%22\%2C\%22name\%22\%3A\%22Z(-\%5C\%5Cpi\%2F8)\%22\%2C\%22matrix\%22\%3A\%22\%7B\%7B0.9238795\%2B0.3826834i\%2C0\%7D\%2C\%7B0\%2C0.9238795-0.3826834i\%7D\%7D\%22\%7D\%2C\%7B\%22id\%22\%3A\%22~i2pe\%22\%2C\%22name\%22\%3A\%22Z(-pi\%2F4)\%22\%2C\%22matrix\%22\%3A\%22\%7B\%7B\%E2\%88\%9A\%C2\%BD\%2B\%E2\%88\%9A\%C2\%BDi\%2C0\%7D\%2C\%7B0\%2C\%E2\%88\%9A\%C2\%BD-\%E2\%88\%9A\%C2\%BDi\%7D\%7D\%22\%7D\%5D\%7D}

\newcommand{\qtwelvethreetwo}{https://algassert.com/quirk\#circuit=\%7B\%22cols\%22\%3A\%5B\%5B\%22H\%22\%2C\%22H\%22\%2C\%22H\%22\%2C\%22H\%22\%2C\%22H\%22\%5D\%2C\%5B\%22zpar\%22\%2C\%22zpar\%22\%2C\%22zpar\%22\%2C\%22zpar\%22\%2C\%22zpar\%22\%2C\%22\%E2\%88\%9Ai\%22\%5D\%2C\%5B\%22zpar\%22\%2C\%22zpar\%22\%2C\%22zpar\%22\%2C1\%2C1\%2C\%22\%E2\%88\%9Ai\%22\%5D\%2C\%5B\%22zpar\%22\%2C\%22zpar\%22\%2C1\%2C\%22zpar\%22\%2C1\%2C\%22\%E2\%88\%9Ai\%22\%5D\%2C\%5B\%22zpar\%22\%2C\%22zpar\%22\%2C1\%2C1\%2C\%22zpar\%22\%2C\%22\%E2\%88\%9Ai\%22\%5D\%2C\%5B\%22zpar\%22\%2C1\%2C\%22zpar\%22\%2C\%22zpar\%22\%2C1\%2C\%22\%E2\%88\%9Ai\%22\%5D\%2C\%5B\%22zpar\%22\%2C1\%2C\%22zpar\%22\%2C1\%2C\%22zpar\%22\%2C\%22\%E2\%88\%9Ai\%22\%5D\%2C\%5B\%22zpar\%22\%2C1\%2C1\%2C\%22zpar\%22\%2C\%22zpar\%22\%2C\%22\%E2\%88\%9Ai\%22\%5D\%2C\%5B\%22~d9b7\%22\%5D\%2C\%5B\%22H\%22\%5D\%2C\%5B\%22\%7C0\%E2\%9F\%A9\%E2\%9F\%A80\%7C\%22\%5D\%2C\%5B\%22Bloch\%22\%2C\%22Bloch\%22\%2C\%22Bloch\%22\%2C\%22Bloch\%22\%2C\%22Bloch\%22\%5D\%2C\%5B1\%2C\%22zpar\%22\%2C1\%2C\%22zpar\%22\%2C\%22zpar\%22\%2C\%22\%E2\%88\%9Ai\%22\%5D\%2C\%5B1\%2C\%22zpar\%22\%2C\%22zpar\%22\%2C1\%2C\%22zpar\%22\%2C\%22\%E2\%88\%9Ai\%22\%5D\%2C\%5B1\%2C\%22zpar\%22\%2C\%22zpar\%22\%2C\%22zpar\%22\%2C1\%2C\%22\%E2\%88\%9Ai\%22\%5D\%2C\%5B1\%2C\%22~d9b7\%22\%5D\%2C\%5B1\%2C\%22H\%22\%5D\%2C\%5B1\%2C\%22\%7C0\%E2\%9F\%A9\%E2\%9F\%A80\%7C\%22\%5D\%2C\%5B1\%2C1\%2C\%22zpar\%22\%2C\%22zpar\%22\%2C\%22zpar\%22\%2C\%22\%E2\%88\%9Ai\%22\%5D\%5D\%2C\%22gates\%22\%3A\%5B\%7B\%22id\%22\%3A\%22~dvdi\%22\%2C\%22matrix\%22\%3A\%22\%7B\%7B\%E2\%88\%9A\%C2\%BD-\%E2\%88\%9A\%C2\%BDi\%2C0\%7D\%2C\%7B0\%2C\%E2\%88\%9A\%C2\%BD\%2B\%E2\%88\%9A\%C2\%BDi\%7D\%7D\%22\%7D\%2C\%7B\%22id\%22\%3A\%22~2gmh\%22\%2C\%22name\%22\%3A\%22Z(pi\%2F4)\%22\%2C\%22matrix\%22\%3A\%22\%7B\%7B\%E2\%88\%9A\%C2\%BD-\%E2\%88\%9A\%C2\%BDi\%2C0\%7D\%2C\%7B0\%2C\%E2\%88\%9A\%C2\%BD\%2B\%E2\%88\%9A\%C2\%BDi\%7D\%7D\%22\%7D\%2C\%7B\%22id\%22\%3A\%22~ntmg\%22\%2C\%22name\%22\%3A\%22Z(pi\%2F4)\%22\%2C\%22matrix\%22\%3A\%22\%7B\%7B0.9238795-0.3826834i\%2C0\%7D\%2C\%7B0\%2C0.9238795\%2B0.3826834i\%7D\%7D\%22\%7D\%2C\%7B\%22id\%22\%3A\%22~d9b7\%22\%2C\%22name\%22\%3A\%22Z(pi\%2F8)\%22\%2C\%22matrix\%22\%3A\%22\%7B\%7B0.9238795-0.3826834i\%2C0\%7D\%2C\%7B0\%2C0.9238795\%2B0.3826834i\%7D\%7D\%22\%7D\%2C\%7B\%22id\%22\%3A\%22~703g\%22\%2C\%22name\%22\%3A\%22Z(-\%5C\%5Cpi\%2F8)\%22\%2C\%22matrix\%22\%3A\%22\%7B\%7B0.9238795\%2B0.3826834i\%2C0\%7D\%2C\%7B0\%2C0.9238795-0.3826834i\%7D\%7D\%22\%7D\%2C\%7B\%22id\%22\%3A\%22~i2pe\%22\%2C\%22name\%22\%3A\%22Z(-pi\%2F4)\%22\%2C\%22matrix\%22\%3A\%22\%7B\%7B\%E2\%88\%9A\%C2\%BD\%2B\%E2\%88\%9A\%C2\%BDi\%2C0\%7D\%2C\%7B0\%2C\%E2\%88\%9A\%C2\%BD-\%E2\%88\%9A\%C2\%BDi\%7D\%7D\%22\%7D\%5D\%7D}
\newcommand{\qfourteentwotwo}{https://algassert.com/quirk\#circuit=\%7B\%22cols\%22\%3A\%5B\%5B\%22H\%22\%2C\%22H\%22\%2C\%22H\%22\%2C\%22H\%22\%2C\%22H\%22\%5D\%2C\%5B\%22zpar\%22\%2C\%22zpar\%22\%2C\%22zpar\%22\%2C\%22zpar\%22\%2C\%22zpar\%22\%2C\%22\%E2\%88\%9Ai\%22\%5D\%2C\%5B\%22zpar\%22\%2C\%22zpar\%22\%2C\%22zpar\%22\%2C1\%2C1\%2C\%22\%E2\%88\%9Ai\%22\%5D\%2C\%5B\%22zpar\%22\%2C\%22zpar\%22\%2C1\%2C\%22zpar\%22\%2C1\%2C\%22\%E2\%88\%9Ai\%22\%5D\%2C\%5B\%22zpar\%22\%2C\%22zpar\%22\%2C1\%2C1\%2C\%22zpar\%22\%2C\%22\%E2\%88\%9Ai\%22\%5D\%2C\%5B\%22zpar\%22\%2C1\%2C\%22zpar\%22\%2C\%22zpar\%22\%2C1\%2C\%22\%E2\%88\%9Ai\%22\%5D\%2C\%5B\%22zpar\%22\%2C1\%2C\%22zpar\%22\%2C1\%2C\%22zpar\%22\%2C\%22\%E2\%88\%9Ai\%22\%5D\%2C\%5B\%22zpar\%22\%2C1\%2C1\%2C\%22zpar\%22\%2C\%22zpar\%22\%2C\%22\%E2\%88\%9Ai\%22\%5D\%2C\%5B1\%2C\%22zpar\%22\%2C\%22zpar\%22\%2C1\%2C\%22zpar\%22\%2C\%22\%E2\%88\%9Ai\%22\%5D\%2C\%5B1\%2C\%22zpar\%22\%2C1\%2C\%22zpar\%22\%2C\%22zpar\%22\%2C\%22\%E2\%88\%9Ai\%22\%5D\%2C\%5B1\%2C\%22zpar\%22\%2C\%22zpar\%22\%2C\%22zpar\%22\%2C1\%2C\%22\%E2\%88\%9Ai\%22\%5D\%2C\%5B\%22~d9b7\%22\%2C\%22~d9b7\%22\%5D\%2C\%5B\%22H\%22\%2C\%22H\%22\%5D\%2C\%5B\%22\%7C0\%E2\%9F\%A9\%E2\%9F\%A80\%7C\%22\%2C\%22\%7C0\%E2\%9F\%A9\%E2\%9F\%A80\%7C\%22\%5D\%2C\%5B\%22Bloch\%22\%2C\%22Bloch\%22\%2C\%22Bloch\%22\%2C\%22Bloch\%22\%2C\%22Bloch\%22\%5D\%2C\%5B1\%2C1\%2C\%22zpar\%22\%2C\%22zpar\%22\%2C\%22zpar\%22\%2C\%22\%E2\%88\%9Ai\%22\%5D\%2C\%5B1\%2C1\%2C\%22~d9b7\%22\%5D\%2C\%5B1\%2C1\%2C\%22H\%22\%5D\%5D\%2C\%22gates\%22\%3A\%5B\%7B\%22id\%22\%3A\%22~dvdi\%22\%2C\%22matrix\%22\%3A\%22\%7B\%7B\%E2\%88\%9A\%C2\%BD-\%E2\%88\%9A\%C2\%BDi\%2C0\%7D\%2C\%7B0\%2C\%E2\%88\%9A\%C2\%BD\%2B\%E2\%88\%9A\%C2\%BDi\%7D\%7D\%22\%7D\%2C\%7B\%22id\%22\%3A\%22~2gmh\%22\%2C\%22name\%22\%3A\%22Z(pi\%2F4)\%22\%2C\%22matrix\%22\%3A\%22\%7B\%7B\%E2\%88\%9A\%C2\%BD-\%E2\%88\%9A\%C2\%BDi\%2C0\%7D\%2C\%7B0\%2C\%E2\%88\%9A\%C2\%BD\%2B\%E2\%88\%9A\%C2\%BDi\%7D\%7D\%22\%7D\%2C\%7B\%22id\%22\%3A\%22~ntmg\%22\%2C\%22name\%22\%3A\%22Z(pi\%2F4)\%22\%2C\%22matrix\%22\%3A\%22\%7B\%7B0.9238795-0.3826834i\%2C0\%7D\%2C\%7B0\%2C0.9238795\%2B0.3826834i\%7D\%7D\%22\%7D\%2C\%7B\%22id\%22\%3A\%22~d9b7\%22\%2C\%22name\%22\%3A\%22Z(pi\%2F8)\%22\%2C\%22matrix\%22\%3A\%22\%7B\%7B0.9238795-0.3826834i\%2C0\%7D\%2C\%7B0\%2C0.9238795\%2B0.3826834i\%7D\%7D\%22\%7D\%2C\%7B\%22id\%22\%3A\%22~703g\%22\%2C\%22name\%22\%3A\%22Z(-\%5C\%5Cpi\%2F8)\%22\%2C\%22matrix\%22\%3A\%22\%7B\%7B0.9238795\%2B0.3826834i\%2C0\%7D\%2C\%7B0\%2C0.9238795-0.3826834i\%7D\%7D\%22\%7D\%2C\%7B\%22id\%22\%3A\%22~i2pe\%22\%2C\%22name\%22\%3A\%22Z(-pi\%2F4)\%22\%2C\%22matrix\%22\%3A\%22\%7B\%7B\%E2\%88\%9A\%C2\%BD\%2B\%E2\%88\%9A\%C2\%BDi\%2C0\%7D\%2C\%7B0\%2C\%E2\%88\%9A\%C2\%BD-\%E2\%88\%9A\%C2\%BDi\%7D\%7D\%22\%7D\%5D\%7D}
\newcommand{\qsummaryfifteen}{https://algassert.com/quirk\#circuit=\%7B\%22cols\%22\%3A\%5B\%5B\%22H\%22\%2C\%22H\%22\%2C\%22H\%22\%2C\%22H\%22\%2C\%22H\%22\%5D\%2C\%5B\%22zpar\%22\%2C\%22zpar\%22\%2C\%22zpar\%22\%2C\%22zpar\%22\%2C\%22zpar\%22\%2C\%22\%E2\%88\%9Ai\%22\%5D\%2C\%5B\%22zpar\%22\%2C\%22zpar\%22\%2C\%22zpar\%22\%2C1\%2C1\%2C\%22\%E2\%88\%9Ai\%22\%5D\%2C\%5B\%22zpar\%22\%2C\%22zpar\%22\%2C1\%2C\%22zpar\%22\%2C1\%2C\%22\%E2\%88\%9Ai\%22\%5D\%2C\%5B\%22zpar\%22\%2C\%22zpar\%22\%2C1\%2C1\%2C\%22zpar\%22\%2C\%22\%E2\%88\%9Ai\%22\%5D\%2C\%5B\%22zpar\%22\%2C1\%2C\%22zpar\%22\%2C\%22zpar\%22\%2C1\%2C\%22\%E2\%88\%9Ai\%22\%5D\%2C\%5B\%22zpar\%22\%2C1\%2C\%22zpar\%22\%2C1\%2C\%22zpar\%22\%2C\%22\%E2\%88\%9Ai\%22\%5D\%2C\%5B\%22zpar\%22\%2C1\%2C1\%2C\%22zpar\%22\%2C\%22zpar\%22\%2C\%22\%E2\%88\%9Ai\%22\%5D\%2C\%5B1\%2C\%22zpar\%22\%2C\%22zpar\%22\%2C1\%2C\%22zpar\%22\%2C\%22\%E2\%88\%9Ai\%22\%5D\%2C\%5B1\%2C\%22zpar\%22\%2C1\%2C\%22zpar\%22\%2C\%22zpar\%22\%2C\%22\%E2\%88\%9Ai\%22\%5D\%2C\%5B1\%2C\%22zpar\%22\%2C\%22zpar\%22\%2C\%22zpar\%22\%2C1\%2C\%22\%E2\%88\%9Ai\%22\%5D\%2C\%5B\%22~d9b7\%22\%2C\%22~d9b7\%22\%5D\%2C\%5B\%22H\%22\%2C\%22H\%22\%5D\%2C\%5B\%22\%7C0\%E2\%9F\%A9\%E2\%9F\%A80\%7C\%22\%2C\%22\%7C0\%E2\%9F\%A9\%E2\%9F\%A80\%7C\%22\%5D\%2C\%5B\%22Bloch\%22\%2C\%22Bloch\%22\%2C\%22Bloch\%22\%2C\%22Bloch\%22\%2C\%22Bloch\%22\%5D\%2C\%5B1\%2C1\%2C\%22zpar\%22\%2C\%22zpar\%22\%2C\%22zpar\%22\%2C\%22\%E2\%88\%9Ai\%22\%5D\%2C\%5B1\%2C1\%2C\%22~d9b7\%22\%5D\%2C\%5B1\%2C1\%2C\%22H\%22\%5D\%2C\%5B1\%2C1\%2C\%22\%7C0\%E2\%9F\%A9\%E2\%9F\%A80\%7C\%22\%5D\%2C\%5B1\%2C1\%2C1\%2C\%22~d9b7\%22\%5D\%2C\%5B1\%2C1\%2C1\%2C\%22H\%22\%5D\%2C\%5B1\%2C1\%2C1\%2C\%22\%7C0\%E2\%9F\%A9\%E2\%9F\%A80\%7C\%22\%5D\%5D\%2C\%22gates\%22\%3A\%5B\%7B\%22id\%22\%3A\%22~dvdi\%22\%2C\%22matrix\%22\%3A\%22\%7B\%7B\%E2\%88\%9A\%C2\%BD-\%E2\%88\%9A\%C2\%BDi\%2C0\%7D\%2C\%7B0\%2C\%E2\%88\%9A\%C2\%BD\%2B\%E2\%88\%9A\%C2\%BDi\%7D\%7D\%22\%7D\%2C\%7B\%22id\%22\%3A\%22~2gmh\%22\%2C\%22name\%22\%3A\%22Z(pi\%2F4)\%22\%2C\%22matrix\%22\%3A\%22\%7B\%7B\%E2\%88\%9A\%C2\%BD-\%E2\%88\%9A\%C2\%BDi\%2C0\%7D\%2C\%7B0\%2C\%E2\%88\%9A\%C2\%BD\%2B\%E2\%88\%9A\%C2\%BDi\%7D\%7D\%22\%7D\%2C\%7B\%22id\%22\%3A\%22~ntmg\%22\%2C\%22name\%22\%3A\%22Z(pi\%2F4)\%22\%2C\%22matrix\%22\%3A\%22\%7B\%7B0.9238795-0.3826834i\%2C0\%7D\%2C\%7B0\%2C0.9238795\%2B0.3826834i\%7D\%7D\%22\%7D\%2C\%7B\%22id\%22\%3A\%22~d9b7\%22\%2C\%22name\%22\%3A\%22Z(pi\%2F8)\%22\%2C\%22matrix\%22\%3A\%22\%7B\%7B0.9238795-0.3826834i\%2C0\%7D\%2C\%7B0\%2C0.9238795\%2B0.3826834i\%7D\%7D\%22\%7D\%2C\%7B\%22id\%22\%3A\%22~703g\%22\%2C\%22name\%22\%3A\%22Z(-\%5C\%5Cpi\%2F8)\%22\%2C\%22matrix\%22\%3A\%22\%7B\%7B0.9238795\%2B0.3826834i\%2C0\%7D\%2C\%7B0\%2C0.9238795-0.3826834i\%7D\%7D\%22\%7D\%2C\%7B\%22id\%22\%3A\%22~i2pe\%22\%2C\%22name\%22\%3A\%22Z(-pi\%2F4)\%22\%2C\%22matrix\%22\%3A\%22\%7B\%7B\%E2\%88\%9A\%C2\%BD\%2B\%E2\%88\%9A\%C2\%BDi\%2C0\%7D\%2C\%7B0\%2C\%E2\%88\%9A\%C2\%BD-\%E2\%88\%9A\%C2\%BDi\%7D\%7D\%22\%7D\%5D\%7D}

\newcommand{\qparentsix}{https://algassert.com/quirk\#circuit=\%7B\%22cols\%22\%3A\%5B\%5B\%22H\%22\%2C\%22H\%22\%2C\%22H\%22\%2C\%22H\%22\%2C\%22H\%22\%2C\%22H\%22\%5D\%2C\%5B\%22zpar\%22\%2C\%22zpar\%22\%2C1\%2C\%22zpar\%22\%2C\%22zpar\%22\%2C\%22zpar\%22\%2C\%22\%E2\%88\%9Ai\%22\%5D\%2C\%5B\%22zpar\%22\%2C1\%2C\%22zpar\%22\%2C\%22zpar\%22\%2C\%22zpar\%22\%2C\%22zpar\%22\%2C\%22\%E2\%88\%9Ai\%22\%5D\%2C\%5B\%22zpar\%22\%2C\%22zpar\%22\%2C1\%2C1\%2C1\%2C\%22zpar\%22\%2C\%22\%E2\%88\%9Ai\%22\%5D\%2C\%5B\%22zpar\%22\%2C1\%2C\%22zpar\%22\%2C1\%2C1\%2C\%22zpar\%22\%2C\%22\%E2\%88\%9Ai\%22\%5D\%2C\%5B\%22zpar\%22\%2C\%22zpar\%22\%2C1\%2C1\%2C\%22zpar\%22\%2C1\%2C\%22\%E2\%88\%9Ai\%22\%5D\%2C\%5B\%22zpar\%22\%2C1\%2C\%22zpar\%22\%2C1\%2C\%22zpar\%22\%2C1\%2C\%22\%E2\%88\%9Ai\%22\%5D\%2C\%5B\%22zpar\%22\%2C\%22zpar\%22\%2C1\%2C\%22zpar\%22\%2C1\%2C1\%2C\%22\%E2\%88\%9Ai\%22\%5D\%2C\%5B\%22zpar\%22\%2C1\%2C\%22zpar\%22\%2C\%22zpar\%22\%2C1\%2C1\%2C\%22\%E2\%88\%9Ai\%22\%5D\%2C\%5B\%22H\%22\%5D\%2C\%5B\%22Chance\%22\%5D\%2C\%5B\%22\%7C0\%E2\%9F\%A9\%E2\%9F\%A80\%7C\%22\%5D\%2C\%5B1\%2C\%22zpar\%22\%2C\%22zpar\%22\%2C\%22zpar\%22\%2C\%22zpar\%22\%2C\%22zpar\%22\%2C\%22\%E2\%88\%9Ai\%22\%5D\%2C\%5B1\%2C\%22zpar\%22\%2C\%22zpar\%22\%2C1\%2C1\%2C\%22zpar\%22\%2C\%22\%E2\%88\%9Ai\%22\%5D\%2C\%5B1\%2C\%22zpar\%22\%2C\%22zpar\%22\%2C1\%2C\%22zpar\%22\%2C1\%2C\%22\%E2\%88\%9Ai\%22\%5D\%2C\%5B1\%2C\%22zpar\%22\%2C\%22zpar\%22\%2C\%22zpar\%22\%2C1\%2C1\%2C\%22\%E2\%88\%9Ai\%22\%5D\%2C\%5B1\%2C\%22H\%22\%2C\%22H\%22\%5D\%2C\%5B1\%2C\%22Chance\%22\%2C\%22Chance\%22\%5D\%2C\%5B1\%2C\%22\%7C0\%E2\%9F\%A9\%E2\%9F\%A80\%7C\%22\%2C\%22\%7C0\%E2\%9F\%A9\%E2\%9F\%A80\%7C\%22\%5D\%2C\%5B1\%2C1\%2C1\%2C\%22Bloch\%22\%2C\%22Bloch\%22\%2C\%22Bloch\%22\%5D\%2C\%5B1\%2C1\%2C1\%2C\%22zpar\%22\%2C\%22zpar\%22\%2C\%22zpar\%22\%2C\%22\%E2\%88\%9Ai\%22\%5D\%2C\%5B\%22\%E2\%80\%A6\%22\%2C\%22\%E2\%80\%A6\%22\%5D\%5D\%2C\%22gates\%22\%3A\%5B\%7B\%22id\%22\%3A\%22~dvdi\%22\%2C\%22matrix\%22\%3A\%22\%7B\%7B\%E2\%88\%9A\%C2\%BD-\%E2\%88\%9A\%C2\%BDi\%2C0\%7D\%2C\%7B0\%2C\%E2\%88\%9A\%C2\%BD\%2B\%E2\%88\%9A\%C2\%BDi\%7D\%7D\%22\%7D\%2C\%7B\%22id\%22\%3A\%22~2gmh\%22\%2C\%22name\%22\%3A\%22Z(pi\%2F4)\%22\%2C\%22matrix\%22\%3A\%22\%7B\%7B\%E2\%88\%9A\%C2\%BD-\%E2\%88\%9A\%C2\%BDi\%2C0\%7D\%2C\%7B0\%2C\%E2\%88\%9A\%C2\%BD\%2B\%E2\%88\%9A\%C2\%BDi\%7D\%7D\%22\%7D\%2C\%7B\%22id\%22\%3A\%22~ntmg\%22\%2C\%22name\%22\%3A\%22Z(pi\%2F4)\%22\%2C\%22matrix\%22\%3A\%22\%7B\%7B0.9238795-0.3826834i\%2C0\%7D\%2C\%7B0\%2C0.9238795\%2B0.3826834i\%7D\%7D\%22\%7D\%2C\%7B\%22id\%22\%3A\%22~d9b7\%22\%2C\%22name\%22\%3A\%22Z(pi\%2F8)\%22\%2C\%22matrix\%22\%3A\%22\%7B\%7B0.9238795-0.3826834i\%2C0\%7D\%2C\%7B0\%2C0.9238795\%2B0.3826834i\%7D\%7D\%22\%7D\%2C\%7B\%22id\%22\%3A\%22~703g\%22\%2C\%22name\%22\%3A\%22Z(-\%5C\%5Cpi\%2F8)\%22\%2C\%22matrix\%22\%3A\%22\%7B\%7B0.9238795\%2B0.3826834i\%2C0\%7D\%2C\%7B0\%2C0.9238795-0.3826834i\%7D\%7D\%22\%7D\%2C\%7B\%22id\%22\%3A\%22~i2pe\%22\%2C\%22name\%22\%3A\%22Z(-pi\%2F4)\%22\%2C\%22matrix\%22\%3A\%22\%7B\%7B\%E2\%88\%9A\%C2\%BD\%2B\%E2\%88\%9A\%C2\%BDi\%2C0\%7D\%2C\%7B0\%2C\%E2\%88\%9A\%C2\%BD-\%E2\%88\%9A\%C2\%BDi\%7D\%7D\%22\%7D\%5D\%7D}

\newcommand{\qfourteensixtwo}
{https://algassert.com/quirk\#circuit={\%22cols\%22:[[\%22H\%22,\%22H\%22,\%22H\%22,\%22H\%22,\%22H\%22,\%22H\%22,\%22H\%22],[1,1,1,\%22zpar\%22,1,1,1,\%22i\%22],[1,1,1,1,\%22zpar\%22,1,1,\%22i\%22],[1,1,1,1,1,1,\%22zpar\%22,\%22√-i\%22],[1,1,1,\%22zpar\%22,1,1,\%22zpar\%22,\%22√-i\%22],[1,1,1,1,\%22zpar\%22,1,\%22zpar\%22,\%22√-i\%22],[1,1,1,1,1,\%22zpar\%22,\%22zpar\%22,\%22√i\%22],[1,1,1,\%22zpar\%22,1,\%22zpar\%22,\%22zpar\%22,\%22√i\%22],[1,1,1,1,\%22zpar\%22,\%22zpar\%22,\%22zpar\%22,\%22√i\%22],[\%22zpar\%22,1,1,\%22zpar\%22,\%22zpar\%22,1,\%22zpar\%22,\%22√i\%22],[1,\%22zpar\%22,1,\%22zpar\%22,\%22zpar\%22,1,\%22zpar\%22,\%22√i\%22],[1,1,\%22zpar\%22,\%22zpar\%22,\%22zpar\%22,1,\%22zpar\%22,\%22√i\%22],[1,1,1,\%22zpar\%22,\%22zpar\%22,\%22zpar\%22,\%22zpar\%22,\%22√i\%22],[\%22zpar\%22,\%22zpar\%22,1,\%22zpar\%22,\%22zpar\%22,1,\%22zpar\%22,\%22√i\%22],[\%22zpar\%22,1,\%22zpar\%22,\%22zpar\%22,\%22zpar\%22,1,\%22zpar\%22,\%22√i\%22],[1,\%22zpar\%22,\%22zpar\%22,\%22zpar\%22,\%22zpar\%22,1,\%22zpar\%22,\%22√i\%22],[\%22zpar\%22,\%22zpar\%22,\%22zpar\%22,\%22zpar\%22,\%22zpar\%22,1,\%22zpar\%22,\%22√i\%22],[1,1,1,1,1,1,\%22H\%22],[1,1,1,1,1,1,\%22|0⟩⟨0|\%22]]}}
\newcommand{\qeighteenfourtwo}{https://algassert.com/quirk\#circuit=\%7B\%22cols\%22\%3A\%5B\%5B\%22H\%22\%2C\%22H\%22\%2C\%22H\%22\%2C\%22H\%22\%2C\%22H\%22\%2C\%22H\%22\%5D\%2C\%5B1\%2C1\%2C1\%2C1\%2C\%22zpar\%22\%2C1\%2C\%22\%E2\%88\%9Ai\%22\%5D\%2C\%5B1\%2C1\%2C1\%2C1\%2C1\%2C\%22zpar\%22\%2C\%22\%E2\%88\%9Ai\%22\%5D\%2C\%5B1\%2C1\%2C1\%2C\%22zpar\%22\%2C\%22zpar\%22\%2C1\%2C\%22\%E2\%88\%9Ai\%22\%5D\%2C\%5B1\%2C1\%2C1\%2C\%22zpar\%22\%2C1\%2C\%22zpar\%22\%2C\%22\%E2\%88\%9Ai\%22\%5D\%2C\%5B\%22zpar\%22\%2C1\%2C\%22zpar\%22\%2C1\%2C\%22zpar\%22\%2C1\%2C\%22\%E2\%88\%9A-i\%22\%5D\%2C\%5B\%22zpar\%22\%2C1\%2C\%22zpar\%22\%2C1\%2C1\%2C\%22zpar\%22\%2C\%22\%E2\%88\%9A-i\%22\%5D\%2C\%5B\%22zpar\%22\%2C1\%2C1\%2C1\%2C\%22zpar\%22\%2C\%22zpar\%22\%2C\%22\%E2\%88\%9Ai\%22\%5D\%2C\%5B1\%2C\%22zpar\%22\%2C\%22zpar\%22\%2C1\%2C\%22zpar\%22\%2C1\%2C\%22\%E2\%88\%9A-i\%22\%5D\%2C\%5B1\%2C\%22zpar\%22\%2C\%22zpar\%22\%2C1\%2C1\%2C\%22zpar\%22\%2C\%22\%E2\%88\%9A-i\%22\%5D\%2C\%5B1\%2C\%22zpar\%22\%2C1\%2C1\%2C\%22zpar\%22\%2C\%22zpar\%22\%2C\%22\%E2\%88\%9Ai\%22\%5D\%2C\%5B1\%2C1\%2C\%22zpar\%22\%2C\%22zpar\%22\%2C\%22zpar\%22\%2C1\%2C\%22\%E2\%88\%9Ai\%22\%5D\%2C\%5B1\%2C1\%2C\%22zpar\%22\%2C\%22zpar\%22\%2C1\%2C\%22zpar\%22\%2C\%22\%E2\%88\%9Ai\%22\%5D\%2C\%5B1\%2C1\%2C\%22zpar\%22\%2C1\%2C\%22zpar\%22\%2C\%22zpar\%22\%2C\%22\%E2\%88\%9Ai\%22\%5D\%2C\%5B1\%2C1\%2C1\%2C\%22zpar\%22\%2C\%22zpar\%22\%2C\%22zpar\%22\%2C\%22\%E2\%88\%9Ai\%22\%5D\%2C\%5B\%22zpar\%22\%2C\%22zpar\%22\%2C\%22zpar\%22\%2C1\%2C\%22zpar\%22\%2C1\%2C\%22\%E2\%88\%9A-i\%22\%5D\%2C\%5B\%22zpar\%22\%2C\%22zpar\%22\%2C\%22zpar\%22\%2C1\%2C1\%2C\%22zpar\%22\%2C\%22\%E2\%88\%9A-i\%22\%5D\%2C\%5B\%22zpar\%22\%2C\%22zpar\%22\%2C1\%2C1\%2C\%22zpar\%22\%2C\%22zpar\%22\%2C\%22\%E2\%88\%9Ai\%22\%5D\%2C\%5B1\%2C1\%2C\%22zpar\%22\%2C\%22zpar\%22\%2C\%22zpar\%22\%2C\%22zpar\%22\%2C\%22\%E2\%88\%9Ai\%22\%5D\%2C\%5B1\%2C1\%2C1\%2C1\%2C\%22zpar\%22\%2C\%22zpar\%22\%2C\%22i\%22\%5D\%2C\%5B1\%2C1\%2C1\%2C1\%2C\%22H\%22\%5D\%2C\%5B1\%2C1\%2C1\%2C1\%2C\%22\%7C0\%E2\%9F\%A9\%E2\%9F\%A80\%7C\%22\%5D\%2C\%5B1\%2C1\%2C1\%2C1\%2C1\%2C\%22H\%22\%5D\%2C\%5B1\%2C1\%2C1\%2C1\%2C1\%2C\%22\%7C0\%E2\%9F\%A9\%E2\%9F\%A80\%7C\%22\%5D\%2C\%5B\%22zpar\%22\%2C\%22zpar\%22\%2C1\%2C1\%2C1\%2C1\%2C\%22\%E2\%88\%9Ai\%22\%5D\%2C\%5B1\%2C1\%2C\%22zpar\%22\%2C\%22zpar\%22\%2C1\%2C1\%2C\%22\%E2\%88\%9Ai\%22\%5D\%5D\%7D}
\newcommand{\qtwentyfoursixtwo}{https://algassert.com/quirk\#circuit=\%7B\%22cols\%22\%3A\%5B\%5B\%22H\%22\%2C\%22H\%22\%2C\%22H\%22\%2C\%22H\%22\%2C\%22H\%22\%2C\%22H\%22\%2C\%22H\%22\%2C\%22H\%22\%5D\%2C\%5B1\%2C1\%2C1\%2C1\%2C1\%2C1\%2C\%22zpar\%22\%2C1\%2C\%22\%E2\%88\%9Ai\%22\%5D\%2C\%5B1\%2C1\%2C1\%2C1\%2C1\%2C1\%2C1\%2C\%22zpar\%22\%2C\%22\%E2\%88\%9Ai\%22\%5D\%2C\%5B1\%2C1\%2C1\%2C1\%2C1\%2C1\%2C\%22zpar\%22\%2C\%22zpar\%22\%2C\%22\%E2\%88\%9Ai\%22\%5D\%2C\%5B\%22zpar\%22\%2C1\%2C1\%2C1\%2C1\%2C1\%2C\%22zpar\%22\%2C1\%2C\%22\%E2\%88\%9Ai\%22\%5D\%2C\%5B\%22zpar\%22\%2C1\%2C1\%2C1\%2C1\%2C1\%2C1\%2C\%22zpar\%22\%2C\%22\%E2\%88\%9Ai\%22\%5D\%2C\%5B1\%2C\%22zpar\%22\%2C1\%2C1\%2C1\%2C1\%2C\%22zpar\%22\%2C1\%2C\%22\%E2\%88\%9Ai\%22\%5D\%2C\%5B1\%2C\%22zpar\%22\%2C1\%2C1\%2C1\%2C1\%2C1\%2C\%22zpar\%22\%2C\%22\%E2\%88\%9Ai\%22\%5D\%2C\%5B1\%2C1\%2C\%22zpar\%22\%2C1\%2C1\%2C1\%2C\%22zpar\%22\%2C1\%2C\%22\%E2\%88\%9Ai\%22\%5D\%2C\%5B1\%2C1\%2C\%22zpar\%22\%2C1\%2C1\%2C1\%2C1\%2C\%22zpar\%22\%2C\%22\%E2\%88\%9Ai\%22\%5D\%2C\%5B1\%2C1\%2C1\%2C\%22zpar\%22\%2C1\%2C1\%2C\%22zpar\%22\%2C1\%2C\%22\%E2\%88\%9Ai\%22\%5D\%2C\%5B1\%2C1\%2C1\%2C\%22zpar\%22\%2C1\%2C1\%2C1\%2C\%22zpar\%22\%2C\%22\%E2\%88\%9Ai\%22\%5D\%2C\%5B1\%2C1\%2C1\%2C1\%2C\%22zpar\%22\%2C1\%2C\%22zpar\%22\%2C1\%2C\%22\%E2\%88\%9Ai\%22\%5D\%2C\%5B1\%2C1\%2C1\%2C1\%2C\%22zpar\%22\%2C1\%2C1\%2C\%22zpar\%22\%2C\%22\%E2\%88\%9Ai\%22\%5D\%2C\%5B1\%2C1\%2C1\%2C1\%2C1\%2C\%22zpar\%22\%2C\%22zpar\%22\%2C1\%2C\%22\%E2\%88\%9Ai\%22\%5D\%2C\%5B1\%2C1\%2C1\%2C1\%2C1\%2C\%22zpar\%22\%2C1\%2C\%22zpar\%22\%2C\%22\%E2\%88\%9Ai\%22\%5D\%2C\%5B\%22zpar\%22\%2C1\%2C1\%2C1\%2C1\%2C1\%2C\%22zpar\%22\%2C\%22zpar\%22\%2C\%22\%E2\%88\%9Ai\%22\%5D\%2C\%5B1\%2C\%22zpar\%22\%2C1\%2C1\%2C1\%2C1\%2C\%22zpar\%22\%2C\%22zpar\%22\%2C\%22\%E2\%88\%9Ai\%22\%5D\%2C\%5B1\%2C1\%2C\%22zpar\%22\%2C1\%2C1\%2C1\%2C\%22zpar\%22\%2C\%22zpar\%22\%2C\%22\%E2\%88\%9Ai\%22\%5D\%2C\%5B1\%2C1\%2C1\%2C\%22zpar\%22\%2C1\%2C1\%2C\%22zpar\%22\%2C\%22zpar\%22\%2C\%22\%E2\%88\%9Ai\%22\%5D\%2C\%5B1\%2C1\%2C1\%2C1\%2C\%22zpar\%22\%2C1\%2C\%22zpar\%22\%2C\%22zpar\%22\%2C\%22\%E2\%88\%9Ai\%22\%5D\%2C\%5B1\%2C1\%2C1\%2C1\%2C1\%2C\%22zpar\%22\%2C\%22zpar\%22\%2C\%22zpar\%22\%2C\%22\%E2\%88\%9Ai\%22\%5D\%2C\%5B\%22zpar\%22\%2C\%22zpar\%22\%2C\%22zpar\%22\%2C\%22zpar\%22\%2C\%22zpar\%22\%2C\%22zpar\%22\%2C\%22zpar\%22\%2C1\%2C\%22\%E2\%88\%9Ai\%22\%5D\%2C\%5B\%22zpar\%22\%2C\%22zpar\%22\%2C\%22zpar\%22\%2C\%22zpar\%22\%2C\%22zpar\%22\%2C\%22zpar\%22\%2C1\%2C\%22zpar\%22\%2C\%22\%E2\%88\%9Ai\%22\%5D\%2C\%5B\%22zpar\%22\%2C\%22zpar\%22\%2C\%22zpar\%22\%2C\%22zpar\%22\%2C\%22zpar\%22\%2C\%22zpar\%22\%2C\%22zpar\%22\%2C\%22zpar\%22\%2C\%22\%E2\%88\%9Ai\%22\%5D\%2C\%5B1\%2C1\%2C1\%2C1\%2C1\%2C1\%2C\%22H\%22\%5D\%2C\%5B1\%2C1\%2C1\%2C1\%2C1\%2C1\%2C\%22\%7C0\%E2\%9F\%A9\%E2\%9F\%A80\%7C\%22\%5D\%2C\%5B1\%2C1\%2C1\%2C1\%2C1\%2C1\%2C1\%2C\%22H\%22\%5D\%2C\%5B1\%2C1\%2C1\%2C1\%2C1\%2C1\%2C1\%2C\%22\%7C0\%E2\%9F\%A9\%E2\%9F\%A80\%7C\%22\%5D\%5D\%7D}

\newcommand{\qeighteenfivetwo}{https://algassert.com/quirk\#circuit=\%7B\%22cols\%22\%3A\%5B\%5B\%22H\%22\%2C\%22H\%22\%2C\%22H\%22\%2C\%22H\%22\%2C\%22H\%22\%2C\%22H\%22\%2C\%22H\%22\%5D\%2C\%5B1\%2C1\%2C1\%2C1\%2C1\%2C\%22zpar\%22\%2C1\%2C\%22\%E2\%88\%9Ai\%22\%5D\%2C\%5B\%22zpar\%22\%2C1\%2C1\%2C1\%2C1\%2C1\%2C\%22zpar\%22\%2C\%22\%E2\%88\%9A\%22\%5D\%2C\%5B1\%2C\%22zpar\%22\%2C1\%2C1\%2C1\%2C1\%2C\%22zpar\%22\%2C\%22\%E2\%88\%9A\%22\%5D\%2C\%5B1\%2C1\%2C\%22zpar\%22\%2C1\%2C1\%2C1\%2C\%22zpar\%22\%2C\%22\%E2\%88\%9A\%22\%5D\%2C\%5B1\%2C1\%2C1\%2C\%22zpar\%22\%2C1\%2C\%22zpar\%22\%2C1\%2C\%22\%E2\%88\%9Ai\%22\%5D\%2C\%5B1\%2C1\%2C1\%2C\%22zpar\%22\%2C1\%2C1\%2C\%22zpar\%22\%2C\%22\%E2\%88\%9Ai\%22\%5D\%2C\%5B1\%2C1\%2C1\%2C1\%2C\%22zpar\%22\%2C\%22zpar\%22\%2C1\%2C\%22\%E2\%88\%9Ai\%22\%5D\%2C\%5B1\%2C1\%2C1\%2C1\%2C\%22zpar\%22\%2C1\%2C\%22zpar\%22\%2C\%22\%E2\%88\%9Ai\%22\%5D\%2C\%5B1\%2C1\%2C1\%2C1\%2C1\%2C\%22zpar\%22\%2C\%22zpar\%22\%2C\%22\%E2\%88\%9Ai\%22\%5D\%2C\%5B\%22zpar\%22\%2C\%22zpar\%22\%2C1\%2C1\%2C1\%2C1\%2C\%22zpar\%22\%2C\%22\%E2\%88\%9A\%22\%5D\%2C\%5B\%22zpar\%22\%2C1\%2C\%22zpar\%22\%2C1\%2C1\%2C1\%2C\%22zpar\%22\%2C\%22\%E2\%88\%9A\%22\%5D\%2C\%5B1\%2C\%22zpar\%22\%2C\%22zpar\%22\%2C1\%2C1\%2C1\%2C\%22zpar\%22\%2C\%22\%E2\%88\%9A\%22\%5D\%2C\%5B1\%2C1\%2C1\%2C\%22zpar\%22\%2C\%22zpar\%22\%2C\%22zpar\%22\%2C1\%2C\%22\%E2\%88\%9Ai\%22\%5D\%2C\%5B1\%2C1\%2C1\%2C\%22zpar\%22\%2C\%22zpar\%22\%2C1\%2C\%22zpar\%22\%2C\%22\%E2\%88\%9Ai\%22\%5D\%2C\%5B1\%2C1\%2C1\%2C\%22zpar\%22\%2C1\%2C\%22zpar\%22\%2C\%22zpar\%22\%2C\%22\%E2\%88\%9Ai\%22\%5D\%2C\%5B1\%2C1\%2C1\%2C1\%2C\%22zpar\%22\%2C\%22zpar\%22\%2C\%22zpar\%22\%2C\%22\%E2\%88\%9Ai\%22\%5D\%2C\%5B\%22zpar\%22\%2C\%22zpar\%22\%2C\%22zpar\%22\%2C1\%2C1\%2C1\%2C\%22zpar\%22\%2C\%22\%E2\%88\%9A\%22\%5D\%2C\%5B1\%2C1\%2C1\%2C\%22zpar\%22\%2C\%22zpar\%22\%2C\%22zpar\%22\%2C\%22zpar\%22\%2C\%22\%E2\%88\%9Ai\%22\%5D\%2C\%5B1\%2C1\%2C1\%2C1\%2C1\%2C\%22H\%22\%5D\%2C\%5B1\%2C1\%2C1\%2C1\%2C1\%2C\%22\%7C0\%E2\%9F\%A9\%E2\%9F\%A80\%7C\%22\%5D\%2C\%5B1\%2C1\%2C1\%2C1\%2C1\%2C1\%2C\%22H\%22\%5D\%2C\%5B1\%2C1\%2C1\%2C1\%2C1\%2C1\%2C\%22\%7C0\%E2\%9F\%A9\%E2\%9F\%A80\%7C\%22\%5D\%2C\%5B1\%2C1\%2C1\%2C\%22zpar\%22\%2C\%22zpar\%22\%2C1\%2C1\%2C\%22\%E2\%88\%9Ai\%22\%5D\%2C\%5B\%22zpar\%22\%2C\%22zpar\%22\%2C\%22zpar\%22\%2C1\%2C1\%2C1\%2C1\%2C\%22\%E2\%88\%9Ai\%22\%5D\%2C\%5B\%22zpar\%22\%2C\%22zpar\%22\%2C1\%2C1\%2C1\%2C1\%2C1\%2C\%22\%E2\%88\%9Ai\%22\%5D\%2C\%5B\%22zpar\%22\%2C1\%2C\%22zpar\%22\%2C1\%2C1\%2C1\%2C1\%2C\%22\%E2\%88\%9Ai\%22\%5D\%2C\%5B1\%2C\%22zpar\%22\%2C\%22zpar\%22\%2C1\%2C1\%2C1\%2C1\%2C\%22\%E2\%88\%9Ai\%22\%5D\%5D\%7D}
\newcommand{\qtwentyfourtwo}{https://algassert.com/quirk\#circuit=\%7B\%22cols\%22\%3A\%5B\%5B\%22H\%22\%2C\%22H\%22\%2C\%22H\%22\%2C\%22H\%22\%2C\%22H\%22\%2C\%22H\%22\%2C\%22H\%22\%2C1\%5D\%2C\%5B1\%2C1\%2C1\%2C1\%2C\%22zpar\%22\%2C1\%2C1\%2C\%22\%5Cu221a-i\%22\%5D\%2C\%5B1\%2C1\%2C1\%2C1\%2C1\%2C\%22zpar\%22\%2C1\%2C\%22\%5Cu221a-i\%22\%5D\%2C\%5B1\%2C1\%2C1\%2C1\%2C1\%2C1\%2C\%22zpar\%22\%2C\%22\%5Cu221a-i\%22\%5D\%2C\%5B1\%2C1\%2C1\%2C1\%2C\%22zpar\%22\%2C\%22zpar\%22\%2C\%22zpar\%22\%2C\%22\%5Cu221a-i\%22\%5D\%2C\%5B\%22zpar\%22\%2C1\%2C1\%2C1\%2C\%22zpar\%22\%2C\%22zpar\%22\%2C1\%2C\%22\%5Cu221ai\%22\%5D\%2C\%5B\%22zpar\%22\%2C1\%2C1\%2C1\%2C\%22zpar\%22\%2C1\%2C\%22zpar\%22\%2C\%22\%5Cu221ai\%22\%5D\%2C\%5B\%22zpar\%22\%2C1\%2C1\%2C1\%2C1\%2C\%22zpar\%22\%2C\%22zpar\%22\%2C\%22\%5Cu221ai\%22\%5D\%2C\%5B1\%2C\%22zpar\%22\%2C1\%2C1\%2C\%22zpar\%22\%2C\%22zpar\%22\%2C1\%2C\%22\%5Cu221ai\%22\%5D\%2C\%5B1\%2C\%22zpar\%22\%2C1\%2C1\%2C\%22zpar\%22\%2C1\%2C\%22zpar\%22\%2C\%22\%5Cu221ai\%22\%5D\%2C\%5B1\%2C\%22zpar\%22\%2C1\%2C1\%2C1\%2C\%22zpar\%22\%2C\%22zpar\%22\%2C\%22\%5Cu221ai\%22\%5D\%2C\%5B1\%2C1\%2C\%22zpar\%22\%2C1\%2C\%22zpar\%22\%2C\%22zpar\%22\%2C1\%2C\%22\%5Cu221ai\%22\%5D\%2C\%5B1\%2C1\%2C\%22zpar\%22\%2C1\%2C\%22zpar\%22\%2C1\%2C\%22zpar\%22\%2C\%22\%5Cu221ai\%22\%5D\%2C\%5B1\%2C1\%2C\%22zpar\%22\%2C1\%2C1\%2C\%22zpar\%22\%2C\%22zpar\%22\%2C\%22\%5Cu221ai\%22\%5D\%2C\%5B1\%2C1\%2C1\%2C\%22zpar\%22\%2C\%22zpar\%22\%2C\%22zpar\%22\%2C1\%2C\%22\%5Cu221ai\%22\%5D\%2C\%5B1\%2C1\%2C1\%2C\%22zpar\%22\%2C\%22zpar\%22\%2C1\%2C\%22zpar\%22\%2C\%22\%5Cu221ai\%22\%5D\%2C\%5B1\%2C1\%2C1\%2C\%22zpar\%22\%2C1\%2C\%22zpar\%22\%2C\%22zpar\%22\%2C\%22\%5Cu221ai\%22\%5D\%2C\%5B\%22zpar\%22\%2C\%22zpar\%22\%2C\%22zpar\%22\%2C\%22zpar\%22\%2C\%22zpar\%22\%2C1\%2C1\%2C\%22\%5Cu221ai\%22\%5D\%2C\%5B\%22zpar\%22\%2C\%22zpar\%22\%2C\%22zpar\%22\%2C\%22zpar\%22\%2C1\%2C\%22zpar\%22\%2C1\%2C\%22\%5Cu221ai\%22\%5D\%2C\%5B\%22zpar\%22\%2C\%22zpar\%22\%2C\%22zpar\%22\%2C\%22zpar\%22\%2C1\%2C1\%2C\%22zpar\%22\%2C\%22\%5Cu221ai\%22\%5D\%2C\%5B\%22zpar\%22\%2C\%22zpar\%22\%2C\%22zpar\%22\%2C\%22zpar\%22\%2C\%22zpar\%22\%2C\%22zpar\%22\%2C\%22zpar\%22\%2C\%22\%5Cu221ai\%22\%5D\%2C\%5B1\%2C1\%2C1\%2C1\%2C\%22H\%22\%2C1\%2C1\%2C1\%5D\%2C\%5B1\%2C1\%2C1\%2C1\%2C\%22\%7C0\%5Cu27e9\%5Cu27e80\%7C\%22\%2C1\%2C1\%2C1\%5D\%2C\%5B1\%2C1\%2C1\%2C1\%2C1\%2C\%22H\%22\%2C1\%2C1\%5D\%2C\%5B1\%2C1\%2C1\%2C1\%2C1\%2C\%22\%7C0\%5Cu27e9\%5Cu27e80\%7C\%22\%2C1\%2C1\%5D\%2C\%5B1\%2C1\%2C1\%2C1\%2C1\%2C1\%2C\%22H\%22\%2C1\%5D\%2C\%5B1\%2C1\%2C1\%2C1\%2C1\%2C1\%2C\%22\%7C0\%5Cu27e9\%5Cu27e80\%7C\%22\%2C1\%5D\%5D\%7D
}

\newcommand{\quirk}[2]{\href{#1}{\textsc{#2}}}

\newcommand{\CZ}{\mathrm{CZ}}
\newcommand{\CS}{\mathrm{CS}}
\newcommand{\CCZ}{\mathrm{CCZ}}
\newcommand{\CkZ}[1]{\mathrm{C}^{#1}\mathrm{Z}}

\newcommand\IBM{%
    IBM Quantum,
    IBM T.~J. Watson Research Center,
    Yorktown Heights, NY 10598, USA}

\begin{document}

\title{Borrowed Identities: Malleable Distillation Factories and a Unified Numerical Search}

\author{Shraddha Singh}
\thanks{Present address: \IBM}
\thanks{Corresponding address: shraggygkp@gmail.com}
\affiliation{Google Quantum AI}
\affiliation{Department of Applied Physics, Yale University,
             New Haven, Connecticut 06511, USA}
\affiliation{Yale Quantum Institute, Yale University,
             New Haven, Connecticut 06511, USA}
\author{Craig Gidney}
\affiliation{Google Quantum AI}
\author{Cody Jones}
\affiliation{Google Quantum AI}

%% ---------------------------------------------------------------
%%  Abstract
%% ---------------------------------------------------------------
\begin{abstract}
Magic-state distillation is one of the leading overheads in fault-tolerant
quantum computation. Existing methods for finding distillation factories require
a transversal gate to act correctly on the entire codespace, a constraint that
limits both generality and search efficiency. We introduce a strictly weaker
\emph{borrowed-identity} condition, requiring only that the distillation circuit
act as the identity on a single input state. It applies uniformly across all
levels of the Clifford hierarchy and unifies, within a single level, factories
that distill different magic states---for example, the $\ket{T}$, $\ket{\CS}$,
and $\ket{\CCZ}$ factories. The condition is closed-form for symmetric
circuits---recovering the quantum Reed--Muller family at every level---and linear
for two-group circuits, so a brute-force search runs in seconds and recovers,
within its range, all distance-2 factories known from code-construction
approaches, including entangled-output and multi-output factories previously
outside the scope of any single numerical search. The same circuits are
\emph{malleable}: one parent encodes many factories, so the output magic-state
type is chosen at compile time rather than fixed by a hard-coded design. The
framework also reaches beyond CSS codes to synthillation and non-CSS catalytic
factories.
\end{abstract}
\maketitle

%% ===============================================================
%% ===============================================================
%% ===============================================================

Fault-tolerant quantum computation (FTQC) requires a universal gate set, the Clifford group with a non-Clifford gate~\cite{bravyi2005universal}. Yet, most practical quantum error-correcting codes admit a constant-depth fault-tolerant implementation only for the  Clifford group~\cite{eastin2009restrictions}. The standard route to universality is gate teleportation~\cite{gottesman1999demonstrating}: a non-Clifford gate
$U$ is applied by consuming a magic state $\ket{U}$. For the
$T=e^{-i\pi Z/8}$ gate, the relevant resource is
$\ket{T}=T\ket{+}=\ket{0}+e^{i\pi/4}\ket{1}$ (unnormalized).
Magic-state distillation (MSD)~\cite{bravyi2005universal,rodriguez2025experimental} purifies
many low-fidelity logical copies of this state, prepared by injection~\cite{singh2022high,gidney2023cleaner,liu2026situ,li2015magic,gupta2024encoding,zhang2025constant} or cultivation techniques~\cite{gidney2024magic,jacoby2025magic,sahay2026fold,rosenfeld2025magic}, into fewer high-fidelity ones using
only Clifford operations, with overhead
\begin{equation}
  O\!\left(\log^\gamma(1/\epsilon)\right),\quad
  \gamma=\log_d\!\left(\tfrac{N}{k}\right),
  \label{eq:overhead}
\end{equation}
where $N$ inputs yield $k$ outputs at distance $d$ and target error
rate $\epsilon$. Low $\gamma$ is a desirable property for reducing FTQC overhead. Despite advances in high-quality magic-state preparation~\cite{singh2022high,gidney2024magic,jacoby2025magic,sahay2026fold,rosenfeld2025magic}, MSD remains a crucial component of projected fault-tolerant architectures~\cite{gidney2025factor,zhou2025resource}.

Prior constructions of distillation factories with practically relevant $[[N,k,d]]$ parameters broadly operate in the \emph{Heisenberg picture},
demanding that a transversal gate act correctly on every logical
state, the constraint of generalized triorthogonality~\cite{bravyi2012magic,rengaswamy2020optimality,haah2018codes,bombin2007exact,eastin2013distilling,krishna2019towards,wills2024constant} or quasitransversality~\cite{campbell2017unified},
making the corresponding searches non-linear matrix problems barring limited exceptions~\cite{meier2012magic,jones2013multilevel}. A complementary perspective on distillation was offered by
Litinski~\cite{litinski2019magic}, who observed that the factories
$[[15,1,3]]$, $[[14,2,2]]$, and $[[20,4,2]]$ can be represented via spacetime dual of the respective transversal implementation on CSS codes with a common circuit
structure: they are all derived from identity circuits on $\ket{+}^{\otimes n}$
built from multi-qubit $Z$-rotations. These circuit examples were not given an algebraic formulation and no systematic search was built
on it.

In this Letter, we formalize this observation of Ref.~\cite{litinski2019magic} as a
\emph{borrowed-identity condition} in the Schr\"odinger picture,
asking only that the circuit act as identity on a specific input
state through the collective phase action of its gates, with no requirement that it act as identity on the full codespace -- the identity is borrowed from the input state. A distillation factory is then obtained from a borrowed identity on $n$
qubits using $N$ gates by removing all gates acting solely on $k$ qubits: the $k$
affected qubits are left in a desired magic state $\ket{\theta}$ and serve as distilled outputs, while the remaining
$n - k$ qubits are measured in the $X$ basis for error detection. On conventions: we use $n$ for the circuit's qubit count and $N$ for its gate count---the number of $X$-stabilizers ($n-k$) and input magic states ($N$), respectively, in the spacetime-dual CSS code~\footnote{most of the literature uses $n$ for our $N$}. Under this spacetime rotation---exchanging qubits with stabilizer ancillae and phase rotations with magic-state inputs---the $n$-qubit borrowed identity of $N$ multi-qubit $Z$-rotations becomes an $N$-qubit CSS code with $n-k$ $X$-stabilizers on which $T^{\otimes N}$ is a logical Clifford-hierarchy gate, and $\mathcal{C}\ket{+}^{\otimes n}=e^{i\phi}\ket{+}^{\otimes n}$ becomes all $X$-stabilizer ancillae returning to $\ket{+}$.

We derive the conditions under which a borrowed identity serves as a distillation factory; unlike all prior frameworks, they are
parameterized by $\theta = \pi/2^l$ and hold at every level $l$ of the Clifford hierarchy without modification: the same algebraic
machinery searches for $S$-gate factories ($l=2$), $T$-gate
factories ($l=3$), and higher-level rotation factories ($l \ge 4$)
simultaneously. The relaxation has three concrete payoffs. First, a closed-form
master equation (Theorem~\ref{thm:main}) recovers the quantum Reed--Muller
code~\cite{PhysRevA.54.1862} $[[2^{l+1}-1, 1, 3]]$
analytically at every $l$. Second, extending the construction to circuits with entangled outputs (like $\ket{\CZ}$, or $\ket{\CS}$ and $\ket{\CCZ}$ magic states) recovers the
generalized triorthogonal family~\cite{haah2018codes,jones2013low}, for example, the hypercube quantum code family $[[2^l,l,2]]$~\cite{11096910,webster2022xp} such as $[[8,3,2]]$ ($T$-to-CCZ) and $[[16,4,2]]$ ($T$-to-CCCZ); and the $T$-to-CS factory $[[12,2,2]]$~\cite{webster2023transversal} (equivalent to the $4\,\CS\!\to\!1\,\CS$ factory derived from AG codes~\cite{anqi2026qec}), providing the first systematic search over these factories for arbitrary $l$. A two-group asymmetric circuit yields a phase condition \emph{linear over $\mathbb{Z}_{2^{l+1}}$} (Theorem~\ref{thm:asymmetric-master}), with verification costing $O(k(n-k))$ phase computations, growing only polynomially in $n$ and $k$. Third, we show malleable circuits for arbitrary $l$ that unify factories distilling different magic states at the same Clifford-hierarchy level (for example, the $T$-to-$T$, $T$-to-CS, and $T$-to-CCZ factories at $l=3$).

An unoptimized brute-force Python implementation on the asymmetric circuit search conditions, run on a standard
laptop (Apple M4, 16 GB RAM), recovers in $\sim 9$ seconds every distance-$2$ factory within
the swept range $l \in \{2,3,4\}$, $k \le 7$, $n \le 11$, and gate-weight parameters $s_\mathrm{total}, s_O \le 7$ (21{,}920 valid factories in total)~\footnote{The count of valid borrowed-identity solutions---admissible (parameter-tuple, sign) configurations, including trivial (stabilizer) outputs and counting each factory once per realizing $(n,s)$ tuple---so it exceeds the number of distinct factories.}. Unlike traditional distillation conditions parametrized by $N$, our search is parametrized by $n$, which stays small even for $[[N,k,d]]$ factories with large $N$. The result includes well-known factories that no prior single framework captures simultaneously: the H-code ($l=2$) and Bravyi--Haah ($l=3$) families, all distance-$2$ entries of the Nezami--Haah triorthogonal-code catalogue~\cite{nezami2022classification} at the largest $k$ for each $N$, and the small distance-$2$ AG-codes~\cite{anqi2026qec} (see Fig.~\ref{fig:search-results}). The $k \le 7$ ceiling is
a property of the swept range and not of the framework: since the
runtime scales only polynomially in $n$
and $k$, any distance-$2$ factory at larger $k$ is reachable
by widening the sweep. 

%% ===============================================================
%% ===============================================================
\vspace{10pt}
\paragraph*{Framework---}
%\section{Framework}
\label{sec:framework}
For $\theta=\pi/2^l$, $l\in\mathbb{Z}_{\ge 1}$, the \emph{magic
state} is $\ket{\theta}\coloneqq\ket{0}+e^{i2\theta}\ket{1}$
(unnormalized). The weight-$w$ phase rotations to be used in the circuit are 
\begin{equation}
Z^{\otimes w}(\theta)\coloneqq
  \exp\!\left(i\theta(I-\bigotimes_{j=1}^{w}Z_j)\right), 
  \label{eq:gate_def}
\end{equation}
 which apply phase $e^{i2\theta}$ to basis
states with an odd number of $1$s among the $w$
designated qubits. The gates $Z^{\otimes w}(\pi/2^l)$ are diagonal elements of the
$l$-th level of the Clifford hierarchy
$\mathbb{C}_l$~\cite{gottesman1999demonstrating}, and the
diagonal subgroup of $\mathbb{C}_l$ is closed under products~\cite{cui2017diagonal}.

\begin{definition}[Borrowed identity]
\label{def:borrowed}
A circuit $\mathcal{C}$ on $n$ qubits, built from parity phase
gates with $\theta = \pi/2^l$, is a \emph{borrowed identity} if
$\mathcal{C}\ket{+}^{\otimes n} = e^{i\phi}\ket{+}^{\otimes n}$
for some global phase $\phi$. This framework formalizes the construction of Ref.~\cite{litinski2019magic} to re-envision distillation factories in the Schr\"odinger picture, where only the state of the qubits is tracked, unlike the Heisenberg picture (see App.~\ref{app:automorphisms}), where the transversality of an operation is tracked. Because this circuit is built entirely from gates in the diagonal
$\mathbb{C}_l$ subgroup, the
$k$-qubit output state of the factory, obtained by removing gates from the borrowed identity that act solely on the output qubits, is necessarily a resource
state that teleports a diagonal $\mathcal{C}_l$ gate using conditional operations from $\mathcal{C}_{l-1}$. A distance-2 factory is guaranteed with such a borrowed identity construction since all phase rotations have dependency on a non-output (check) qubit which will detect any single error in the circuit upon measurement in the $X$ basis.
\end{definition} 

Our framework changes the search from generalized triorthogonal codes to borrowed-identity circuits, treated below in three settings: symmetric (closed-form conditions), two-group (brute-force search), and malleable circuits.
 
\paragraph{Symmetric circuits---}
%\section{Symmetric circuits}
\label{sec:symmetric}
When a circuit is invariant under the full symmetric group $S_n$
on the $n$ qubits~\footnote{$S_n$ acts by relabeling qubit registers, not by physically permuting hardware: qubits with identical circuit connectivity are interchangeable for the borrowed-identity condition.}, the gate set
is specified by allowed weights $\mathcal{W}\subseteq\{1,\ldots,n\}$
with all $\binom{n}{w}$ gates of each $w\in\mathcal{W}$ included.
States of equal Hamming weight $i$ accumulate identical phase, so it
suffices to track the Dicke sector
$\ket{\binom{n}{i}}\coloneqq\sum_{|\mathbf{x}|=i}\ket{\mathbf{x}}$.

\begin{lemma}[Dicke-sector phase accumulation]
\label{lem:dicke}
For the above symmetric circuit, every element of $\ket{\binom{n}{i}}$
accumulates the total phase
\begin{equation}
  \Phi(i)\equiv 2^i\theta\sum_{\substack{w\in\mathcal{W}\\w\ge i}}
  \binom{n-i}{w-i}\pmod{2\pi}.
  \label{eq:dicke_phase}
\end{equation}
\end{lemma}

The full inductive proof is given in App.~\ref{app:dicke_proof}.

\begin{theorem}[Symmetric master condition]
\label{thm:main}
A symmetric circuit with gate set $\mathcal{W}$ and $\binom{n}{w}$ gates of each weight is a borrowed identity if and only if
\begin{equation}
  \sum_{\substack{w\in\mathcal{W}\\w\ge i}}\binom{n-i}{w-i}
  \in 2^{l-i+1}\mathbb{Z}
  \quad\text{for all }i\in\{1,\ldots,l\}.
  \label{eq:master-thm}
\end{equation}
For $i>l$, the prefactor in Eq.~(\ref{eq:dicke_phase}) is $2^i\theta= 2m\pi$ where $m$ is an integer; the condition is
satisfied automatically. The distance of the code is determined by the minimum number ($d$) of gates whose combination has an odd overlap on not more than $k$ qubits.
\end{theorem}
For $\mathcal{W}=\{w:w\equiv 1\pmod{s}\}$, increasing $s$, the weight separation parameter, straightforwardly increases the output qubits to yield $[[N-s,s,2]]$, for every choice of $s$, where $N$ is the total number of gates in the borrowed-identity for an admissible $\mathcal{W}$ satisfying Eq.~\eqref{eq:master-thm}. For example, when $s\in\{1,2\}$,
Eq.~\eqref{eq:master-thm} admits closed-form solutions. Taking
$\mathcal{W}=\{1,2,\ldots,n\}$ ($s=1$) gives $n_{\min}=l+1$ with $N=2^{n_{\min}}-1$ yielding the
code parameters $[[2^{l+1}-2,1,2]]$; taking odd weights~\footnote{Even weights cannot yield $k=1$ factories} only ($s=2$) gives $n_{\min}=l+2$ while the total number of gates is given by $2^{n_{\min}-1}$. This circuit allows removing two weight-$1$ gates from the circuit to yield $[[2^{l+1}-2,2,2]]$, recovering the smallest Bravyi--Haah code $[[14,2,2]]$ for $l=3$ and the smallest H-code  $[[6,2,2]]$~\cite{jones2013multilevel} for $l=2$ (the $l=1$ member $[[2,2,2]]$ is degenerate; the $[[4,2,2]]$ Iceberg code~\cite{PhysRevA.54.4741} instead appears as the $l=2$, degree-$2$ two-group member, see Table~\ref{tab:canonical}).
The solution for $s=2$ can be extended to yield
$[[2^{l+1}-1,1,3]]$ if we only remove a single weight-1 gate, recovering the Steane
code $[[7,1,3]]$~\cite{steane1996multiple} at $l=2$ and the
$[[15,1,3]]$ distillation factory~\cite{bravyi2005universal} at $l=3$. Factories with $s\ge 3$ have no closed-form solution and require numerical search ($k\ge3$, $d=2$ fall here); increasing $s$ does not raise the distance, which saturates at $d=3$ (App.~\ref{app:dicke_proof}). The $s=2$ result thus gives distance-$3$ distillation at every Clifford level---the quantum Reed--Muller family~\cite{PhysRevA.54.1862}, here derived without reference to any codespace~\cite{gong2024computation,luo2020fault,kubica2015universal}. The more efficient entangled-output and higher-rate factories, however, require relaxing the full $S_n$ symmetry.

\paragraph{Two-group symmetric circuits---}
%\section{Two-group symmetric circuits}
\label{sec:asymmetric}
\begin{figure*}[!htbp]
\centering
\includegraphics[width=\textwidth]{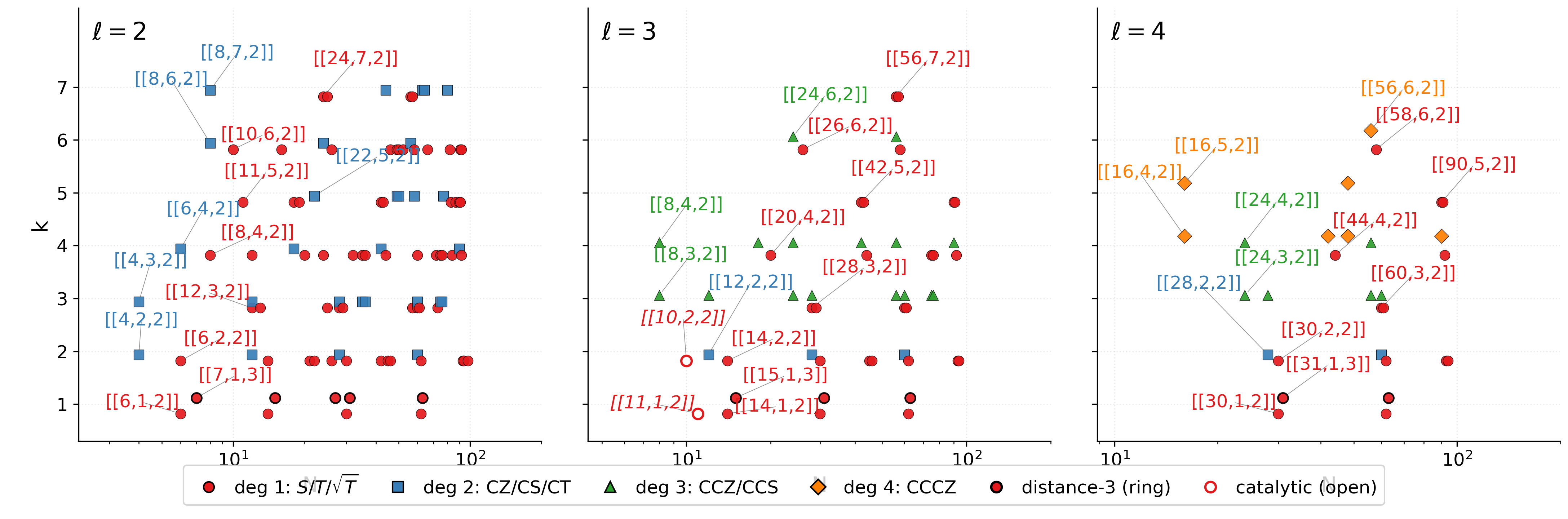}
\caption{\textbf{Distillation factories recovered by the two-group search.}
Each marker is a factory at $(N,k)$: $N$ magic states of type $Z^{\otimes w}(\pi/2^l)$
yielding a $k$-qubit output, at Clifford-hierarchy levels $l\in\{2,3,4\}$ (three
panels). Marker shape and color encode the output type as listed in the legend (see App.~\ref{app:output-extraction} for output classification). All markers are distance-$2$ except the
bold-ringed $k=1$ markers, which are distance-$3$. Filled markers are produced by
the two-group search over the swept range $k\le7$, $n\le11$,
$s_\mathrm{total},s_O\le7$ (App.~\ref{app:asymmetric_search}); open markers are the non-CSS factories, which the search does not produce but fall under the borrowed identity framework. In each column the top
panel shows $N\le100$ with the canonical (smallest-$N$) representative annotated
per degree and $k$; the full extent ($N>1000$) is shown in Fig.~\ref{fig:search-ext}, with parameter tuples for $N\le30$ factories listed in Tables~\ref{tab:asymmetric-l2}--\ref{tab:asymmetric-l4} for explicit circuit construction.}
\label{fig:search-results}
\end{figure*}

%% ===============================================================
%In the previous section, we focused on the case of symmetric borrowed-indentity circuits where all qubits were interchangeable. Now, we 
To simplify search on asymmetric circuits, we partition the $n$ qubits into $k$ \emph{output} qubits
($\mathcal{O}$) and $n - k$ \emph{check} qubits ($\mathcal{S}$),
keeping qubits within each group interchangeable
($S_k \times S_{n - k}$ symmetry). We denote the sets of
admissible gate weights on each group by $\mathcal{W}_O$ and
$\mathcal{W}_S$, and the set of admissible total weights by
$\mathcal{W}_\mathrm{total}(\equiv \mathcal{W})$. The three weight-separation
parameters $(s_\mathrm{total}, s_O, s_S)$ for these sets,
\begin{equation}
    \mathcal{W}_i=\{w:w\equiv 1\pmod{s_i}\}\quad i\in\{O,S,\mathrm{total}\},
\end{equation}
can be
chosen independently. A gate is labeled by the pair
$(w_O, w_S)$ with $w_O \in \mathcal{W}_O$,
$w_S \in \mathcal{W}_S$,
$w_O + w_S \in \mathcal{W}_\mathrm{total}$. Each gate type
$(w_O, w_S) \in \mathcal{W}$ carries a sign
$\sigma_{w_O, w_S} \in \{+1, -1\}$: gates of type
$(w_O, w_S)$ implement
$Z^{\otimes (w_O + w_S)}(\sigma_{w_O, w_S}\,\theta)$, i.e., the
rotation angle is $+\theta$ when $\sigma_{w_O, w_S} = +1$ and
$-\theta$ when $\sigma_{w_O, w_S} = -1$. The sign assignment is
part of the circuit specification.

By the $S_k \times S_{n-k}$ symmetry, the phase accumulated by
a basis state depends only on the Dicke-sector indices
$(i_O, i_S)$: the Hamming weights on the output and check
blocks. We denote this phase $\Phi(i_O, i_S)$.

\begin{lemma}[Asymmetric Dicke-sector phase]
\label{lem:dicke-asym}
Assuming the borrowed-identity condition is satisfied at all
$(i_O', i_S')$ with $i_O' + i_S' < i_O + i_S$, the phase at
Dicke sector $(i_O, i_S)$ is
\begin{align}
  \Phi(i_O, i_S) \equiv\ &(-1)^{i_O + i_S + 1}\, 2^{i_O + i_S}\,\theta
  \!\!\!\!\sum_{(w_O, w_S) \in \mathcal{W}}\!\!\!\!
  \sigma_{w_O, w_S}\notag\\
  &\times \binom{k - i_O}{w_O - i_O}\binom{n - k - i_S}{w_S - i_S}
  \pmod{2\pi}.
  \label{eq:phase-asym}
\end{align}
\end{lemma}

Proof by induction on $i_O + i_S$ via Pascal's recurrence on
both blocks; see App.~\ref{app:lemma-asym}.

\begin{theorem}[Asymmetric master condition]
\label{thm:asymmetric-master}
A two-group circuit with gate types $\mathcal{W}$ and signs
$\{\sigma_{w_O, w_S}\} \subset \{+1, -1\}$ is a borrowed
identity if and only if
\begin{equation}
  \!\!\!\sum_{(w_O, w_S) \in \mathcal{W}}\!\!\!\!
  \sigma_{w_O, w_S}\,
  \binom{k - i_O}{w_O - i_O}\!\binom{n - k - i_S}{w_S - i_S}
  \in 2^{l - i_O - i_S + 1}\mathbb{Z}
  \label{eq:asymmetric-master}
\end{equation}
for all $(i_O, i_S)$ with $i_O + i_S \le l$. The condition is vacuous for other values of $i_O+i_S>l$ as for these values, Eq.~\eqref{eq:phase-asym} is automatically an integer multiple of $2\pi$.  
\end{theorem}

Equation~\eqref{eq:asymmetric-master} is \emph{linear in $\{\sigma_{w_O, w_S}\}$ over $\mathbb{Z}_{2^{l + 1}}$}, allowing the validity of any sign assignment to be verified in
$O(k(n - k))$ time per Dicke sector. The output magic state is extracted by removing all gates of
type $(w_O, 0)$ with
$w_O \in \mathcal{W}_O \cap \mathcal{W}_\mathrm{total}$. If the borrowed identity comprises
$N_\mathrm{circuit}$ gates, the resulting input magic-state
count is
\begin{equation}
  N = N_\mathrm{circuit}\;\;
  -\!\!\!\!\sum_{\substack{w \in \mathcal{W}_O
    \cap \mathcal{W}_\mathrm{total}}}\!\!\!\!
  \binom{k}{w}.
  \label{eq:N-asym}
\end{equation}
We present our search results, based on Eqs.~[\ref{eq:phase-asym}-\ref{eq:N-asym}], in Fig.~\ref{fig:search-results}, where the outputs are classified into Clifford equivalence classes~\cite{campbell2017unified} (see App.~\ref{app:output-extraction}): factories whose outputs differ by a Clifford circuit are grouped together~\footnote{e.g.\ $[[8, 4, 2]]$
at $l = 3$ is a Clifford-padded version of the cube
$[[8, 3, 2]]$; it is retained as a distinct CCZ-class factory, since the extra qubit yields additional correlated-error detection (Fig.~\ref{fig:five-six-qubit-circuit}(b))}. Details, including the deterministic sign assignment, are in App.~\ref{app:asymmetric_search}. We fix $s_S=1$ for which no search over signs was needed. Since every gate surviving the $(w_O,0)$ removal touches a check qubit, single errors are always caught, guaranteeing distance $2$; choosing $s_O=s_S$ reproduces the distance-$2$ and $3$ symmetric families. Some multi-output factory examples are worked out in
App.~\ref{app:examples}. Reaching arbitrary distance $d$ at $k>1$ requires breaking the check-qubit symmetry so that the smallest set of gates with cancelling check-supports has size $d$; we leave this to future work.

The full sweep---runtime, scaling, and per-level catalogue---is detailed in
App.~\ref{app:asymmetric_search}. By contrast, the Nezami--Haah classification~\cite{nezami2022classification}
enumerates affine equivalence classes of Reed--Muller polynomials, a space
growing exponentially in $N$ (${\sim}5\times10^{12}$ candidates at $N\le 38$)
and specialized to $l=3$. Within one search space we recover, at $k>1$, both unentangled $T$-output
factories and entangled $\CS$/$\CCZ$ factories, families that earlier searches
could reach only one at a time. Fig.~\ref{fig:search-results} shows the density of factories over $(N,k)$ for
$N\le 100$ of the $21,920$ factories recovered; Tables~\ref{tab:asymmetric-l2}--\ref{tab:asymmetric-l4} list $N\le 30$
representatives, and App.~\ref{app:asymmetric_search} details how the families of
Fig.~\ref{fig:search-results} (Bravyi--Haah, H-code/Iceberg, generalized
triorthogonal, and distance-$2$ Nezami--Haah)
are recovered. Notably, the search finds the $[[56,6,2]]$ and $[[12,2,2]]$ factories: a
$6$-qubit state Clifford-equivalent to two $\CCZ$ states, and a $2$-qubit
state Clifford-equivalent to one $\CS$ state, respectively. These match the
$\CCZ(8)\!\to\!\CCZ(2)$~\cite{chamberland2022building} and $\CS(4)\!\to\!\CS(1)$ factories derived from
algebraic-geometry (AG) codes~\cite{anqi2026qec}.

We now turn to factories unreachable by the two-group symmetric search at $l=3$,
to show how the framework can be extended to uncover all distance-$2$ distillation circuits. The Campbell--Howard synthillation
families~\cite{campbell2017unified}---the T-to-CCZ family $[[6m+2,3m,2]]$ ($[[14,6,2]]$,
\quirk{\qfourteensixtwo}{Quirk}) and our newly found T-to-CS family $[[6m+6,2m,2]]$
($[[18,4,2]]$, \quirk{\qeighteenfourtwo}{Quirk})---both break the two-group
symmetry: their gate sets are not symmetric over the outputs. The T-to-CCZ
family is realized with pure $\pi/8$ rotations; the T-to-CS family additionally
uses  multi-weight phase rotations of type ($i
\le l$) $c\pi/2^{i}$ rotations ($c\in \mathbb{Z}_{2^i}$)---the \emph{extended}
borrowed-identity formalism which is not restricted to $i=l$---to reach its optimal $T$-count. The two Quirk circuits linked above make clear that these factories require breaking the two-group symmetry constraint. As an existence proof, we built
a symmetry-free, targeted-output search that solves the extended borrowed-identity
condition directly for a chosen output at levels $l\in\{2,3,4\}$ (see App.~\ref{app:symfree}). In addition to the above families, it also finds factories with mixed outputs, like $[[18,5,2]]$ (\quirk{\qeighteenfivetwo}{Quirk}) which yields one CS and one CCZ (better than combining $[[8,3,2]]$ and $[[12,2,2]]$ factories), not found explicitly in the synthillation literature, so the borrowed-identity picture handles mixed-output
factories too. 

The ultimate goal for this formalism is a
\emph{tractable} search that captures the smallest factory at each distance for different types of magic states. This is motivated by the growing use of the $[[8,3,2]]$ factory---which has the smallest number of input qubits---in recent resource estimates of fault-tolerant quantum computing~\cite{zhou2025resource,gidney2025factor}. Note that, $[[8,3,2]]$ cannot be converted into a non-hybrid $\CCZ$-to-$\CCZ$ factory unlike the $[[56,6,2]]$ or $[[12,2,2]]$ factories. The payoff of searching across types is that factory size depends strongly on both the output and the input: distilling input $T$ states into $\CCZ$ and $\CS$ outputs can admit smaller factories than $T$ (with non-CSS constructions smaller still), while distilling from higher-level inputs such as $\CS$-to-$\CS$ or $\CCZ$-to-$\CCZ$ costs more than the corresponding hybrid factory. At distance $2$, our two-group search already covers the smallest factory of each output type, while the larger batched synthillation variants are reached by the symmetry-free search; higher distance would require breaking the symmetry (below Eq.~\eqref{eq:N-asym}).%The ultimate goal, beyond this existence proof, is a \emph{tractable} search that captures the smallest factory at each distance for different types of magic states. At distance-$2$ these factories are already covered by our two-group symmetric search. For higher-distance we would need to break the symmetry, as discussed below Eq.~\ref{eq:N-asym}.

%===============================================================

\paragraph{Malleable circuits---}
%\section{Malleable circuits}
\label{sec:malleability}
Restricting the previous construction to $|\mathcal{S}| = 1$, we can
analytically derive that $k = n-1$ yields
$[[2^{n} - 2^{n-1}, n-1, 2]] = [[2^{n-1}, n-1, 2]]$ factories for
$s_O=s_S=s_\mathrm{total} = 1$. This is equivalent to constructing a partial borrowed
identity on the output qubits $\mathcal{O}$ — with each gate additionally
conditioned on the single check qubit in $\mathcal{S}$ — leaving
the check qubit unentangled at the end. In this section we discuss the
effects of applying this procedure \emph{sequentially}: at each step $j$, the
unentangled check qubit from the previous step is eliminated, and a new qubit from
$\mathcal{O}_{j-1}$ is promoted to the single check qubit in
$\mathcal{S}_j$. After $j = n$ steps, the complete parent circuit reduces to the symmetric borrowed identity at level $l$ with $s=s_{\mathrm{total}}={1,2}$, i.e., a
solution of Theorem~\ref{thm:main}. Specifically, after sequence step
$j$, the resulting factory is $[[2^{n-s+1}-2^{n-j-s+1}, n-j, 2]]$, with the $k=1$ endpoint ($j=n-1$) at distance $3$ for $s=2$.

For $n_{\min} = l+1$, the minimum $n$ for an $s=1$ symmetric solution, the sequential construction reveals the output magic-state type directly. After step $j$, the factory $[[2^{l+1} - 2^{l+1-j}, l+1-j, 2]]$ has an output state whose unitary preparation requires the same number of $T$ gates as the number of gates removed from the $n_{\min}$ symmetric borrowed identity at $s=1$, namely $2^{l+1-j} - 1$. These factories are already recovered by the asymmetric search; we highlight them to illustrate \emph{malleability} (for example, at $l=2$ the chain $[[4,2,2]]\to[[6,1,2]]$ gives outputs $\ket{\CZ}$ then $\ket{S}$).

\begin{proposition}[Factory malleability]
\label{prop:malleable}
A single borrowed identity circuit on $n$ qubits encodes multiple distinct distillation factories. The output magic-state type is determined by which gates are removed, not by the circuit itself.
\end{proposition}

 We illustrate this with two concrete families for $l=3$, for which we have provided annotated circuits in App.~\ref{app:annotated_quirk_circuits}, in addition to the active Quirk~\cite{quirk} links throughout the main text.
\paragraph*{Four-qubit malleable family.} 
For $l=3$, the closed-form solution for symmetric circuits at $s=1$ yields $n_{\min}=4$, and thus the output at step $j$ requires $2^{4-j}-1\ T$ states for unitary preparation. The single four-qubit borrowed identity (\quirk{\qparentfour}{Quirk})
yields five distinct factories spanning all three magic states T, CS and CCZ: 
\begin{itemize}
\item Remove all weight-$\le 3$ gates on the output block:
  $[[8, 3, 2]]$ with output
  $\ket{\CCZ}\equiv 7\ T$ states~(\quirk{\qaaazthreetwo}{Quirk}).
\item Sequential step $j = 2$ on the same parent: $[[12, 2, 2]]$
  with output $\ket{\CS}\equiv 3\ T$ states~(\quirk{\qtwelvetwotwo}{Quirk}). 
\item Sequential step $j = 3$: $[[14, 1, 2]]$ with output
  $\ket{T}$~(\quirk{\qparentfour}{Quirk}). 
\item Apply $\ket{\CCZ} \to \ket{T}$ catalytic conversion after step $j=1$, retaining the check qubit (performing check on the output qubits): $[[11, 1, 2]]$ yields
  output $\ket{T}$~(\quirk{\qelevenoneone}{Quirk}). 
\item Same catalytic conversion but now retaining two output qubits, and measure one output qubit to disentangle the magic states: $[[10, 2, 2]]$
  with two $\ket{T}$ outputs~(\quirk{\qtentwotwo}{Quirk}). 
\end{itemize}
Of these, $[[8,3,2]]$~\cite{bombin2007exact,jones2013composite,eastin2013distilling} and $[[10,2,2]]$~\cite{meier2012magic} are known; $[[11,1,2]]$ is introduced here. The framework recovers all five as endpoints of one parent. Moreover, the $[[10,2,2]]$ factory we present has a failure probability of $9p$ and not $10p$, given the error probability of a single $T$ state is $p$ (see Fig.~\ref{fig:four-qubit-circuit}). The two catalytic factories use Clifford gates in addition to the components of a borrowed identity circuit---hence non-CSS---and lower the input $T$-count: the catalytic conversion uses entangling gates already contained in the extended borrowed identity, while the Hadamard gates are a new addition.

\paragraph*{Five-qubit malleable family.}
Note that symmetric circuits with $k>1$ cannot have distance $d>2$ even for the $s=2$ case, yielding $[[2^{l+1}-2^{l+1-j},l+2-j,2]]$ factories (since $n_{\min}=l+2$ for $s=2$). The five-qubit parent identity (\quirk{\qparentfive}{Quirk}) extends construction to higher $k$ and to distance-$3$ endpoints:
\begin{itemize}
\item $j = 1$: $[[8, 4, 2]]$ with output
  $\equiv 8\ T$ states, that is Clifford equivalent to the CCZ state~(\quirk{\qaaazfourtwo}{Quirk}).
\item $j = 2$: $[[12, 3, 2]]$ with output
  $\equiv 4\ T$ states, that is Clifford equivalent to the CCZ state~(\quirk{\qtwelvethreetwo}{Quirk}). 
\item $j = 3$: $[[14, 2, 2]]$ with two output
  $\ket{T}$~(\quirk{\qfourteentwotwo}{Quirk}).
\item $j = 4$: $[[15, 1, 3]]$ with a single output
  $\ket{T}$ but at increased distance, as distance-3 is retained for $k=1$~(\quirk{\qsummaryfifteen}{Quirk}).
\end{itemize}
This five-qubit circuit shows that preparing a state unitarily equivalent to higher T count need not raise the output $T$-count of the factory itself: $[[8,4,2]]$ ($s{=}3,n{=}5,k{=}4$) has the same catalytic $T$-count as $[[8,3,2]]$ and an output in the same Clifford equivalence class, but its extra qubits give redundant detection of correlated errors that escape the four-qubit realization (App.~\ref{app:annotated_quirk_circuits})---quantifying when this is worth the qubit cost is left open.

At $l=4$ the $n=5$ parent gives $[[16,4,2]]\to[[24,3,2]]\to[[28,2,2]]\to[[30,1,2]]$ (CCCZ, CCS, CT, $\sqrt{T}$) and the $[[31,1,3]]$ endpoint at $s=2$. Based on this construction, we formalized a sequential search which yielded no factories beyond the asymmetric one; the formalism and search results with worked examples of four- and five-qubit parent circuit examples can be found in App.~\ref{app:sequential}. The malleability it makes explicit lets the output magic-state type within a level $l$ be chosen at compile time rather than fixed at hardware-design time---a direct consequence for fault-tolerant compilation.

\vspace{10pt}
\paragraph*{Discussion---}
%\section{Discussion}
We have presented a unified Schr\"odinger-picture framework for magic-state
distillation. A single borrowed-identity condition, holding at every Clifford
level, reaches three classes of factory within one search space
(Fig.~\ref{fig:search-results})---unentangled $T$-output, entangled-output, and non-CSS
catalytic factories---and recovers, in a single search, every distance-$2$
factory known from code constructions. The symmetric case is closed-form, giving
the Reed--Muller chain $[[2^{l+1}-1,1,3]]$ at distance $3$ for every $l$. Its
structural payoff is \emph{malleability}: one parent circuit encodes many
factories, so the output magic-state type is chosen during compilation by gate
removal rather than fixed at hardware design---with no analogue in
Heisenberg-picture frameworks
(Figs.~\ref{fig:four-qubit-circuit}--\ref{fig:five-six-qubit-circuit}).

Three open questions remain toward a most general yet tractable search. Foremost is the smallest \emph{multi-output} ($k>1$) factory at each Clifford level and distance $d>2$: malleability leads us to conjecture that the hybrid $T$-to-any-output factories are smallest for $\CCZ$ (removing $7$ magic states from the identity), largest for $T$-output, and intermediate for $\CS$ ($4$ removed). This is compelling because the $d=2$ analogues $[[8,3,2]]$ and $[[12,2,2]]$ are the smallest CSS-code factories known. At higher-distances ($d$), the borrowed-identity search requires breaking the check-qubit symmetry so that the smallest gate set with cancelling check-supports has size $d$ while retaining more than one output; whether an exhaustive borrowed-identity search can do so without sacrificing its linear, level-uniform structure is open.
Adding \emph{Hadamard gates} to the extended borrowed identity formalism, a single non-diagonal primitive, is equally important since it reaches non-CSS factories such as $[[10,2,2]]$~\cite{meier2012magic}
and the synthillation circuits through a systematic search, at higher rate and
distance, not possible in any search known yet. Finally, the close correspondence between purely non-Clifford borrowed
identities and code-based constructions hints at an \emph{automorphism} between them, a map that would place many known factories within our framework. Together
these directions point toward a single algebraic foundation for distillation across all
Clifford levels and output types, and a sharper boundary between the Schr\"odinger- and Heisenberg-picture approaches.
\paragraph*{Code and Data Availability---}
All circuits are available via the Quirk~\cite{quirk} links in the text and
appendices; the search code and complete output ($N\ge1000$) are at
Ref.~\cite{singh_shraggymsd_simulations_2026}.
\paragraph*{Acknowledgements---} SS would like to thank the Google student research 2024 program, and acknowledges discussions with Anqi Gong and Shubham Jain. While preparing this manuscript, we became aware of concurrent and independent work~\cite{jacinto2026exploring} studying compact distillation factories via a SAT-based search, establishing no-go theorems and single-output ($k=1$) protocols at distance $\geq 3$; single-output factories at higher distance are likewise within reach of the borrowed-identity framework through SAT search but the focus of our work was instead a fast exhaustive search on the multi-output ($k>1$) regime and the malleability between factories distilling different magic states at the same Clifford level.

\bibliographystyle{apsrev4-2}
\bibliography{distill}

%% ===============================================================
%%  Appendices
%% ===============================================================
\appendix

\section{Schr\"odinger vs.\ Heisenberg picture: specialization
to triorthogonality}
\label{app:automorphisms}
%% -----------------------------------------------------------
\begin{proposition}
\label{prop:schrodinger-heisenberg}
Every transversal-$T$ distillation circuit on a CSS code is, under
spacetime duality, a borrowed identity. The converse fails: the
catalytic factories $[[10, 2, 2]]$ and $[[11, 1, 2]]$ constructed
in Sec.~\ref{sec:malleability} are borrowed identities whose
construction includes Clifford+$T$ post-processing and therefore
does not arise as the transversal $T$ of any CSS code.
\end{proposition}

The strict containment is the source of the framework's reach. The specialization of the borrowed-identity condition to
triorthogonality (transversal $T$ at $l = 3$) and
quasitransversality~\cite{campbell2017unified} via the duality is
detailed in this Appendix.
\paragraph{Heisenberg picture.}
For a CSS code with $X$-stabilizer group generated by
$\{X^{\otimes b_1}, \ldots, X^{\otimes b_n}\}$ on $N$ physical
qubits, the action of $T^{\otimes N}$ on each generator follows
from $T X T^\dagger = (X + Y)/\sqrt{2}$:
\begin{equation}
  T X^{\otimes b}T^\dagger
  = X^{\otimes b} \cdot O_b,
\end{equation}
with $O_b$ a $Z$-polynomial on the support of $b$. For
$T^{\otimes N}$ to implement a logical $T$ gate, each $O_b$
must act as
\begin{equation}
  O_b = \frac{I + i\sqrt{2}\,Z_L}{2^{|b|/2}},
\end{equation}
pinning $O_b$ on \emph{both} $\ket{\pm}_L$. This is the
triorthogonality condition of
Ref.~\cite{bravyi2012magic}; the matrix search of
Ref.~\cite{nezami2022classification} enumerates solutions.

\paragraph{Spacetime dual.}
Under spacetime rotation, the $N$ physical qubits become
$N$ multi-qubit $Z$-rotations, the $n-k$ $X$-stabilizer checks plus the $k$-qubit output
become $n$ qubits prepared in $\ket{+}$, and the transversal
$T^{\otimes N}$ becomes the application of $N$ rotations
$Z^{\otimes b_i}(\pi/4)$ in time. The +1 outcome of all $X$-stabilizer
measurements in the original picture correspond, in the dual,
to the borrowed-identity condition $\mathcal{C}\ket{+}^{\otimes
n} = e^{i\phi}\ket{+}^{\otimes n}$.

%% -----------------------------------------------------------
%% -----------------------------------------------------------
\section{Proof of Lemma~\ref{lem:dicke} (Dicke-sector phase
accumulation) and distance guarantees}
\label{app:dicke_proof}
%% -----------------------------------------------------------

We prove Eq.~\eqref{eq:dicke_phase} by induction on $i$. Direct
counting of how many gates flip a given basis element of the
Dicke sector $\ket{\binom{n}{i}}$ gives the starting form
\begin{equation}
  \Phi(i)=2\theta\!\!\!\sum_{\substack{m\in 2\mathbb{Z}+1\\m\le i}}
  \binom{i}{m}\sum_{\substack{w\in\mathcal{W}\\w\ge m}}
  \binom{n-i}{w-m},
  \label{eq:dicke_phase1}
\end{equation}
where $m$ counts the number of designated qubits in the support of a weight-$w$ gate that overlap with the $i$ flipped qubits; the parity of $m$ determines whether the gate contributes
$2\theta$ or $0$ phase.

\paragraph{Base case $i=1$.}
The only odd $m\le 1$ is $m=1$, so
\begin{equation}
\Phi(1)
=2\theta\sum_{w\in\mathcal{W}}\binom{n-1}{w-1}
=2^{1}\theta\sum_{w\in\mathcal{W}}\binom{n-1}{w-1},
\end{equation}
matching Eq.~\eqref{eq:dicke_phase} at $i=1$.

\paragraph{Inductive step.}
Assume $\Phi(j)\equiv 0\pmod{2\pi}$ for all $j<i$. Apply
Pascal's recurrence
$\binom{n-i}{w-m}=\binom{n-i+1}{w-m}-\binom{n-i}{w-m-1}$
exactly $i-m$ times, raising the bottom binomial index from
$w-m$ to $w-i$:
\begin{equation}
  \binom{n-i}{w-m}=\sum_{r=0}^{i-m}(-1)^{r}\binom{i-m}{r}
  \binom{n-m-r}{w-m-r}.
\end{equation}
Substituting into Eq.~\eqref{eq:dicke_phase1}, setting $j=m+r$,
and reordering the double sum gives
\begin{align}
  \Phi(i)
  &=2\theta\sum_{j=1}^{i}\!\!
  \sum_{\substack{m\text{ odd}\\m\le j}}\!\!
  (-1)^{j-m}\binom{i}{m}\binom{i-m}{j-m}\notag\\
  &\hspace{4em}\times\sum_{\substack{w\in\mathcal{W}\\w\ge j}}
  \binom{n-j}{w-j}.
\end{align}
The identity
$\binom{i}{m}\binom{i-m}{j-m}=\binom{i}{j}\binom{j}{m}$
collects the inner-sum coefficient at each $j$ as
\begin{equation}
  C(i,j)=2\theta\binom{i}{j}\!\!
  \sum_{\substack{m\text{ odd}\\m\le j}}
  (-1)^{j-m}\binom{j}{m}.
\end{equation}
At $j=i$, the standard odd-extraction identity
$\sum_{m\text{ odd}}\binom{j}{m}=2^{j-1}$ gives
$C(i,i)=2^{i}\theta$. For $j<i$, $C(i,j)$ is proportional to
$\Phi(j)$ with integer coefficient $\binom{i}{j}$; by the
inductive hypothesis these terms vanish modulo $2\pi$. Only the
$j=i$ term survives, yielding Eq.~\eqref{eq:dicke_phase}.\;\qed

\paragraph{Distance saturation.} For $s\ge3$ the distance saturates at $d=3$: two equal-weight gates always overlap oddly on two qubits, and adjoining a weight-$1$ gate on one of them leaves a single-qubit odd overlap---if the output lives there, a combination of three gate errors goes undetected.

%% -----------------------------------------------------------
\section{Proof of Lemma~\ref{lem:dicke-asym} (Asymmetric
Dicke-sector phase)}
\label{app:lemma-asym}
%% -----------------------------------------------------------

We prove Eq.~\eqref{eq:phase-asym} by induction on
$i_O + i_S$, extending Lemma~\ref{lem:dicke} from one Dicke index to two. The argument follows the same three-step structure as App.~\ref{app:dicke_proof} --- Pascal's recurrence to raise binomial indices, the identity
$\binom{i}{m}\binom{i - m}{j - m} = \binom{i}{j}\binom{j}{m}$
to refactor the double sum, and an odd-extraction identity to evaluate the parity-restricted inner sum --- now applied
block-by-block, with the two blocks connected only by the
parity constraint on the total overlap $m_O + m_S$.

Per the definition of weight-$w$ rotation gate in Eq.~\eqref{eq:gate_def}, a weight-$(w_O, w_S)$
gate with sign $\sigma_{w_O, w_S}$ contributes phase
$2\sigma_{w_O, w_S}\theta$ to a basis element of Dicke sector
$(i_O, i_S)$ when its overlap $m_O$ with the $i_O$ flipped
output qubits and overlap $m_S$ with the $i_S$ flipped check
qubits satisfies $m_O + m_S$ odd, and phase $0$ otherwise. The
number of weight-$(w_O, w_S)$ gates with overlap $(m_O, m_S)$
is
$\binom{i_O}{m_O}\binom{k - i_O}{w_O - m_O}
\binom{i_S}{m_S}\binom{n - k - i_S}{w_S - m_S}$.
Summing over all gate types and all odd-parity overlaps gives
the starting form
\begin{align}
  \Phi(i_O, i_S) = 2\theta\!\!\!\!\!
  &\sum_{(w_O, w_S) \in \mathcal{W}}\!\!\!\!
  \sigma_{w_O, w_S}
  \sum_{\substack{m_O + m_S \text{ odd}\\
                  m_O \le i_O,\ m_S \le i_S}}
  \binom{i_O}{m_O}\binom{i_S}{m_S}\notag\\
  &\quad \times \binom{k - i_O}{w_O - m_O}
  \binom{n - k - i_S}{w_S - m_S}.
  \label{eq:phase-asym-direct}
\end{align}

\paragraph{Base cases.} The $(i_O, i_S)=(0,0)$ case is trivial, so we verify Eq.~\eqref{eq:phase-asym} at the two smallest
sectors of $i_O+i_S=1$: $(i_O, i_S) = (1, 0)$ and $(0, 1)$.

At $(1, 0)$: the constraint ``$m_O + m_S$ odd with
$m_O \le 1, m_S \le 0$'' forces $m_O = 1, m_S = 0$, so
Eq.~\eqref{eq:phase-asym-direct} becomes
\begin{equation}
  \Phi(1, 0) = 2\theta\!\!\!\!\!
  \sum_{(w_O, w_S) \in \mathcal{W}}\!\!\!\!
  \sigma_{w_O, w_S}\,
  \binom{1}{1}\binom{0}{0}
  \binom{k - 1}{w_O - 1}\binom{n - k}{w_S},
\end{equation}
which simplifies to
$2\theta \sum \sigma_{w_O, w_S}\binom{k - 1}{w_O - 1}
\binom{n - k}{w_S}$ --- matching the
prefactor $(-1)^{1+0+1} 2^{1+0}\theta = 2\theta$ of
Eq.~\eqref{eq:phase-asym} at $(1, 0)$ with binomial factors
$\binom{k - 1}{w_O - 1}\binom{n - k}{w_S}$. By symmetry, the
analogous identity holds at $(0, 1)$.

\paragraph{Inductive step.}
Assume Eq.~\eqref{eq:phase-asym} holds, and equals zero modulo
$2\pi$ by the borrowed-identity hypothesis, for all
$(i_O', i_S')$ with $i_O' + i_S' < i_O + i_S$. We show
Eq.~\eqref{eq:phase-asym} also holds at $(i_O, i_S)$.

The strategy parallels App.~\ref{app:dicke_proof}: use Pascal's
recurrence on both blocks to reorganize
Eq.~\eqref{eq:phase-asym-direct} into a sum of contributions
indexed by sectors $(j_O, j_S) \le (i_O, i_S)$. Iterating the
recurrence
$\binom{k - i_O}{w_O - m_O} = \binom{k - i_O + 1}{w_O - m_O} -
\binom{k - i_O}{w_O - m_O - 1}$ exactly $i_O - m_O$ times on
the output block, and analogously $i_S - m_S$ times on the
check block, gives
\begin{align}
  \binom{k - i_O}{w_O - m_O} &=
  \sum_{j_O = m_O}^{i_O} (-1)^{j_O - m_O}
  \binom{i_O - m_O}{j_O - m_O}\binom{k - j_O}{w_O - j_O},\\
  \binom{n - k - i_S}{w_S - m_S} &=
  \sum_{j_S = m_S}^{i_S} (-1)^{j_S - m_S}
  \binom{i_S - m_S}{j_S - m_S}\binom{n - k - j_S}{w_S - j_S}.
\end{align}

Substituting into Eq.~\eqref{eq:phase-asym-direct} and applying
the binomial identity $\binom{i_O}{m_O}\binom{i_O - m_O}{j_O -
m_O} = \binom{i_O}{j_O}\binom{j_O}{m_O}$ (and analogously for
the check block), the contribution at fixed $(j_O, j_S)$
factorizes as a product of three pieces: a sign and prefactor $(-1)^{j_O + j_S + 1}\, 2\theta\,\binom{i_O}{j_O}\binom{i_S}{j_S}$
[the overall $-1$ from $(-1)^{m_O + m_S} = -1$ on every
odd-parity term, which can be pulled out of the $(m_O, m_S)$
sum], a combinatorial weight
\begin{equation}
\sum_{\substack{m_O + m_S \text{ odd}\\m_O \le j_O,\ m_S \le j_S}}
\binom{j_O}{m_O}\binom{j_S}{m_S}
=
\begin{cases}
2^{j_O + j_S - 1} & j_O + j_S \ge 1,\\
0 & j_O + j_S = 0,
\end{cases}
\label{eq:two-var-odd}
\end{equation}
[obtained by splitting odd parity into ``$m_O$ even, $m_S$
odd'' and ``$m_O$ odd, $m_S$ even'' and applying
$\sum_{m\text{ even}}\binom{j}{m} = \sum_{m\text{ odd}}
\binom{j}{m} = 2^{j-1}$ for $j \ge 1$], and a gate-type sum
$\sum_{(w_O, w_S)\in\mathcal{W}} \sigma_{w_O, w_S}
\binom{k - j_O}{w_O - j_O}\binom{n - k - j_S}{w_S - j_S}$.

The full expression at $(j_O, j_S)$ therefore reduces to
$\binom{i_O}{j_O}\binom{i_S}{j_S}$ times the closed-form right-hand
side of Eq.~\eqref{eq:phase-asym} evaluated at $(j_O, j_S)$ ---
i.e., to $\binom{i_O}{j_O}\binom{i_S}{j_S} \cdot \Phi(j_O, j_S)$.
For $(j_O, j_S) \ne (i_O, i_S)$ with $j_O + j_S \ge 1$, this is an
integer multiple of $\Phi(j_O, j_S)$, which vanishes modulo
$2\pi$ by the inductive hypothesis. The $(0, 0)$ contribution
vanishes from Eq.~\eqref{eq:two-var-odd}. Only the top-sector
contribution at $(j_O, j_S) = (i_O, i_S)$ survives, yielding
Eq.~\eqref{eq:phase-asym}. \qed
%--------------------------------------
\section{Output-state classification via residual phase}
\label{app:output-extraction}
Setting $i_S=0$ in Lemma~\ref{lem:dicke-asym}, the residual phase
$\Phi'(i_O,0)$ depends only on the Hamming weight $i_O=|x|$ of the output
qubits, so the output state is
\begin{equation}
  |\Psi_{\mathrm{out}}\rangle
  = \frac{1}{\sqrt{2^k}}
    \sum_{x\in\{0,1\}^k} e^{i\Phi'(|x|,\,0)}\,|x\rangle,
\end{equation}
with $\Phi'(0,0)=0$.  Define the $d$-th iterated finite difference of
$\Phi'(\cdot,0)/\theta$ at $0$ as
\begin{equation}
  \Delta^d\!\left[\frac{\Phi'(\cdot,0)}{\theta}\right](0)
  \;:=\;
  \sum_{j=0}^{d}(-1)^{d-j}\binom{d}{j}\frac{\Phi'(j,0)}{\theta}.
  \label{eq:finite-diff}
\end{equation}

\begin{proposition}
$|\Psi_{\mathrm{out}}\rangle$ contains, as its dominant entangling content, a
$d$-qubit level-$l$ Clifford-hierarchy resource---non-Clifford (magic)
precisely for $l\ge3$---where $d$ is the largest index in $\{1,\dots,k\}$
for which \eqref{eq:finite-diff} is non-zero modulo $2^{l+1}$.
\end{proposition}

\begin{proof}
The inductive step of Lemma~\ref{lem:dicke-asym} shows that $\Phi'(i_O,0)/\theta$
receives a non-vanishing contribution only from the top sector
$(j_O,j_S)=(i_O,0)$, with all lower sectors cancelling by the
borrowed-identity condition.  The binomial factor
$\binom{i_O}{j_O}\binom{i_S}{j_S}\big|_{i_S=j_S=0}=\binom{i_O}{d}$
at $j_O=d$ is precisely the degree-$d$ Newton basis element, so
Eq.~\eqref{eq:phase-asym} at $i_S=0$ is already the Newton expansion
\begin{equation}
  \frac{\Phi'(i_O,0)}{\theta}
  = \sum_{d=1}^{k}
    \Delta^d\!\left[\frac{\Phi'(\cdot,0)}{\theta}\right](0)
    \binom{i_O}{d}.
\end{equation}
Each degree-$d$ term corresponds to a symmetric layer of $d$-qubit diagonal
gates applied to all size-$d$ subsets of the output qubits.  Following
Ref.~\cite{campbell2017unified}, such a layer realizes a level-$l$ diagonal
gate---non-Clifford (magic) precisely for $l\ge3$---if and only if its
Newton coefficient is non-zero modulo
$2^{l+1}$.  The largest such $d$ identifies the dominant resource;
and is used to define the Clifford-equivalence classes, based on which the outputs are classified.
\end{proof}
\section{Example verifications}
\label{app:examples}
%% -----------------------------------------------------------

We verify the borrowed-identity and output-extraction
conditions for two canonical factories: the sequential
$[[8, 3, 2]]$ and the asymmetric $[[20, 4, 2]]$.

\subsection{The $[[8, 3, 2]]$ factory ($n = 4$, $k = 3$,
$l = 3$)}
\label{app:eight32}

The $[[8, 3, 2]]$ factory (\quirk{\qaaazthreetwo}{Quirk})
arises from the sequential construction at $j = 1$ with
$\mathcal{W}_1 = \{1, 2, 3, 4\}$, $n = 4$, $\theta = \pi/8$.

\paragraph{Borrowed-identity check.}
The Dicke-sector condition (Lemma~\ref{lem:dicke}) at $i = 1$
requires
$\sum_{w \in \mathcal{W}} \binom{3}{w - 1}
\in 2^{l - i + 1}\mathbb{Z} = 8\mathbb{Z}$:
\begin{equation}
  \tbinom{3}{0} + \tbinom{3}{1} + \tbinom{3}{2} + \tbinom{3}{3}
  = 1 + 3 + 3 + 1
  = 8 \equiv 0 \pmod{8}.\;\checkmark
\end{equation}
Higher Dicke sectors $i \in \{2, 3\}$ give modulus exponents
$l - i + 1 \in \{2, 1\}$, with the sum $\sum_w \binom{3}{w-i}
\in \{4, 2\}$ respectively, both satisfying the relevant
divisibility. The borrowed-identity gate count is
$\binom{4}{1} + \binom{4}{2} + \binom{4}{3} + \binom{4}{4}
= 15$.

\paragraph{Output extraction.}
Removing all gates of weight $w \le k = 3$ whose support lies
entirely on the three designated output qubits leaves
$N = 15 - (\binom{3}{1} + \binom{3}{2} + \binom{3}{3})
= 15 - 7 = 8$ gates of types
$(w_O, w_S) \in \{(0, 1), (1, 1), (2, 1), (3, 1)\}$.
The residual phase on the output Dicke sectors,
normalized so that $\Phi'(0)/\theta = 0$, evaluates to
\begin{equation}
  \Phi'(i_O)/\theta  \equiv
  \begin{cases}
    0 & i_O = 0,\\
    -8 & i_O = 1, 2, 3,
  \end{cases}
  \pmod{16}.
\end{equation}
The third finite difference in $i_O$ is constant and nonzero,
identifying the residual as a degree-$3$ polynomial---the
signature of $\ket{\CCZ}$ output up to Clifford
correction.\;\checkmark

\subsection{The $[[20, 4, 2]]$ factory ($n = 7$, $k = 4$,
$l = 3$)}
\label{app:twenty42}

We verify the asymmetric borrowed-identity condition at
$(l, n, k) = (3, 7, 4)$, $\theta = \pi/8$,
$(s_\mathrm{total}, s_O, s_S) = (2, 3, 1)$.

\paragraph{Gate set.}
$\mathcal{W}_\mathrm{total} = \{1, 3, 5, 7\}$,
$\mathcal{W}_O = \{0, 1, 4\}$,
$\mathcal{W}_S = \{0, 1, 2, 3\}$. The allowed pairs are
$\mathcal{W} = \{(0, 1), (0, 3), (1, 0), (1, 2), (4, 1),
(4, 3)\}$. The borrowed-identity gate count is
\begin{align}
  N_\mathrm{circuit}
  &= \tbinom{4}{0}\tbinom{3}{1} + \tbinom{4}{0}\tbinom{3}{3}
   + \tbinom{4}{1}\tbinom{3}{0}\notag\\
  &\quad + \tbinom{4}{1}\tbinom{3}{2}
   + \tbinom{4}{4}\tbinom{3}{1} + \tbinom{4}{4}\tbinom{3}{3}\notag\\
  &= 3 + 1 + 4 + 12 + 3 + 1 = 24.
\end{align}

\paragraph{Sign assignment.}
The two-cell deterministic check of asymmetric circuits fails at the all-positive
assignment. The second cell, with
$\sigma_{0, 1} = \sigma_{0, 3} = -1$ on the check-only gates
and $\sigma = +1$ on every gate touching at least one output
qubit, satisfies Eq.~\eqref{eq:asymmetric-master}; we verify
this explicitly below.

\paragraph{Phase verification.}
By Lemma~\ref{lem:dicke-asym}, the closed-form phase at sector
$(i_O, i_S)$ is
\begin{align}
  \Phi(i_O, i_S)/\theta &\equiv
  (-1)^{i_O + i_S + 1}\, 2^{i_O + i_S}\nonumber\\&\quad\times
  \!\!\!\!\sum_{(w_O, w_S) \in \mathcal{W}}\!\!\!\!
  \sigma_{w_O, w_S}\,T_{w_O, w_S}^{(i_O, i_S)}
  \pmod{16},
\end{align}
where $T_{w_O, w_S}^{(i_O, i_S)} \coloneqq
\binom{k - i_O}{w_O - i_O}\binom{n - k - i_S}{w_S - i_S}$ is
the gate-overlap binomial product, with $k = 4$, $n - k = 3$.
For each of the six allowed pairs $(w_O, w_S)$, we tabulate
$T_{w_O, w_S}^{(i_O, i_S)}$ at the relevant non-trivial Dicke
sectors:
\begin{center}
\begin{tabular}{c|cccccc}
$(i_O, i_S)$
& $(0, 1)$ & $(0, 3)$ & $(1, 0)$ & $(1, 2)$ & $(4, 1)$ & $(4, 3)$\\
\hline
$(0, 1)$ & 1 & 1 & 0 & 8 & 1 & 1\\
$(1, 0)$ & 0 & 0 & 1 & 3 & 3 & 1\\
$(0, 2)$ & 0 & 1 & 0 & 4 & 0 & 1\\
$(1, 1)$ & 0 & 0 & 0 & 2 & 1 & 1\\
$(1, 2)$ & 0 & 0 & 0 & 1 & 0 & 1\\
$(2, 1)$ & 0 & 0 & 0 & 0 & 1 & 1\\
\vdots & \multicolumn{6}{c}{(remaining sectors analogously)}
\end{tabular}
\end{center}
Multiplying each $T_{w_O, w_S}^{(i_O, i_S)}$ by the
corresponding sign $\sigma_{w_O, w_S}$ and the prefactor
$(-1)^{i_O + i_S + 1} 2^{i_O + i_S}$, the result at every
$(i_O, i_S)$ vanishes modulo $16$.\;\checkmark

\paragraph{Output extraction.}
Removing the four $(1, 0)$ gates leaves $N = 20$ gates. The
residual output phase satisfies
\begin{equation}
  \Phi'(i_O, 0)/\theta = -2 i_O \pmod{16},
\end{equation}
a degree-$1$ polynomial in $i_O$, identifying the output as
$\ket{T}^{\otimes 4}$ up to Clifford correction.\;\checkmark
%% -----------------------------------------------------------
\section{Asymmetric search: methods and catalogue}
\label{app:asymmetric_search}
%% -----------------------------------------------------------

\subsection{Search procedure}
\label{app:asymmetric_procedure}

The asymmetric search enumerates over the skip parameters
$(s_\mathrm{total}, s_O)$ with $s_S = 1$ fixed. For each parameter
tuple $(l, n, k, s_\mathrm{total}, s_O)$, the gate set $\mathcal{W}$
is determined by the three weight-separation sets
$\mathcal{W}_\mathrm{total}$, $\mathcal{W}_O$, $\mathcal{W}_S$ as
described for asymmetric circuits. Validity is checked by
verifying Eq.~\eqref{eq:asymmetric-master} at each Dicke sector
$(i_O, i_S)$, with $\Phi(i_O, i_S)$ evaluated via
Lemma~\ref{lem:dicke-asym} in $O(k(n - k))$ phase computations per
sign assignment. 

\paragraph{Sign assignment.}
At each parameter tuple, two deterministic sign assignments are
checked in sequence. Cell 1 sets $\sigma_{w_O, w_S} = +1$ for
every gate type. Cell 2 sets $\sigma_{0, w_S} = -1$ on every
check-only gate type ($w_O = 0$) and $\sigma_{w_O, w_S} = +1$
on every gate touching at least one output qubit ($w_O \ge 1$).
The parameter tuple is valid if
Eq.~\eqref{eq:asymmetric-master} holds for either cell and
invalid otherwise. Empirically, mixed sign patterns on check-only gates are never required to recover the factories
listed in
Tables~\ref{tab:asymmetric-l2}--\ref{tab:asymmetric-l4}, so the sign-enumeration cost is bounded by $2$ verifications
per tuple rather than $2^{|\mathcal{W}|}$. We check only whether one of the two cells succeeds, because under the two-group symmetry considered here, with $s_S=1$, the sign assignment affects only validity and not the depth of the circuit. However, circuits with $s_S\neq 1$ or different permutation symmetries may yield more efficient factories from a search over mixed-sign types--this scope is left for future work. 
\paragraph{Output classification.}
After extraction via Eq.~\eqref{eq:N-asym}, the residual phase
polynomial $\Phi'(i_O, 0)/\theta$ is computed by finite differences
in $i_O$. A polynomial of degree $1$ in $i_O$ corresponds to
unentangled $\ket{T}^{\otimes k}$ outputs; a polynomial of degree
$2, 3, \ldots, l$ indicates an entangled $\mathbb{C}_{l}$ magic
state, with degree-$2$ giving $\ket{\CS}$, degree-$3$ giving
$\ket{\CCZ}$, and degree-$l$ at $l \ge 4$ giving the $\CkZ{l-1}$
gate of the $l$-th level. 
Solutions whose extracted output has vanishing residual phase are excluded as trivial stabilizer-state outputs. The rest are classified by residual polynomial degree (see App.~\ref{app:output-extraction}); for example, $[[8, 4, 2]]$ at $l = 3$ is CCZ-equivalent despite its four-qubit output and is classified as such in Fig.~\ref{fig:search-results}.

\subsection{Search complexity and runtime}
\label{app:asymmetric_runtime}
\begin{figure}[h]\centering
\includegraphics[width=\columnwidth]{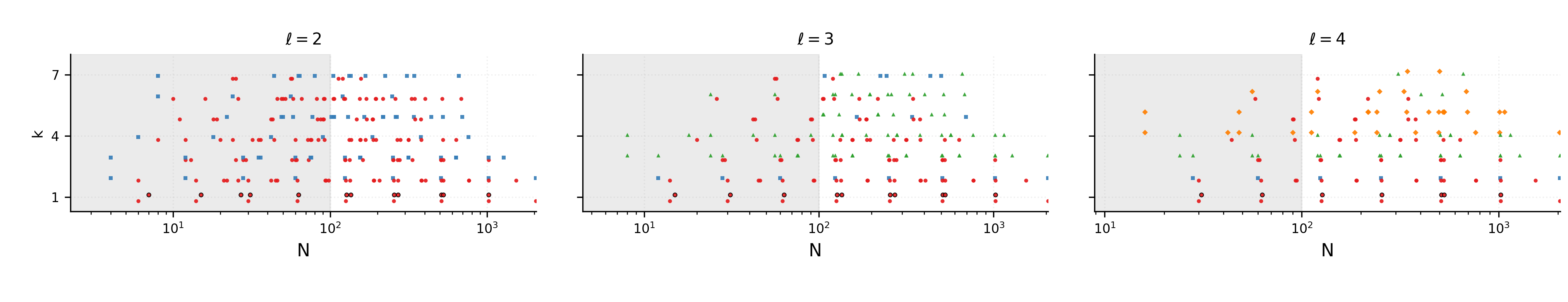}
\caption{Full extent of the two-group catalogue ($N$ up to ${\sim}10^3$); the shaded band marks the $N\le100$ region shown in Fig.~\ref{fig:search-results}.}
\label{fig:search-ext}
\end{figure}
For each $(l, n, k)$ triple, the number of
$(s_\mathrm{total}, s_O)$ parameter tuples is
$|S_\mathrm{total}| \times |S_O|$; per tuple, the deterministic two-cell check requires at most $2$ verifications. Each
verification costs $O(k(n-k))$ phase computations. The total cost at given $(l, n, k)$ is thus
\begin{equation}
O\bigl(|S_\mathrm{total}| \cdot |S_O|  \cdot k(n-k)\bigr),
\label{eq:asymmetric-cost}
\end{equation}
polynomial in $n$ and $k$ at fixed $|S_\mathrm{total}|$, $|S_O|$.

To make this scaling concrete, we ran four sweeps of progressively
wider parameter envelope: an $l \in \{2, 3\}$, $k\le 4$,
$n \le 9$, $s_\mathrm{total}, s_O \le 4$ proof-of-concept range
completes in $1.1$ seconds; extending to $l \in \{2, 3, 4\}$
with the same other bounds gives $1.3$ seconds; widening to
$l \in \{2, 3, 4\}$, $k \le 7$, $n \le 11$,
$s_\mathrm{total}, s_O \le 6$ takes $\sim 7$ seconds; and the full
sweep used throughout this paper, $l \in \{2, 3, 4\}$, $k \le 7$,
$n \le 11$, $s_\mathrm{total}, s_O \le 7$, completes in $\sim 9$ seconds
on a standard laptop (Apple M4, 16 GB RAM) and produces $21{,}920$
valid factories. 

The sub-linear growth from $\sim 7$ to $\sim 9$ seconds across the
$s_\mathrm{total}, s_O \in \{1, \ldots, 6\} \to \{1, \ldots, 7\}$
extension reflects the early-termination structure: at fixed
$(l, n, k)$, larger skip parameters admit progressively fewer
gate types, so the verification at $(s_\mathrm{total}, s_O) = (7, 7)$
is cheaper than at $(1, 1)$. 
The parameter range was chosen to cover all known distance-$2$
factories rather than to saturate the verifier; extending to
larger $n$, larger $k$, or higher $l$ carries trivial additional
cost.
\subsection{Catalogue of factory families}
\label{app:asymmetric_catalogue}
\begin{table}[h]
\centering
\caption{Distance-$2$ family seeds at $l = 2$ recovered by the
asymmetric search (for range covered in Fig.~\ref{fig:search-results}), limited to $N\le 30$. The $k=1$ rows with $s_\mathrm{total}\ge2$ (e.g.\ $[[7,1]]$, $[[15,1]]$) are distance-$3$; all other rows are distance-$2$.}
\label{tab:asymmetric-l2}
\begin{tabular}{c c c c c}
\toprule
$[[N, k]]$ & class & $n$ & $s_\mathrm{total}$ & $s_O$ \\
\midrule
$[[6, 1]]$    & $S$   & 3 & 1 & 1 \\
$[[7, 1]]$    & $S$   & 4 & 2 & 1 \\
$[[14, 1]]$   & $S$   & 4 & 1 & 1 \\
$[[15, 1]]$   & $S$   & 5 & 2 & 1 \\
$[[27, 1]]$   & $S$   & 7 & 4 & 1 \\
$[[30, 1]]$   & $S$   & 5 & 1 & 1 \\
\midrule
$[[4, 2]]$    & CZ    & 3 & 1 & 1 \\
$[[6, 2]]$    & $S$   & 4 & 2 & 1 \\
$[[12, 2]]$   & CZ    & 4 & 1 & 1 \\
$[[14, 2]]$   & $S$   & 5 & 2 & 1 \\
$[[21, 2]]$   & $S$   & 5 & 1 & 2 \\
$[[22, 2]]$   & $S$   & 6 & 2 & 2 \\
$[[26, 2]]$   & $S$   & 7 & 4 & 1 \\
$[[28, 2]]$   & CZ    & 5 & 1 & 1 \\
$[[30, 2]]$   & $S$   & 6 & 2 & 1 \\
\midrule
$[[4, 3]]$    & CZ    & 4 & 2 & 1 \\
$[[12, 3]]$   & $S$   & 5 & 1 & 3 \\
$[[12, 3]]$   & CZ    & 5 & 2 & 1 \\
$[[13, 3]]$   & $S$   & 6 & 2 & 3 \\
$[[25, 3]]$   & $S$   & 7 & 4 & 1 \\
$[[28, 3]]$   & $S$   & 6 & 1 & 3 \\
$[[28, 3]]$   & CZ    & 6 & 2 & 1 \\
$[[29, 3]]$   & $S$   & 7 & 2 & 3 \\
\midrule
$[[6, 4]]$    & CZ    & 5 & 1 & 3 \\
$[[8, 4]]$    & $S$   & 6 & 2 & 3 \\
$[[18, 4]]$   & CZ    & 6 & 1 & 3 \\
$[[20, 4]]$   & $S$   & 7 & 2 & 3 \\
$[[24, 4]]$   & $S$   & 7 & 4 & 1 \\
\midrule
$[[18, 5]]$   & $S$   & 7 & 1 & 5 \\
$[[19, 5]]$   & $S$   & 8 & 2 & 5 \\
$[[22, 5]]$   & CZ    & 7 & 4 & 1 \\
\midrule
$[[8, 6]]$    & CZ    & 7 & 1 & 5 \\
$[[10, 6]]$   & $S$   & 8 & 2 & 5 \\
$[[16, 6]]$   & $S$   & 7 & 4 & 1 \\
$[[24, 6]]$   & CZ    & 8 & 1 & 5 \\
$[[26, 6]]$   & $S$   & 9 & 2 & 5 \\
\bottomrule
\end{tabular}
\end{table}

\begin{table}[h]
\centering
\caption{Distance-$2$ family seeds at $l = 3$ recovered by the
asymmetric search (for range covered in Fig.~\ref{fig:search-results}), limited to $N\le 30$. The $k=1$ row $[[15,1]]$ ($s_\mathrm{total}=2$) is distance-$3$; all other rows are distance-$2$. Format as in
Table~\ref{tab:asymmetric-l2}.}
\label{tab:asymmetric-l3}
\begin{tabular}{c c c c c}
\toprule
$[[N, k]]$ & class & $n$ & $s_\mathrm{total}$ & $s_O$ \\
\midrule
$[[14, 1]]$   & $T$      & 4 & 1 & 1 \\
$[[15, 1]]$   & $T$      & 5 & 2 & 1 \\
$[[30, 1]]$   & $T$      & 5 & 1 & 1 \\
\midrule
$[[12, 2]]$   & $\CS$    & 4 & 1 & 1 \\
$[[14, 2]]$   & $T$      & 5 & 2 & 1 \\
$[[28, 2]]$   & $\CS$    & 5 & 1 & 1 \\
$[[30, 2]]$   & $T$      & 6 & 2 & 1 \\
\midrule
$[[8, 3]]$    & $\CCZ$   & 4 & 1 & 1 \\
$[[12, 3]]$   & $\CCZ$   & 5 & 2 & 1 \\
$[[24, 3]]$   & $\CCZ$   & 5 & 1 & 1 \\
$[[28, 3]]$   & $T$      & 6 & 1 & 3 \\
$[[28, 3]]$   & $\CCZ$   & 6 & 2 & 1 \\
$[[29, 3]]$   & $T$      & 7 & 2 & 3 \\
\midrule
$[[8, 4]]$    & $\CCZ$   & 5 & 2 & 1 \\
$[[18, 4]]$   & $\CCZ$   & 6 & 1 & 3 \\
$[[20, 4]]$   & $T$      & 7 & 2 & 3 \\
$[[24, 4]]$   & $\CCZ$   & 6 & 2 & 1 \\
\midrule
$[[24, 6]]$   & $\CCZ$   & 8 & 1 & 5 \\
$[[26, 6]]$   & $T$      & 9 & 2 & 5 \\
\bottomrule
\end{tabular}
\end{table}

\begin{table}[h]
\centering
\caption{Family seeds at $l = 4$ recovered by the asymmetric
search (for range covered in Fig.~\ref{fig:search-results}), limited to $N\le 30$. The $\CkZ{3}$-output factories
correspond to the new $l=4$ family. Format as in
Table~\ref{tab:asymmetric-l2}.}
\label{tab:asymmetric-l4}
\begin{tabular}{c c c c c}
\toprule
$[[N, k]]$ & class & $n$ & $s_\mathrm{total}$ & $s_O$ \\
\midrule
$[[30, 1]]$   & $\sqrt{T}$ & 5 & 1 & 1 \\
\midrule
$[[28, 2]]$   & CT         & 5 & 1 & 1 \\
$[[30, 2]]$   & $\sqrt{T}$ & 6 & 2 & 1 \\
\midrule
$[[24, 3]]$   & CCS        & 5 & 1 & 1 \\
$[[28, 3]]$   & CCS        & 6 & 2 & 1 \\
\midrule
$[[16, 4]]$   & CCCZ       & 5 & 1 & 1 \\
$[[24, 4]]$   & CCS        & 6 & 2 & 1 \\
\midrule
$[[16, 5]]$   & CCCZ       & 6 & 2 & 1 \\
\bottomrule
\end{tabular}
\end{table}
Tables~\ref{tab:asymmetric-l2}, \ref{tab:asymmetric-l3}, and
\ref{tab:asymmetric-l4} list one representative per
$(k, d, \text{output type}, \text{entanglement})$ combination
recovered by the asymmetric search, at $l = 2, 3, 4$ respectively.
Each row shows the smallest $N$ realizing that combination, the
smallest-$n$ parameter tuple producing that $N$, and the output
magic-state type. The full $21{,}920$-entry catalogue, including
all members of each family at every $n \le 11$, is provided in the
accompanying code
repository~\cite{singh_shraggymsd_simulations_2026}. 

\paragraph{Coverage relative to known catalogues.}
The Bravyi--Haah family $[[3k+8, k, 2]]$ for $k = 2, 4, 6$ appears
systematically at $(s_\mathrm{total}, s_O) = (2, 2j-1)$ with
$k = 2j$: $[[14, 2, 2]]$ and $[[26, 6, 2]]$ are recovered with
all-positive signs; $[[20, 4, 2]]$ requires negative signs
$\sigma_{0, 1} = \sigma_{0, 3} = -1$ on the check-only gates,
while the $k = 2$ and $k = 6$ members are not. 
The cube factory $[[8, 3, 2]]$ and the entangled-output factory
$[[12, 2, 2]]$ are recovered with all-positive signs at
$(s_\mathrm{total}, s_O) = (1, 1)$. The Reed--Muller chain at
$k = 1$ is recovered for $d \in \{2, 3\}$ via
$s_\mathrm{total} \in \{1, 2\}$.
The extended sweep additionally captures every distance-$2$ entry
of the Nezami--Haah triorthogonal-code
catalogue~\cite{nezami2022classification} at the largest $k$ for
each $N \le 32$.

\paragraph{New factory families at $l = 4$.}
A family of $\CkZ{3}$-output factories appears at $l = 4$, $k = 4$,
all-positive signs, $(s_\mathrm{total}, s_O) = (1, 1)$:
$[[16, 4, 2]]$ at $n = 5$ extending through $[[48, 4, 2]]$,
$[[112, 4, 2]]$, $[[240, 4, 2]]$, $[[496, 4, 2]]$,
$[[1008, 4, 2]]$, and $[[2032, 4, 2]]$ at $n = 11$. The smallest
member, $[[16, 4, 2]]$, is the $l = 4$ analogue of the cube
$[[8, 3, 2]]$ at $l = 3$: in both cases, a parent borrowed
identity at $n = l + 2$ produces a degree-$l$ output polynomial,
the highest non-trivial $\mathbb{C}_l$ magic state. Larger
$\CkZ{3}$-output families also appear at $k = 5, 6, 7$ within the
swept range; see Table~\ref{tab:asymmetric-l4}.

\paragraph{Efficiency landmarks.}
Among small distance-$2$ factories with unentangled $T$ output at
$l = 3$, the $[[20, 4, 2]]$ Bravyi--Haah $k=4$ factory attains
$\gamma = \log_2(N/k) = \log_2(20/4) \approx 2.32$, second only
to the Bravyi--Haah $k=6$ factory $[[26, 6, 2]]$ ($\gamma \approx 2.12$) and tying with the catalytic factory $[[10, 2, 2]]$ ($\gamma \approx 2.32$ but with
the smallest $N$)
in the search-recovered set. 
For entangled $\CCZ$-output factories, $[[8, 3, 2]]$ at $l = 3$
attains $\gamma = \log_2(8/3) = 1.4$, the smallest $\gamma$ in the
$\CCZ$-output catalogue. 

%% -----------------------------------------------------------
\section{Sequential search: methods and catalogue}
\label{app:sequential}
%% -----------------------------------------------------------
We can formalize the malleable circuits as follows. For each step $j\in\{1,2,\ldots\}$, the $j$th qubit is disentangled from the rest of the circuit using only multi-weight phase rotations. To do so, we construct a borrowed identity on the remaining $n-j$ qubits and extend each of its gates to also act on the $j$th qubit. After this circuit, the check qubit returns to
$\ket{+}$ and is measured; what survives is an entangled magic state factory on $k=n - j$ qubits. In the main text, we have only discussed examples where $s_j$ is the same for all steps $j$. Note that this is not a necessary condition, and we cover the general case of arbitrary $s_j$ in a brute-force search over sequential factories. Here, we give the closed-form disentangling condition for the sequential search with different $s_j$ in each step. 
\subsection{Construction and phase condition}
\begin{definition}[Sequential borrowed identity]
\label{def:sequential}
Let $\theta = \pi / 2^l$ and partition $n$ qubits into check
qubits $\mathcal{S} = \{q_1, \ldots, q_{n - k}\}$ and output
qubits $\mathcal{O} = \{q_{n - k + 1}, \ldots, q_n\}$. A
\emph{sequential borrowed identity} is constructed by fixing
each $q_j$ in turn and applying all multi-weight phase rotation (Eq.~\ref{eq:gate_def}) from
$\{q_j, \ldots, q_n\}$ with support containing $q_j$ and
weights in $\mathcal{W}_j$, until $q_j$ returns to $\ket{+}$.
The borrowed-identity gate count is
\begin{equation}
  N = \sum_{j = 1}^{n - k}
  \sum_{\substack{w \in \mathcal{W}_j\\ w \le n - j + 1}}
  \binom{n - j}{w - 1}.
\end{equation}
The output qubits are left in a $k$-qubit magic state
$\ket{\Psi_k}$ for some diagonal $\mathbb{C}_l$
gate~\cite{cui2017diagonal}, with the specific gate determined
by the residual phase polynomial
(App.~\ref{app:output-extraction}).
\end{definition}

At step $j$, the active qubit $q_j$ is in superposition
$\ket{0}_j + e^{i \Delta\Phi_j(i)} \ket{1}_j$ entangled with
Dicke states $\ket{i}$ over the remaining qubits. To
disentangle $q_j$, we require $\Delta\Phi_j(i)$ to be (i)
constant in $i$ and (ii) an integer multiple of $\pi/2$, so
that the residual qubit is a stabilizer state measurable as a
check qubit.

\begin{theorem}[Sequential phase condition]
\label{thm:sequential-condition}
For a sequential borrowed identity in which each Dicke-sector
yields a local phase that is a multiple of $2\pi$ on the active
qubit, the disentangling condition at step $j$ is
\begin{equation}
  \Delta\Phi_j(i) =
  \begin{cases}
  \displaystyle 2^{i+1}\theta\,\sigma_j\sum_{w\in\mathcal{W}_j}
  \binom{n-j-i}{w-i}, & i\ge 1,\ j=1,\\[6pt]
  \displaystyle 2\theta\,\sigma_j\sum_{w\in\mathcal{W}_j}
  \binom{n-j}{w-1}\\
  \displaystyle +\,2\theta\!\sum_{q=1}^{j-1}\!\sigma_q\sum_{w\in\mathcal{W}_q}
  \binom{n-q-1}{w-1}, & i=0,\\[6pt]
  0, & i\ge 1,\ j>1,
  \end{cases}
  \label{eq:sequential-condition}
\end{equation}
where all cases are taken mod $2\pi$ and $\sigma_j \in \{\pm 1\}$
are per-step signs. For $j=1$ this recovers
Theorem~\ref{thm:main}. Generically, the construction terminates at $j_f = l$, when the residual phase prefactor reaches $2\theta$ and the output
qubits encode a $\mathbb{C}_l$ magic state.
\end{theorem}

The three cases reflect the structure of the phase accumulation.
For $j=1$ and $i\ge 1$, the standard Dicke-sector simplification
of the even-odd parity difference applies directly.
For $j>1$ and $i>0$, the even and odd parity sums of the local
gate contribution cancel identically, and the residual
contributions from earlier steps also cancel mod $2\pi$, giving
$\Delta\Phi_j(i)=0$.
For $i=0$, the odd parity sum vanishes (no odd $m\le 0$), so only the even parity sum contributes to the local phase; likewise
$\Phi(0,j-1,0)=0$ for the same reason, so the residual reduces
to $\Phi(1,j-1,0)$, giving the cumulative expression in
Eq.~\eqref{eq:sequential-condition}.
Each sequence step halves the residual phase prefactor, so after step $j$ every nonzero Dicke sector carries a phase
$2^{l+1-j}\theta$, equal to $2\theta$ at $j=l$.

\subsection{Worked examples}
 
We verify Eq.~\eqref{eq:sequential-condition} explicitly for the four- and five-qubit malleable parent circuits discussed in
the main text, working at $l=3$ ($\theta=\pi/8$,
$2\pi/\theta=16$) with skip parameter $s=1$ unless otherwise stated.
The four-qubit parent ($n=4$, $\mathcal{W}_j=\{1,\ldots,n-j+1\}$)
produces a sequential chain $[[8,3,2]]\to[[12,2,2]]\to[[14,1,2]]$,
and the five-qubit parent ($n=5$, $s=2$, odd weights
$\mathcal{W}_j=\{w\le n-j+1:w\text{ odd}\}$) extends this to
$[[8,4,2]]\to[[12,3,2]]\to[[14,2,2]]\to[[15,1,3]]$.
 
\paragraph{Step $j=1$: $[[8,3,2]]$ and $[[8,4,2]]$.}
For $\mathcal{W}_1=\{1,2,3,4\}$ at $n=4$, the $i\ge 1$ condition gives
\begin{equation*}
  \Delta\Phi_1(i)/\theta = 2^{i+1}\!\sum_{w=1}^{4}\binom{3-i}{w-i}
  = 2^{i+1}\cdot 2^{3-i} = 16\equiv 0,
\end{equation*}
and at $i=0$: $2\sum_w\binom{3}{w-1}=2\cdot 8=16\equiv 0$.
This yields $[[8,3,2]]$ (\quirk{\qaaazthreetwo}{Quirk}) with
output $\ket{\CCZ}$ (App.~\ref{app:eight32}).
At $n=5$ with $s=2$, $\mathcal{W}_1=\{1,3,5\}$:
$2\sum_{w\in\mathcal{W}_1}\binom{4}{w-1}=2(1+6+1)=16\equiv 0$ and
$2^{i+1}\sum_{w\in\mathcal{W}_1}\binom{4-i}{w-i}\equiv 0$ for $i\ge 1$,
yielding $[[8,4,2]]$ (\quirk{\qaaazfourtwo}{Quirk}) with output
Clifford-equivalent to $\ket{\CCZ}$.
 
\paragraph{Step $j=2$: $[[12,2,2]]$ and $[[12,3,2]]$.}
At $n=4$, adding $\mathcal{W}_2=\{1,2,3\}$:
\begin{align*}
  \Delta\Phi_2(0)/\theta
  &= 2\!\sum_{w\in\mathcal{W}_2}\!\binom{2}{w-1}
   + 2\!\sum_{w\in\mathcal{W}_1}\!\binom{2}{w-1}\nonumber\\
  &= 2(1+2+1)+2(1+2+1+0)=16\equiv 0,
\end{align*}
yielding $[[12,2,2]]$ (\quirk{\qtwelvetwotwo}{Quirk}) with output
$\ket{\CS}$. At $n=5$ with $s=2$ ($\mathcal{W}_1=\{1,3,5\}$,
$\mathcal{W}_2=\{1,3\}$):
\begin{align*}
\Delta\Phi_2(0)/\theta&=2\sum_{w\in\mathcal{W}_2}\binom{3}{w-1}+2\sum_{w\in\mathcal{W}_1}\binom{3}{w-1}\nonumber\\
&=8+8=16\equiv 0,
\end{align*}
yielding $[[12,3,2]]$ (\quirk{\qtwelvethreetwo}{Quirk}) with output
Clifford-equivalent to $\ket{\CCZ}$.

\paragraph{Step $j=3$: $[[14,1,2]]$ and $[[14,2,2]]$.}
At $n=4$, adding $\mathcal{W}_3=\{1,2\}$:
\begin{align*}
  \Delta\Phi_3(0)/\theta
  &= 2\!\sum_{w\in\mathcal{W}_3}\!\binom{1}{w-1}
   + 2\!\sum_{w\in\mathcal{W}_2}\!\binom{1}{w-1}\nonumber\\&\quad
   + 2\!\sum_{w\in\mathcal{W}_1}\!\binom{2}{w-1}\nonumber\\&
  = 2(1+1)+2(1+1+0)+2(1+2+1+0)\nonumber\\&=16\equiv 0,
\end{align*}
yielding $[[14,1,2]]$ (\quirk{\qparentfour}{Quirk}) with output
$\ket{T}$. At $n=5$ with $s=2$, $\mathcal{W}_1=\{1,3,5\}$,
$\mathcal{W}_2=\mathcal{W}_3=\{1,3\}$:
\begin{align*}
  \Delta\Phi_3(0)/\theta
  &= 2\!\sum_{w\in\mathcal{W}_3}\!\binom{2}{w-1}
   + 2\!\sum_{w\in\mathcal{W}_2}\!\binom{2}{w-1}\nonumber\\&\quad
   + 2\!\sum_{w\in\mathcal{W}_1}\!\binom{3}{w-1}\nonumber\\
  &= 4+4+8=16\equiv 0,
\end{align*}
yielding $[[14,2,2]]$ (\quirk{\qfourteentwotwo}{Quirk}) with output
$\ket{T}$ at distance~$2$.
 
\paragraph{Step $j=4$: $[[15,1,3]]$.}
At $n=5$ with $s=2$, adding $\mathcal{W}_4=\{1\}$:
\begin{align*}
  \Delta\Phi_4(0)/\theta
  &= 2\!\sum_{w\in\mathcal{W}_4}\!\binom{1}{w-1}
   + 2\!\sum_{w\in\mathcal{W}_3}\!\binom{1}{w-1}\nonumber\\
  &\quad+ 2\!\sum_{w\in\mathcal{W}_2}\!\binom{2}{w-1}
   + 2\!\sum_{w\in\mathcal{W}_1}\!\binom{3}{w-1}\nonumber\\
  &= 2+2+4+8=16\equiv 0,
\end{align*}
yielding $[[15,1,3]]$ (\quirk{\qsummaryfifteen}{Quirk}) with output
$\ket{T}$ at distance~$3$, completing the five-qubit parent chain.
 
\paragraph{Catalytic conversions.}
Applying a Clifford+$T$ catalytic
$\ket{\CCZ} \to \ket{T}$ conversion to the $j = 1$ output of
the four-qubit parent yields $[[11, 1, 2]]$
(\quirk{\qelevenoneone}{Quirk}) when one output qubit is retained, and $[[10, 2, 2]]$ (\quirk{\qtentwotwo}{Quirk}) when two are retained. Both are verified for the teleported-$T$ case, accounting for $Z$- and $S$-type errors on the
conversion ancilla.

\subsection{Results and qualitative features}
\begin{figure}[h]
\centering
\includegraphics[width=0.48\textwidth]{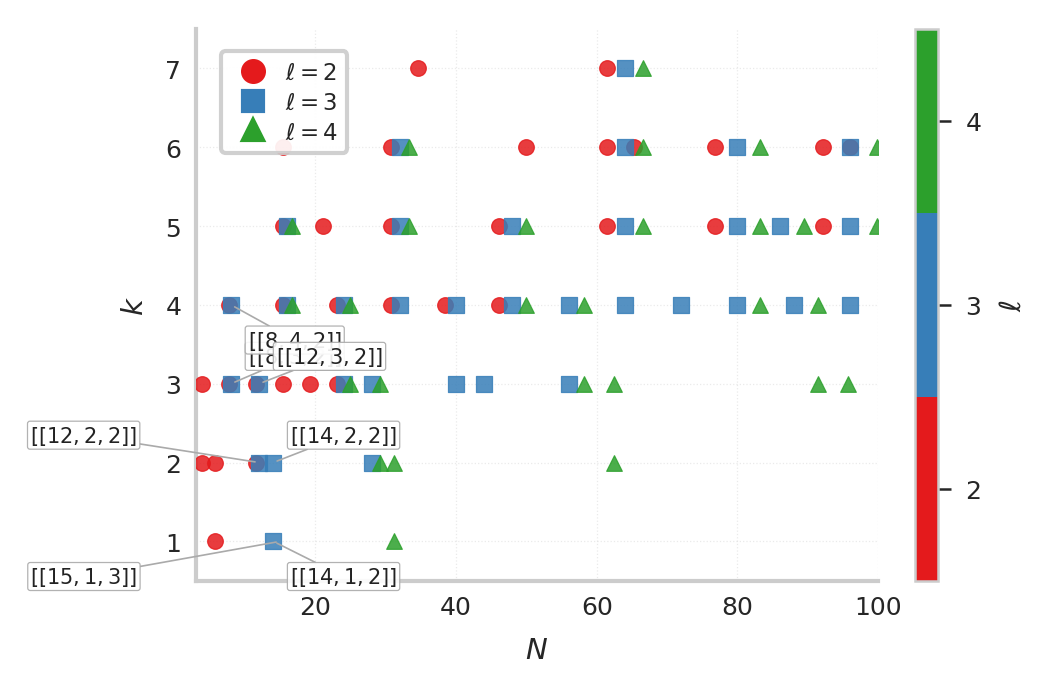}
\caption{%
\textbf{Sequential distillation factories.}
Each marker shows a factory at coordinate $(N, k)$; color and
shape both encode the Clifford-hierarchy level $l$
(red circles: $l=2$; blue squares: $l=3$; green triangles:
$l=4$). We have restricted the search to $j_\mathrm{max}=l$. Multiple markers at the same coordinate indicate
factories valid at multiple $l$ via independent solutions of
Theorem~\ref{thm:sequential-condition}. See
Table~\ref{tab:sequential_factories} for full enumeration of
sequential search output not found in the asymmetric search.%
}\label{fig:sequential_search}
\end{figure}%% ---------------

The sequential search results are shown in Fig.~\ref{fig:sequential_search}.
For each $(l, n)$ with $l \in \{2, 3, 4\}$ and $n \le 6$, we
enumerate sequences $(\mathcal{W}_1, \ldots, \mathcal{W}_J)$
with $J = n - k$ and each
$\mathcal{W}_j = \{w \equiv 1 \bmod s_j : w \le n - j + 1\}$
for $s_j \in \{1, 2, 3, 4\}$, giving $4^J$ candidate tuples
per $(n, J)$. Each candidate is verified in $O(J^2 n)$ time by
evaluating Eq.~\eqref{eq:sequential-condition} for all $i$ at
every step. The full search to $n = 6$ completes in well
under one second.

\paragraph{Non-monotonicity in $l$.}
Unlike the symmetric case, the sequential condition is
\emph{not} a monotone in $l$: the prefactor $2^{i + 1}\theta$
depends on $l$ through $\theta = \pi / 2^l$ in a way that does
not factor through $l \to l - 1$. A sequence
$(\mathcal{W}_1, \ldots, \mathcal{W}_J)$ that disentangles at one $l$ generically does not disentangle for all $l'<l$. Thus, solutions
at different $l$ are independent.

\paragraph{Sign redundancy.}
Each step $j$ carries a sign $\sigma_j \in \{+1, -1\}$.
Exhaustive search over both branches shows that sign flips relabel solutions without producing new code parameters; the sign tuple is suppressed in
Table~\ref{tab:sequential_factories}.

\begin{table}[t]
\centering
\caption{All non-trivial distance-$2$ sequential search outputs with
$l \in \{2, 3, 4\}$, $n \le 6$, $j_\mathrm{max}=l$ that are not found by the asymmetric search. For each $(l, N, k)$: smallest $n$
admitting a solution and lexicographically smallest per-step skip
tuple $\mathbf{s} = (s_1, \ldots, s_{n - k})$ satisfying
Theorem~\ref{thm:sequential-condition} with all $\sigma_j = +1$.}
\label{tab:sequential_factories}
\setlength{\tabcolsep}{5pt}
\renewcommand{\arraystretch}{1.0}
\centering
\begin{tabular}{r c c c l}
\toprule
$N$ & $k$ & $n$ & $\mathbf{s}$ & l\\
\midrule
16 & 3 & 5 & $(2, 1)$& 2\\
20 & 3 & 5 & $(1, 2)$& 2\\
40 & 4 & 6 & $(1, 2)$& 2\\
32 & 5 & 6 & $(1)$& 2\\
\midrule
40 & 3 & 6 & $(2, 1, 1)$& 3\\
44 & 3 & 6 & $(1, 2, 2)$& 3\\
40 & 4 & 6 & $(1, 2)$& 3\\
32 & 5 & 6 & $(1)$& 3\\
\midrule
32 & 5 & 6 & $(1)$& 4\\
\bottomrule
\end{tabular}
\end{table}

\section{Symmetry-free targeted search on extended borrowed identities for synthillation circuits}\label{app:symfree}
\begin{figure} [h]       
  \centering
  \includegraphics[width=0.5\textwidth]{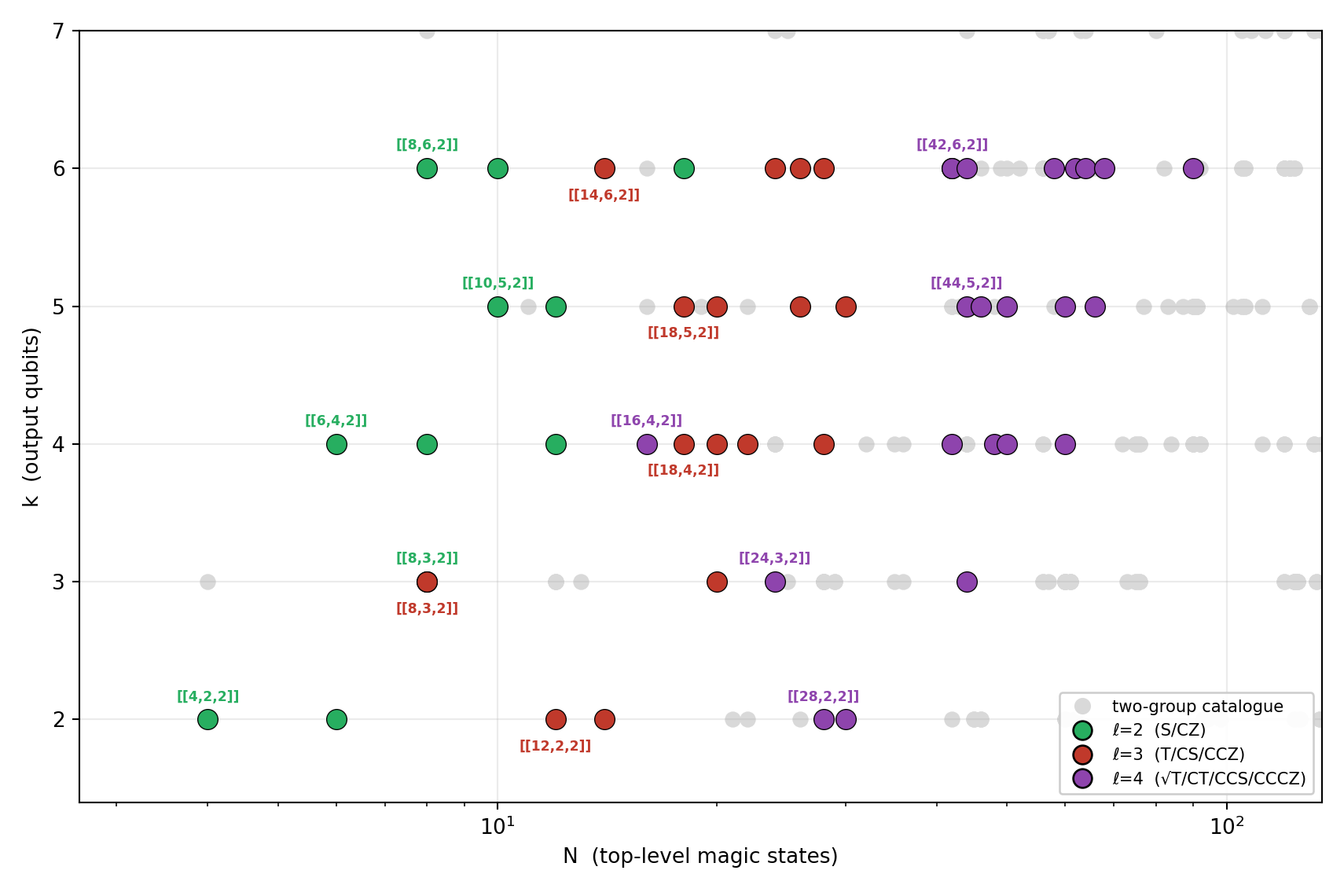}
  \caption{Factories from the symmetry-free targeted-output search (colored)
  over the two-group symmetric catalogue (gray), at $l=2,3,4$, in the
  $(N,k)$ plane ($N=$ top-level magic states, $k=$ output qubits). The
  minimal-$N$ factory at each $k$ is labeled by $[[N,k,d]]$ and its block
  composition.}
  \label{fig:symfree}
\end{figure}
We extend the borrowed identity formalism to include multi-weight phase rotations by angle $c\pi/2^i$ where $i\le l$ and $c\in\mathbb{Z}_{2^i}$.
The extended borrowed-identity condition~(\ref{def:borrowed}) is
\emph{linear} over $\mathbb{Z}_{2^{l+1}}$ in the gate coefficients $\{c\}$.
This makes a direct, symmetry-free search possible: instead of restricting the
gate set to a symmetric weight-class ansatz, we fix a \emph{target output} and
\emph{solve} for the rest of the circuit. For a chosen output structure---a
partition of the $k$ output qubits into magic blocks---the output gates are fixed and the borrowed-identity condition is
solved for the remaining check-coupling gates by $2$-adic Gaussian elimination,
taking the smallest check register that admits a solution.

Re-expressed at the common scale $\pi/2^{l}$, the coefficients are
unconstrained ($c\in\mathbb{Z}_{2^{l+1}}$), so the
solver freely places lower-level (Clifford) rotations. This extension of the borrowed identity is suited to synthillation circuits, which admit all diagonal gates of Clifford-hierarchy level $\mathbb{C}_l$ together with the lower levels $\mathbb{C}_{l'<l}\subset \mathbb{C}_l$~\cite{cui2017diagonal}. We report $N$ as the
number of top-level ($\pi/2^{l}$, odd-coefficient) magic gates---Clifford
gates being free---and compute the distance over these. The construction is
\emph{exact} (a linear solve, not heuristic enumeration); but the
final reduction of $N$ is greedy. Figure~\ref{fig:symfree} shows the
factories found by this symmetry-free search at $l=2,3,4$ that lie outside the two-group symmetric ansatz. They serve as an existence proof that synthillation circuits fall under the borrowed-identity framework with minimal changes.

This symmetry-free targeted-output search is fast but only complementary to the two-group symmetric search discussed in the main text. A symmetry-free SAT search based on our extended borrowed identity condition is inclusive of both and includes higher-distance factories, only excluding the non-CSS results like $[[10,2,2]]$. However, this search is too slow to achieve the final goal of finding the smallest $k>1$ factories at all distances. This SAT search is also general over all Clifford levels.

\section{Canonical examples}\label{app:annotated_quirk_circuits}
\begin{table}[h]
\centering
\caption{Canonical instances for borrowed-identity based distillation factories at level $l$ of the Clifford hierarchy. Clickable links to interactive Quirk
circuits~\cite{quirk} for some instances are provided.}
\label{tab:canonical}
\begin{ruledtabular}
\begin{tabular}{llll}
Approach (\S) & $l=2$&$l=3$&$l=4$\\
\hline
Symmetric (all weights,~\S\ref{sec:symmetric})    &$[[6,1,2]]$& \quirk{\qparentfour}{$[[14,1,2]]$}&$[[30,1,2]]$\\
Symmetric (odd weights,~\S\ref{sec:symmetric}) & $[[7,1,3]]$   & \quirk{\qsummaryfifteen}{$[[15,1,3]]$}&$[[31,1,3]]$\\
Two-group (\dv{deg}$=1$,~\S\ref{sec:asymmetric})  & $[[6,2,2]]$ & \quirk{\qfourteentwotwo}{$[[14,2,2]]$} & $[[30,2,2]]$\\
Two-group (\dv{deg}$=1$,~\S\ref{sec:asymmetric})  & $[[8,4,2]]$ & \quirk{\qtwentyfourtwo}{$[[20,4,2]]$} & $[[44,4,2]]$\\
Two-group (\dv{deg}$=1$,~\S\ref{sec:asymmetric})  & $[[10,6,2]]$ & $[[26,6,2]]$ & $[[58,6,2]]$\\
Two-group (\dv{deg}$=2$,~\S\ref{sec:asymmetric}) & $[[4,2,2]]$  & \quirk{\qtwelvetwotwo}{$[[12,2,2]]$} & $[[28,2,2]]$\\
Two-group (\dv{deg}$=3$,~\S\ref{sec:asymmetric}) &--  & \quirk{\qaaazthreetwo}{$[[8,3,2]]$}& $[[24,3,2]]$\\
Two-group (\dv{deg}$=3$,~\S\ref{sec:asymmetric}) &--   & \quirk{\qtwentyfoursixtwo}{$[[24,6,2]]$}& $[[403,6,2]]$\\
Two-group (\dv{deg}$=4$,~\S\ref{sec:asymmetric}) &--   & -- & $[[16,4,2]]$\\
Catalytic~(\S\ref{sec:malleability})   & --&\quirk{\qelevenoneone}{$[[11,1,2]]$} & --\\
Catalytic~(\S\ref{sec:malleability})   & --&\quirk{\qtentwotwo}{$[[10,2,2]]$} &-- \\
\end{tabular}
\end{ruledtabular}
\end{table}
\begin{figure*}[!h]
    \includegraphics[width=0.85\textwidth]{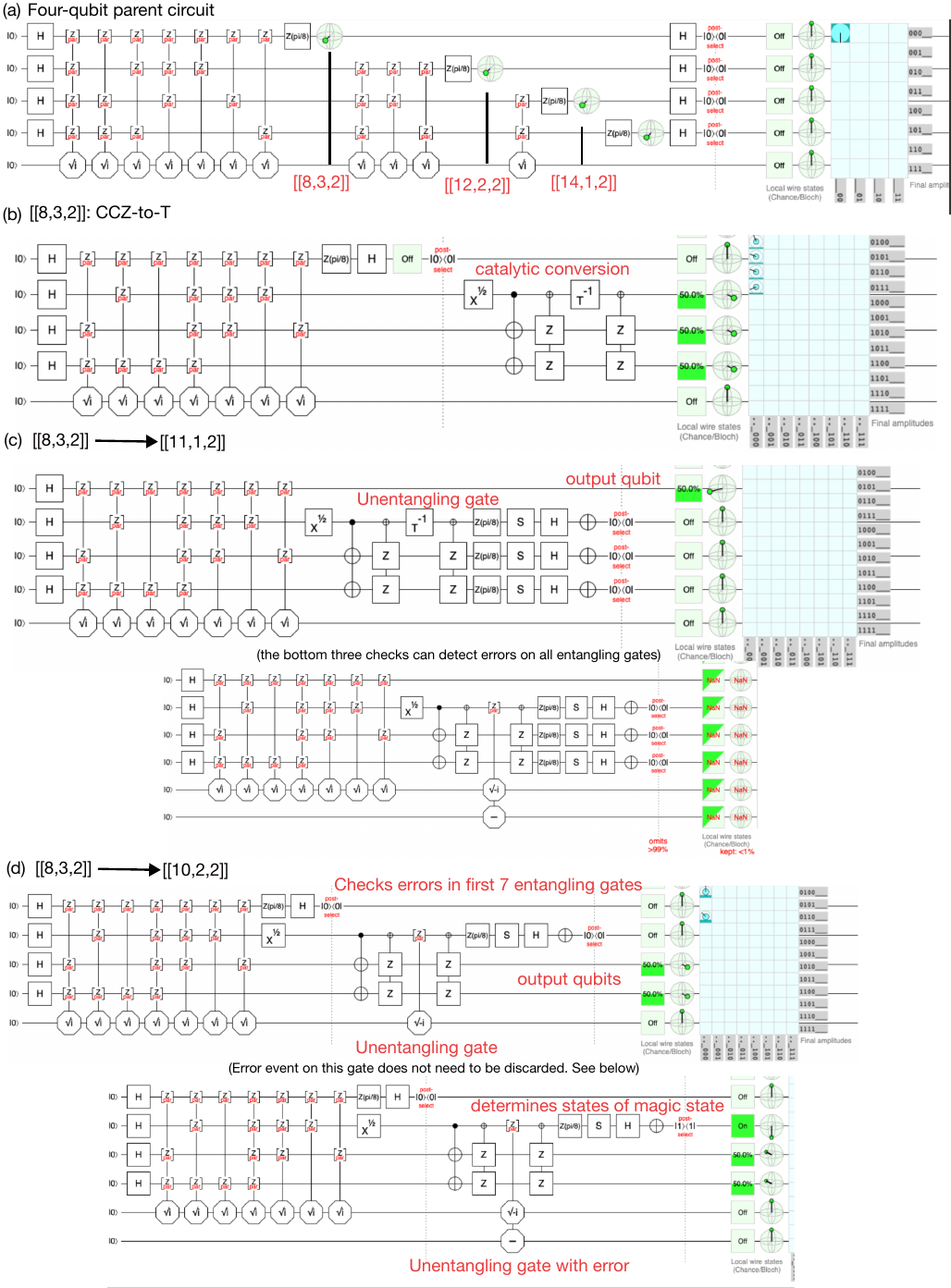}
    \caption{\textbf{The four-qubit parent circuit \quirk{\qparentfour}{[link]}.} The conditional $\sqrt{i}$ denotes multi-weight $\pi/8$ rotations~\cite{litinski2019magic,quirk}. (a) The malleable circuit showing $[[8,3,2]]\rightarrow [[12,2,2]]\rightarrow [[14,1,2]]$. Note that the $T$-to-CS factory $[[12,2,2]]$ (equivalent to the $4\,\CS\!\to\!1\,\CS$) is \emph{larger} than the
$T$-to-CCZ factory $[[8,3,2]]$, even though CS is a lower-level gate, because the number of gates removed from the identity equals the number of $T$ states the output consumes in its unitary synthesis---3 for CS versus 7 for CCZ. Another reason for this is that any pure $T$-to-$CS$ distance-$2$ factory which cannot be interpreted as $CS$-to-$CS$ circuit (analogous to $[[8,3,2]]$) requires a borrowed identity on $3$ qubits which is not possible for $l=3$ (see Sec.~\ref{sec:symmetric}). (b) The three-qubit output from the $[[8,3,2]]$ circuit obeys the catalytic CCZ-to-T conversion. (c) The catalytic conversion to the $[[11,1,2]]$ circuit (top qubit is the output) \quirk{\qelevenoneone}{[link]}. (d) The catalytic conversion to the $[[10,2,2]]$ circuit \quirk{\qtentwotwo}{[link]}, where the error on the unentangling gate does not need to be discarded. The canonical instance of each construction at every level $l$ is collected in Table~\ref{tab:canonical}.}
    \label{fig:four-qubit-circuit}
\end{figure*}
\begin{figure*}[htb]
    \includegraphics[]{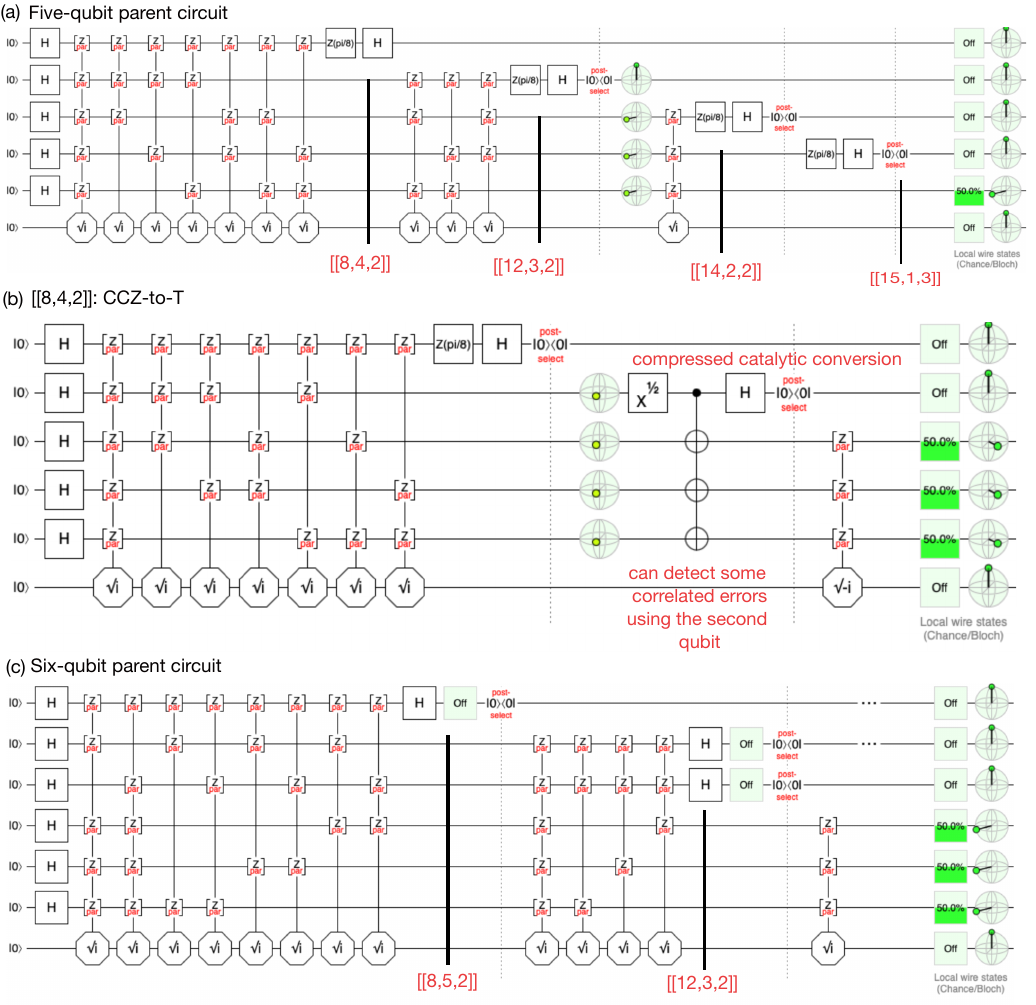}
    \caption{\textbf{The five-qubit \quirk{\qparentfive}{[link]} and six-qubit \quirk{\qparentsix}{[link]} parent circuits.} The conditional $\sqrt{i}$ denotes multi-weight $\pi/8$ rotations~\cite{litinski2019magic,quirk}. (a) The malleable circuit showing $[[8,4,2]]\rightarrow [[12,3,2]]\rightarrow [[14,2,2]]$. By reducing the output to one qubit, this becomes an example of $s=2$ (see Section~\ref{sec:symmetric}), yielding distance-3 factory $[[15,1,3]]$. Note that the $[[12,3,2]]$ factory output also requires only one `catalytic' $T$ rotation to un-entangle the outputs into $3$ T states. Therefore, the $[[12,3,2]]$ output is also CCZ-equivalent. (b) The four-qubit output from the $[[8,4,2]]$ is Clifford equivalent to a CCZ state, yielding $3$ T states via catalytic conversion. This four-qubit output can detect some correlated errors which flip the $X\otimes X\otimes X$ measurement. (c) The extension to a six-qubit parent factory yields a $5$-qubit output Clifford-equivalent to the CCZ state.}
    \label{fig:five-six-qubit-circuit}
\end{figure*}

%% -----------------------------------------------------------
%% -----------------------------------------------------------
\end{document}